\documentclass[10pt]{article}
\usepackage{hyperref}
\usepackage{amsmath,amsthm,amssymb,mathtools}
\usepackage{xcolor}
\usepackage{palatino}
\usepackage{microtype}
\usepackage{comment}
\usepackage{tabularx}
\usepackage{makecell}
\usepackage{multirow}
\usepackage{booktabs}
\usepackage[inkscapelatex=false]{svg}

\usepackage{multirow}
\usepackage{varwidth}
\usepackage{rotating}
\definecolor{darkgreen}{rgb}{0,0.5,0}
\makeatletter
\providecommand{\tabularnewline}{\\}
\newenvironment{cellvarwidth}[1][t]
{\begin{varwidth}[#1]{\linewidth}}
{\@finalstrut\@arstrutbox\end{varwidth}}

\usepackage{graphicx}

\usepackage[letterpaper,margin=1in]{geometry}

\renewenvironment{abstract}
{\quotation}
{\endquotation}

\date{}

\makeatletter
\renewcommand{\fnum@figure}{\textbf{Figure \thefigure}}
\renewcommand{\fnum@table}{\textbf{Table \thetable}}
\makeatother

\usepackage{url}

\newtheorem{theorem}{Theorem}

\newtheorem{proposition}{Proposition}

\newtheorem{remarks}{Remarks}

\renewcommand{\Re}{\mathbb{R}} 
 
\newcommand{\bbD}{\mathbb{D}}

\newcommand{\cA}{\mathcal{A}}

\newcommand{\cE}{\mathcal{E}}
 
\newcommand{\cH}{\mathcal{H}}

\newcommand{\cN}{\mathcal{N}}

\newcommand{\cT}{\mathcal{T}}

\newcommand{\vv}{\boldsymbol{v}}
\newcommand{\x}{\boldsymbol{x}}
\newcommand{\y}{\boldsymbol{y}}
\newcommand{\z}{\boldsymbol{z}}
\newcommand{\w}{\boldsymbol{w}}
\newcommand{\s}{\boldsymbol{s}}

\newcommand{\bv}{\boldsymbol{v}}
\newcommand{\h}{\boldsymbol{h}}
\newcommand{\p}{\boldsymbol{p}}
\newcommand{\q}{\boldsymbol{q}}
\newcommand{\A}{\mathbf{A}}
\newcommand{\X}{\mathbf{X}}
\newcommand{\I}{\mathbf{I}}
\newcommand{\J}{\mathbf{J}}
\newcommand{\W}{\mathbf{W}}

\newcommand{\0}{\mathbf{0}}

\renewcommand{\leq}{\leqslant}
\renewcommand{\geq}{\geqslant}

\def\scititle{A Continuous Energy Ising Machine Leveraging Difference-of-Convex Programming}
\title{\bfseries \boldmath \scititle}

\author{
Debraj~Banerjee$^{1\alpha}$,
Santanu~Mahapatra$^{2\beta}$,
Kunal~N.~Chaudhury$^{1\gamma}$\and
\small$^{1}$Department of Electrical Engineering, Indian Institute of Science, Bangalore 560012, India.\and
\small$^{2}$Department of Electronic Systems Engineering, Indian Institute of Science, Bangalore 560012, India.\and
\small Corresponding authors: 
$^\alpha$debrajb@iisc.ac.in, $^\beta$santanu@iisc.ac.in, $^\gamma$kunal@iisc.ac.in.
}

\begin{document} 

\maketitle

\begin{abstract} \bfseries \boldmath
Many combinatorial optimization problems can be reformulated as finding the ground state of the Ising model. Existing Ising solvers are mostly inspired by simulated annealing. Although annealing techniques offer scalability, they lack convergence guarantees and are sensitive to the cooling schedule. We propose solving the Ising problem by relaxing the binary spins to continuous variables and introducing an attraction potential that steers the solution toward binary spin configurations. A key property of this potential is that its combination with the Ising energy produces a Hamiltonian that can be written as a difference of convex polynomials. This enables us to design efficient iterative algorithms that require a single matrix-vector multiplication per iteration and provide convergence guarantees. We implement our Ising solver on a wide range of GPU platforms, from edge devices to high-performance computing clusters, and demonstrate that it consistently outperforms existing solvers across problem sizes ranging from small ($10^3$ spins) to ultra-large ($10^8$ spins).
\end{abstract}

\section{Introduction}

Combinatorial optimization problems arise in numerous fields, including social networks~\cite{yuan2018}, finance~\cite{roman2019}, cryptography~\cite{wang2020prime}, scheduling~\cite{rieffel2014}, electronic circuit design~\cite{barahona1988circuit}, and biosciences~\cite{kell2012,robert2021}. The challenge with such problems is that the number of possible solutions grows exponentially with the number of variables, rendering classical algorithms slow or impractical for large-scale problems. To address this, researchers are increasingly exploring unconventional computing paradigms. A particularly promising direction is reformulating the problem as finding the ground state of a physical system, such as the Ising model. Indeed, many Nondeterministic Polynomial-time (NP)-hard problems can be naturally expressed in the Ising form, while many others can be efficiently reduced to this framework~\cite{lucas2014}. Significant effort has been directed toward developing specialized hardware and algorithms, collectively known as Ising machines, that can efficiently compute the ground states of Ising models~\cite{tanahashi2019}. As industrial optimization problems grow in complexity, building large-scale Ising solvers has become more critical than ever. Beyond its role in combinatorial optimization, finding the ground state of the Ising model is a fundamental problem in physics, offering insights into phase transitions, magnetism, and critical behaviour in condensed matter systems.

The Ising model was originally introduced by Lenz and Ising as a mathematical model for ferromagnetism~\cite{ising1925}. In this model, the dimensionless energy function is given by
\begin{equation}
\label{eq:IsingModel}
\cE(s_1,\ldots,s_n) = -\frac{1}{2}\sum_{i,j=1}^n J_{ij}s_is_j - \sum_{i=1}^n h_is_i, 
\end{equation}
where $(s_1, \ldots, s_n)$ is the spin vector with each spin $s_i \in \{-1,+1\}$~(up/down spins), \(J_{ij}\) is the coupling coefficient between spins $s_i$ and $s_j$ with \(J_{ii}=0\) (no self-coupling), and $h_i$ is the external field acting on spin \(s_i\). Solving an Ising model refers to finding a spin assignment $s^*_1, \dots, s^*_n \in \{-1,1\}$ that minimizes the energy over all possible spin configurations, that is,
\begin{equation}
\label{eq:IsingProblem}
\cE(s^*_1,\ldots,s^*_n)= \underset{s_i \in \{-1,1\}}{\min}  \ \ \cE(s_1,\ldots,s_n).
\end{equation}
We refer to $(s^*_1, \dots, s^*_n)$ as the ground state spin vector (or simply the ground state) of the Ising model. 

We remark that~\eqref{eq:IsingProblem} can be reformulated using Boolean variables $\{0,1\}$ that arise naturally in combinatorial problems. Furthermore, by introducing an additional spin $s_{n+1}$, we can reformulate~\eqref{eq:IsingModel} as a homogeneous Ising problem (with no external field):
\begin{equation} 
\label{eq:IsingHom}
\underset{s_i  \in \{-1,1\}}{\min}   \ -\frac{1}{2}\sum_{i,j=1}^{n+1} \hat{J}_{ij} s_is_j,
\end{equation}
where the modified coupling $\{\hat{J}_{ij}\}$ is derived from $\{J_{ij}\}$ and $\{h_i\}$. Since the ground state of~\eqref{eq:IsingModel} can be inferred from that of~\eqref{eq:IsingHom}, it suffices to work with the homogeneous model~(Supplementary Note 1).

Solving large Ising models by brute force is computationally challenging, with no known polynomial-time algorithm. In fact, the problem is known to be NP-hard~\cite{barahona1982}. Consequently, rather than attempting to compute the exact ground state, practical approaches typically focus on finding good approximations to it. Over the years, many algorithmic heuristics and specialized hardware have been developed to address this challenge. We refer the reader to~\cite{mohseni2022, moy2022coupled, lo2023ising, maher2024cmos} for a survey of past and recent advances. 

The most widely used technique for solving Ising models is simulated annealing (SA)~\cite{kirkpatrick1983,isakov2015,goto2018boltzmann,bunyk2014}, a method inspired by quantum annealing~\cite{kadow1998,giues2002,das2008}. SA uses Monte Carlo sampling guided by the Boltzmann distribution to search for the ground state. However, it relies on a heuristic cooling schedule to explore the energy landscape and does not come with any mathematical guarantees. Moreover, a gradual temperature reduction is used, which can result in slow convergence~\cite{Yavorsky2019}, making SA less suited for solving large-scale problems. Variants such as noisy mean field annealing (NMFA)~\cite{King2018,Veszeli2021} and mean field annealing from a random state (MARS)~\cite{Yavorsky2019} aim to speed up convergence. However, they are typically prone to approximation errors and sensitive to hyperparameter settings. State-of-the-art techniques based on Hamiltonian dynamics, such as Simulated Bifurcation Machine (SBM)~\cite{puri2017,goto2019,goto2021,kanao2022,wang2023} and its ballistic variant (bSB)~\cite{goto2021}, offer parallelism but remain sensitive to parameter tuning and cooling schedules. On the other hand, Coherent Ising Machines (CIMs)~\cite{inagaki2016,tiunov2019,honjo2021}, which leverage optical hardware, suffer from issues like noise, decoherence, and the need for complex, precisely controlled hardware~\cite{yamamoto2017,pramanik2023}. More recently, neural network-based approaches have shown promise~\cite{zhang2018rbm,niazi2024,jiang2024}. However, they face inherent challenges, including limited scalability and the need for careful hyperparameter tuning~\cite{jiang2024}. Thus, there is a need to develop Ising solvers that can deliver both high-speed performance and reliable solution quality, particularly for ultra-large-scale problems.

\begin{figure}[t!]
   \centering
   \includegraphics[width=0.9\textwidth]{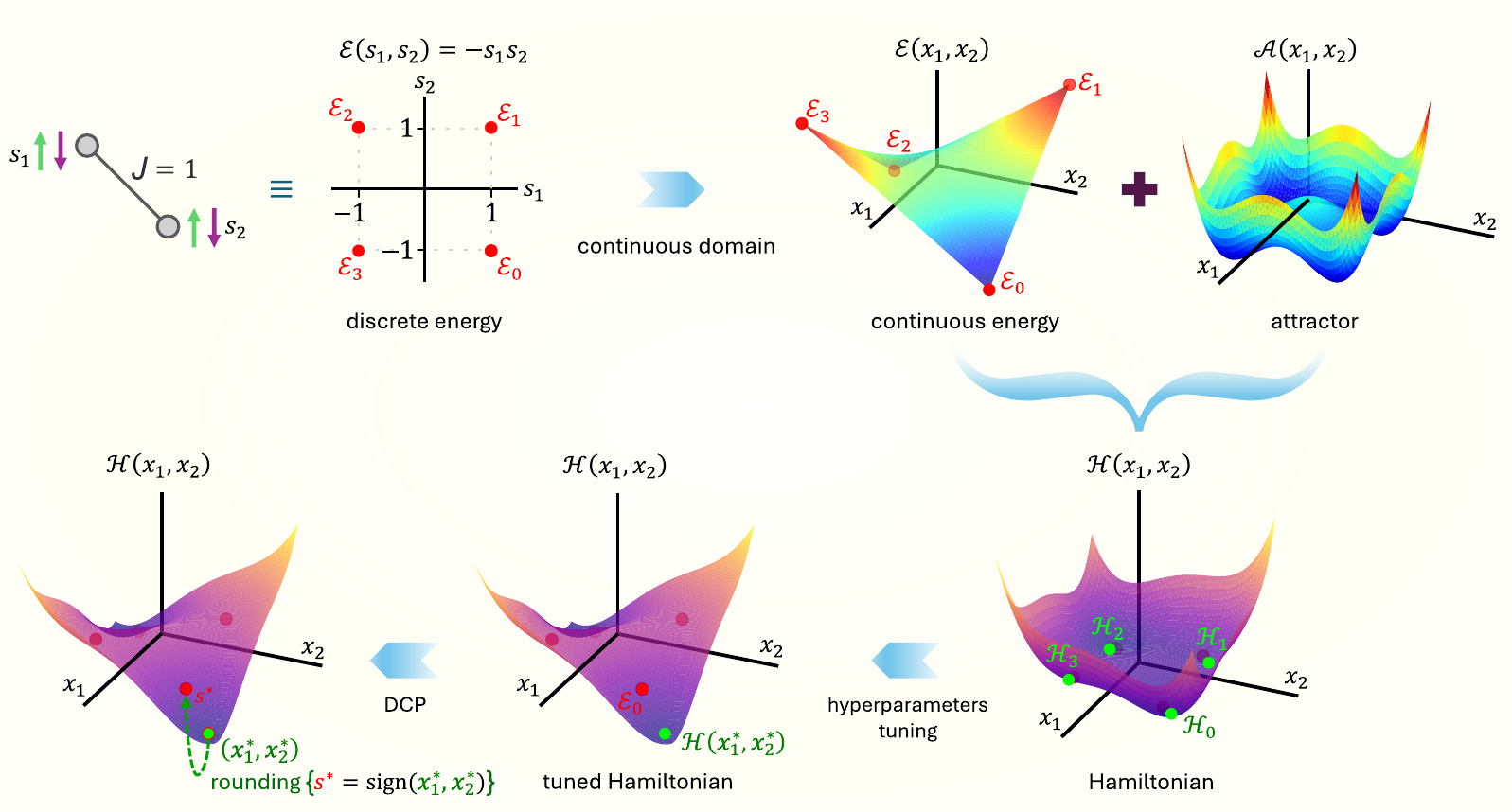} 
   \caption{\textbf{Overview of our Ising solver for a $2$-spin system}. We relax the $2^2 = 4$ discrete energy values \((\cE_0, \cE_1, \cE_2, \cE_3)\) to a smooth energy landscape \(\cE(x_1,x_2)\) defined over the continuous space $(x_1,x_2) \in \Re^2$.  This energy is combined with an attraction potential \(\cA(x_1,x_2)\) to form the Hamiltonian \(\cH  = \cE + \cA\). The local minima of this Hamiltonian, denoted by green dots \((\cH_0, \cH_1, \cH_2, \cH_3)\), serve as candidate ground states.  With suitable parameter tuning, the minimizer $(x^*_1, x^*_2)$ of the Hamiltonian produces a binary vector $\mathrm{sign}(x^*_1, x^*_2)$ that closely matches the true ground state \(\s^* = (s^*_1, s^*_2)\) of the original Ising energy \(\cE\).}
   \label{fig:dcp_ising_illus}
\end{figure}

In this work, we propose a continuous optimization-based Ising solver. Unlike the classical Goemans-Williamson Semidefinite Program (GW-SDP)~\cite{goemans1995}, our method avoids the computational bottleneck associated with large semidefinite programs~(Supplementary Note 5). To improve the scalability, we drop the binary spin constraints entirely and instead couple an attraction potential with the Ising energy, resulting in a Hamiltonian with soft spin constraints. A key insight is that this Hamiltonian can be expressed as a difference of convex polynomials, enabling the use of difference-of-convex programming (DCP)~\cite{abbaszadehpeivasti2023,phan2018}, which, to our knowledge, has not been applied to the Ising model. Specifically, we present two DCP-based iterative solvers that require a single matrix-vector multiplication per iteration and have just two tunable parameters. By exploiting the mathematical properties of our Hamiltonian, we prove that one of our solvers converges to a critical point, and under mild assumptions, to a local minimum, a guarantee that is difficult to obtain in general nonconvex optimization. Our solvers are well-suited for parallel execution on advanced graphical processing units (GPUs) and run efficiently across a wide range of platforms, from low-power edge devices to high-performance computing clusters. We demonstrate that they consistently outperform existing Ising solvers across a broad range of problem sizes, including ultra-large-scale instances with up to $10^8$ spins. Notably, we can solve a fully connected $10^7$-spin Ising model with nearly $50$ trillion coupling terms in just $11$ hours using four NVIDIA H100 GPUs. A quick comparison with existing Ising solvers is provided in Table~\ref{table:ising_solvers_comp_table}.

\section{Results}

\subsection*{Continuous-Energy Model}

We develop a continuous-energy formulation by embedding the binary spin vector in $\Re^n$ (Figure~\ref{fig:dcp_ising_illus}). As a first step, we encode the coupling coefficients in an $(n \times n)$ symmetric matrix $\J$ with $J_{ii}=0$, and express the homogeneous Ising problem as
\begin{equation} 
\label{eq:IsingEergy}
\underset{\s \in \{-1,1\}^n } {\min} \ \cE(\s)= - \frac{1}{2}\s^\top \J \s.
\end{equation}
We can identify the optimization domain $\{-1,1\}^n$ with the vertices of the unit hypercube in $\Re^n$. We then leverage the quadratic structure of $\cE$ to expand the domain to $\bbD = \big\{ \{-\lambda,\lambda\}^n: \ \lambda > 0 \big\}$, and consider the optimization problem
\begin{equation} 
\label{eq:IsingD}
\underset{\x \in \bbD} {\min} \ \  \cE(\x),
\end{equation}
where $\x = (x_1, \ldots, x_n)$ is the new optimization variable. The domain $\bbD$ consists of radial lines emanating from the origin along the vertices of $\{-1,1\}^n$. Though it is larger than the domain in~\eqref{eq:IsingEergy}, we can compute the ground state of~\eqref{eq:IsingEergy} from the global minimizers of~\eqref{eq:IsingD}. Namely, if $\x^*$ is a global minimizer of~\eqref{eq:IsingD}, then $\s^*=\mathrm{sign}(\x^*)$ is a ground state of~\eqref{eq:IsingEergy}. 

The formulation~\eqref{eq:IsingD} takes us a step closer to the continuum. However, from an optimization perspective, the difficulty is that the domain $\bbD$ is a non-convex set that is not even connected. A natural idea is to consider its convex hull. However, the convex hull is the entire space $\Re^n$, making the relaxation too loose and, therefore, ineffective for optimization. To address this, we introduce the attraction function that biases the minimizers towards $\bbD$. Specifically, for any fixed $\alpha, \beta > 0$, we consider the attractor
\begin{equation}
\label{eq:attractor}
    \cA(\x) =  \frac{\beta}{4} \big(x_1^4 + \dots + x_n^4 \big) - \frac{\alpha}{2} \big(x_1^2+ \dots + x_n^2 \big).
\end{equation}
The parameters \(\alpha\) and \(\beta\) are used to control the shape of the attractor (Figure~\ref{fig:attractor_plot_2d_3d}). { From this point onward, we use $\x \in \Re^n$ to denote the continuous relaxation variable.}

\begin{figure}[t!]
   \centering
   \includegraphics[width=0.3\textwidth]{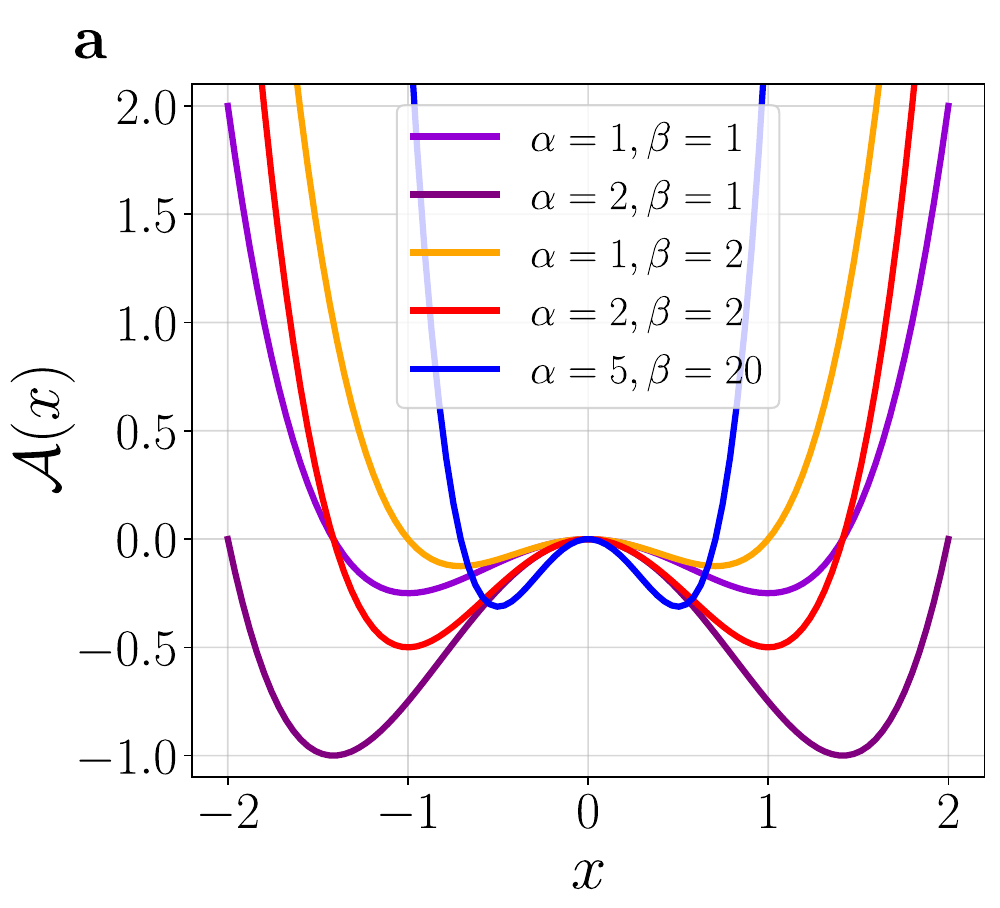} 
   \includegraphics[width=0.4\textwidth]{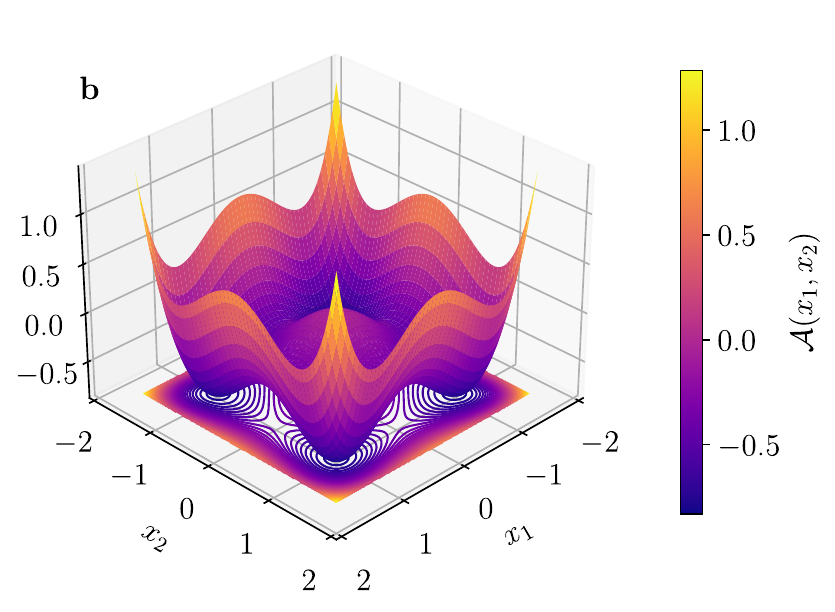}
   \caption{\textbf{Shape of the attractor.} \textbf{(a)} Plot of the one-variable attractor for different $\alpha, \beta$ values. The global minimizers at $\pm (\alpha/\beta)^{1/2}$ are shown. \textbf{(b)} Plot of the attractor in two variables for $\alpha = \beta = 1$, showing the four global minimizers at $(1, 1), (1, -1), (-1, 1)$ and $(-1, -1)$.}
   \label{fig:attractor_plot_2d_3d}
\end{figure}

This specific form of the attractor is motivated by the Hamiltonian of nonlinear parametric oscillators~\cite{goto2018} (Supplementary Note 4). Although this provides a strong physical motivation, our primary interest in~\eqref{eq:attractor} is because of its mathematical properties. In particular, we can show that the global minimizers of $\cA$ are exactly the vertices of the hypercube $ \{-\lambda,\lambda\}^n$, where $\lambda = (\alpha/\beta)^{1/2}$ (Supplementary Note 2). As this hypercube sits inside $\bbD$ for any choice of $\lambda$, we can use this property to bias the minimizers of the Hamiltonian $\cH =  \cA + \cE$ toward $\bbD$. Specifically, we consider the optimization problem
\begin{equation}
\label{eq:proposedH}
    \underset{\x \in \Re^n}{\min} \  \cH(\x) = \cA(\x) +  \cE(\x).
\end{equation}

{In summary, we begin with the exact Ising model in \eqref{eq:IsingEergy} with $\s \in \{-1,1\}^n$, and introduce the relaxed variable $\x \in \bbD$ in \eqref{eq:IsingD} to obtain an equivalent formulation. We then enlarge the domain to $\x \in \Re^n$ and incorporate the attractor in \eqref{eq:attractor}, which leads to the unconstrained energy formulation in \eqref{eq:proposedH}.}

Unlike the original Ising problem~\eqref{eq:IsingEergy}, which has a finite search space, a difficulty with continuous optimization problems is that they may not have a minimizer, i.e., the problem may not be well-posed. Nevertheless, we can prove that $\cH$ is bounded below and always has a global minimizer $\x^*$ (\textbf{Theorem S1} in Supplementary Note 2). 

Moreover, if it so happens that $\x^* \in \bbD$, then $\mathrm{sign}(\x^*)$ is guaranteed to be the ground state of the original Ising problem (Supplementary Note 2). {In practice, the recovered spin configuration $\mathrm{sign}(\x^*)$ may not always coincide with a true ground state. However, by appropriately tuning the parameters $\alpha$ and $\beta$, one can steer $\x^*$ toward $\bbD$. We empirically demonstrate later that this typically leads to solutions of higher quality than those obtained using existing Ising solvers.}

\subsection*{Optimization Algorithm}

A direct way of solving the optimization problem~\eqref{eq:proposedH} is to look for a critical point of $\cH$ by setting its gradient to zero. This leads to a system of cubic equations, which is difficult to solve in practice. Iterative schemes such as gradient descent provide a more feasible alternative, but they require evaluating the gradient $\nabla \cH$ at every step and rely on careful step-size tuning.

We propose an iterative algorithm that has a significantly lower per-iteration cost than gradient descent and comes with a formal convergence guarantee. This is based on the observation that we can write the Hamiltonian $\cH$ as the difference of convex functions for proper settings of $\alpha$ and $\beta$. More precisely, combining the quadratic components of $\cA$ and $\cE$,  we have
\begin{align}\label{eq:Hamilton}
    \cH(\x) &= \frac{\beta}{4} (x_1^4+ \dots + x_n^4) - \frac{\alpha}{2} (x_1^2+ \dots + x_n^2) - \frac{1}{2}\x^\top \J\x \nonumber
   \displaystyle \\ &=\underbrace{\frac{\beta}{4} (x_1^4+ \dots + x_n^4)}_{f(\x)} \,-\, \underbrace{\frac{1}{2}\x^\top (\J+\alpha \I)\x}_{g(\x)}.
\end{align}

The function $f$ is convex for any \(\beta > 0\). On the other hand, $g$ is convex if and only if  $\lambda_{\min}(\J+\alpha \I) \geqslant 0$, where $\lambda_{\min}$ is the smallest eigenvalue. This can be guaranteed by making $\alpha$ sufficiently large (Supplementary Note 2). Under these conditions, the optimization problem~\eqref{eq:proposedH} becomes 
\begin{equation}
\label{eq:DCP}
\min_{\x\in \Re^n} \  f(\x) - g(\x),
\end{equation}
where $f$ and $g$ are convex functions. This falls under the purview of difference-of-convex programming (DCP)~\cite{abbaszadehpeivasti2023,phan2018}. We consider a simple DCP algorithm called DCA~\cite{an2005}, along with its accelerated variant~\cite{phan2018}. 

DCA builds upon the classical principle of bound optimization~\cite{dempster1977maximum}. Starting with an initialization $\x^{(0)}$, DCA generates a sequence of estimates $\x^{(1)}, \x^{(2)}, \cdots$ that is expected to converge to a minimizer of $\cH$. The core idea is to exploit the convexity of $g$ to construct a global convex upper bound on $f-g$ around the current estimate $\x^{(k)}$. This surrogate is then minimized to obtain the next iterate $\x^{(k+1)}$. Applied to problem~\eqref{eq:DCP}, we obtain a simple update rule:
\begin{equation*}
\x^{(k+1)} = \cT(\x^{(k)}),
\end{equation*}
where $\cT$ is a linear transform followed by a pointwise nonlinearity (see Methods). In other words, the algorithm reduces to the repeated application of the operator $\cT$, referred to as a fixed-point algorithm. The cost per iteration is just a matrix-vector multiplication. A simple example illustrating the behaviour of this iterative algorithm is shown in Figure \ref{fig:2d_ising_dca}. We use a two-spin Ising problem to visualize the update steps and the convergence of the iterates. 

\begin{figure}[t!]
   \centering
   \includegraphics[width=1.0\textwidth]{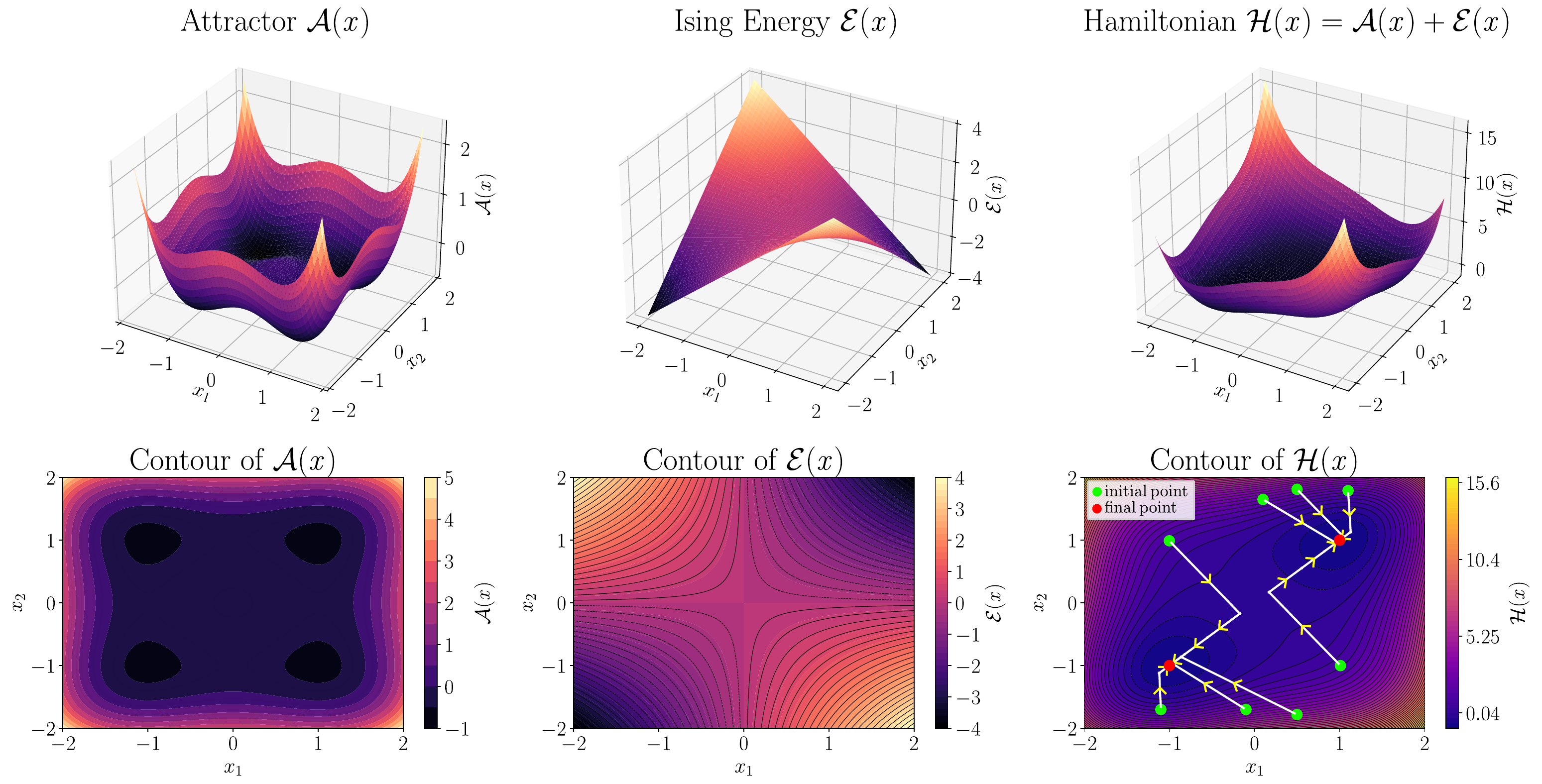}
   \caption{\textbf{Optimization trajectory for a $2$-spin system.} We consider the anti-ferromagnetic model $J_{12} = J_{21} = -1$, whose ground states are $(1,-1)$ and $(-1,1)$. \textbf{Top row}: Surface plots of the attractor $\cA$, Ising energy $\cE$, and the Hamiltonian $\cH=\cA+\cE$. \textbf{Bottom row}: Corresponding contour plots for $\alpha = 1$ and $\beta = 2$. In the bottom-right plot, the white curves represent the optimization trajectories of our Ising solver, starting from various initial points (indicated by green dots). Depending on the initialization, the iterates converge to one of the two degenerate ground states (red dots) after five iterations.}
   \label{fig:2d_ising_dca}
\end{figure}

An important aspect of iterative algorithms is their convergence behaviour. By exploiting the mathematical properties of $\cH$ (coercivity and real-analyticity) and the monotonicity guarantee for bound optimization ($\cH$ is guaranteed to decrease or remain the same at each iteration), we can prove that the iterates $\{\x^{(k)}\}$ converge to a critical point $\x^*$ of $\cH$ (\textbf{Theorem S2} in Supplementary Note 3). This is the strongest guarantee that can generally be achieved for nonconvex problems. We take $\mathrm{sign}(\x^*)$ to be the solution of the Ising problem~\eqref{eq:IsingEergy}.

We also consider a variant of DCA that achieves improved convergence rates using Nesterov acceleration~\cite{Nesterov2018}. We refer to this as Accelerated DCA (ADCA), which is described in detail in the Methods section. The solution quality achieved by ADCA consistently outperforms DCA as the problem size increases. Although the ADCA iterations consistently converge in practice, establishing a rigorous convergence guarantee is a challenging task. 

Our Ising solver has two parameters $\alpha$ and $\beta$. We propose a general parameter setting based on the coupling matrix $\J$ that is easy to configure and yields high-quality solutions. In particular, we recommend the following settings:
\begin{equation}
\label{eq:alphabeta}
\alpha \geqslant   \eta \, \lambda_{\max}(-\J) \quad \mbox{and}  \quad \beta = n\sqrt{n} \ \max_ {1 \leqslant j \leqslant n} \Big( \alpha  + \sum_{i\neq j} |  J_{ij}| \Big),
\end{equation}
where $\eta$ is a tuning parameter. It can be shown that the smallest eigenvalue of $\J$ is negative, so that $\lambda_{\max}(-\J) > 0$ in~\eqref{eq:alphabeta}. The above choice of $\alpha$ guarantees that $\J + \alpha \I$ is positive semidefinite, ensuring the convexity of $g$ in~\eqref{eq:DCP}. On the other hand, the choice of $\beta$ in~\eqref{eq:alphabeta} ensures boundedness of the DCA iterates for any \(\alpha\) (Supplementary Note 3). We have found that this particular setting consistently yields good results.

For large models (\(n \geqslant 10^4\)), computing \(\lambda_{\max}(\J)\) can be computationally demanding. To address this, we approximate \(\lambda_{\max}(\J)\) using Wigner’s semicircle law~\cite{mehta2004rmt}, following the approach in~\cite{goto2019,goto2021} for estimating eigenvalues of large matrices. Specifically, we estimate \(\lambda_{\max}(\J)\) with \(2\langle \J \rangle \sqrt{n}\), where \(\langle \J \rangle\) denotes the sample variance of the entries of \(\J\):
\begin{equation}\label{eq:estimate_eig_J}
    \langle \J \rangle^2 = \frac{1}{n(n-1)} \sum_{i\neq j}^n (J_{ij}-\bar{J})^2, \qquad 
    \bar{J} = \frac{1}{n(n-1)} \sum_{i\neq j}^n J_{ij}
\end{equation}
We tune the parameter $\eta$ in~\eqref{eq:alphabeta} by running a small number of iterations of our Ising solver, observing its behaviour, and adjusting $\alpha$ accordingly (Supplementary Note 6). This straightforward tuning procedure, combined with the low per-iteration cost, makes our DCA-based solvers particularly well-suited for ultra-large-scale Ising problems.

Our overall approach is summarized in Figure~\ref{fig:dcp_ising_illus}. We begin by relaxing the binary spin constraints, incorporating the attractor with the Ising energy, and expressing the resulting Hamiltonian in a DCP form. We then tune the parameters $\alpha$ and $\beta$, and apply DCA to optimize the tuned Hamiltonian. We refer to this as DOCH (Difference-Of-Convex-Hamiltonian), and its accelerated variant as ADOCH.

\subsection*{Benchmarking}

We have implemented and tested our algorithms across a range of NVIDIA GPUs, including a 4 GB 128-core Maxwell processor (Jetson Nano), 4 GB RTX 3050 (laptop), 24 GB RTX 3090, 32 GB V100, and 80 GB H100. A battery-powered 60,000 mAh edge-computing setup based on the Jetson Nano module is shown in Supplementary Figure~12. We have tested the effectiveness of our algorithms on the MAX-CUT problem, where the task is to partition the vertices of a graph into two disjoint subsets that maximize the number of edges between them~\cite{goemans1995}. This NP-hard problem is commonly used as a benchmark for evaluating Ising solvers. We have also implemented and compared our method with several state-of-the-art algorithms such as Simulated Annealing (SA)~\cite{isakov2015, inagaki2016}, ballistic Bifurcation Machine (bSB)~\cite{goto2021}, Simulated Coherent Ising Machine (SimCIM)~\cite{goto2021, inagaki2016}, and Spring Ising Algorithm (SIA)~\cite{jiang2024}. {To avoid being trapped in local minima, we have used $100$ different initializations while running each solver on small to medium-scale Ising models}. For completeness, detailed descriptions of these algorithms, along with the procedure for constructing ultra-large-scale Ising models, are provided in Supplementary Notes 7 to 9.

\subsubsection*{Small models}
We benchmark various Ising solvers on the Sherrington-Kirkpatrick (SK) model~\cite{sherrington1975}, a fully connected Ising model where the symmetric coupling coefficients $J_{ij}=J_{ji} \,( i \neq j)$ are sampled from the standard normal distribution~$\mathcal{N}(0,1)$. Each solver is tasked with minimizing the Ising energy of a system with $10^3$ spins. Their performance is evaluated based on the time required to reach the energy level obtained by the GW-SDP (using $100$ random rounding projections). As the solutions in the SK model typically converge quickly, we restrict the runtime of each solver to under {$0.1$ s (excluding the time needed to load data and compute the Ising energy)}. All experiments are repeated for $100$ random initializations. As shown in Figure~\ref{fig:1k_sk}a, {DOCH and ADOCH reach the GW-SDP energy threshold in about  $2.5$ and $3.5$ ms respectively}. Moreover, the histograms in Figures~\ref{fig:1k_sk}b and \ref{fig:1k_sk}c show that ADOCH consistently finds the lowest energy states with the highest frequency. Overall, our algorithms demonstrate the most robust and optimal performance at steady state compared to the other solvers.

Next, we evaluate our algorithm on the standard $K_{2000}$ benchmark, a fully connected graph with $2000$ nodes and nearly $2$ million edges that is commonly used for MAX-CUT problems. The MAX-CUT problem can be formulated as an Ising model, where the coupling matrix is given by \( \J = -(1/2)\W \), with \( \W \) denoting the weighted adjacency matrix of the graph (Supplementary Note 5). For our experiments, the weights $\{W_{ij}\}$ are sampled independently from $\{-1, 1\}$ with equal probability. Since the ground state is not available, we benchmark solvers based on the lowest Ising energy \(\cE\) achieved within a fixed runtime of {$1$ s (excluding data loading and energy calculation)} and evaluate their consistency across $100$ random initializations. We also compare the time each solver takes to reach the energy level obtained by GW-SDP. As shown in Figure~\ref{fig:2k_pm1}a, the DOCH and ADOCH solvers outperform the best-performing spring Ising algorithm (SIA), reaching the GW-SDP benchmark in approximately {$1.9$ and $2.6$ ms respectively.} {DOCH on average performs the best within the one-second window.} Figures~\ref{fig:2k_pm1}b and \ref{fig:2k_pm1}c further demonstrate that both DOCH and ADOCH consistently yield superior solutions across multiple runs with different initializations.


{We solve MAX-CUT on the (random) $G_{10}$ graph, an $800$-node, $94.01\%$ sparse graph with $\pm 1$ edge weights from the G-set dataset. We compute the mean cut values over $100$ random initializations as shown in Figure~\ref{fig:g10_M4_M7_M8}a. Each solver is run for $1000$ iterations. DOCH and ADOCH take only about $100$ iterations to reach the best solution. In particular, ADOCH attains GW-SDP level performance after just $3$ iterations on average, closely followed by DOCH and SIA. 

Results on additional G-set graphs are provided in Supplementary Figure~12. There, our solvers maintain superior performance with consistently smaller error bars even after $100$ iterations, especially on random $(G_{6}, G_{7}, G_{8})$ and toroidal $(G_{11}, G_{12}, G_{13})$ graphs. For longer runs ($\geq 100$ iterations), bSB and SimCIM gradually improve and beat our solvers for planner graphs $G_{18}, G_{19}, G_{20}, G_{21}$. However, such long runs are computationally prohibitive for ultra-large models. For example, on $10^7$-spin fully connected graphs, other solvers will take about $1000$ hours to run for $100$ iterations, including hyperparameter tuning (see Ultra-large models section).}

We also evaluate the solvers on random graphs with integer-valued edge weights from the MAX-CUT benchmark suite in the Biq Mac Library~\cite{bigq_mac}. Each solver is run for $1000$ iterations and over $100$ different initializations. The best cut values are shown as a bar graph in Figure~\ref{fig:biq-mac}b. We used the cut value obtained after $10^5$ iterations of SA as the largest cut value. For each run, we measure the time taken to reach $99\%$ of that largest cut value and average this across multiple runs to compute the average time-to-solution (average TTS). As shown in Figure~\ref{fig:biq-mac}a, the DOCH, ADOCH solvers consistently achieve the lowest Avg TTS across all graph instances, demonstrating superior convergence speed. Figure~\ref{fig:biq-mac}b further shows that DOCH and ADOCH attain the best or comparable cut values with respect to the other solvers. 


{The experiments on these small-scale Ising problems are conducted using a Jetson Nano module, except for the $K_{2000}$, the 1000-spin SK model, and the G-set graphs shown in Supplementary Figure~13. The computations are performed on a NVIDIA RTX 3050 GPU laptop with 4 GB RAM.}

\subsubsection*{Medium-to-large models} 
For Ising models of this scale, GW-SDP becomes computationally impractical to run. Moreover, for models with $>10^4$ spins, we require more powerful GPUs such as RTX 3090 and V100. In Figure~\ref{fig:g10_M4_M7_M8}b, we illustrate the performance of various Ising solvers on the $10^4$ spin SK model. Each algorithm is run for $1$ s with the same initialization until saturation. The plot in Figure~\ref{fig:g10_M4_M7_M8}b shows the time evolution of the Ising energy for each solver. We observe that {ADOCH outperforms other solvers, and DOCH achieves an energy below $-5 \times 10^5$ in under $0.1$ s}. Performance comparisons for larger Ising models with $10^5$ and $10^6$ spins are provided in (Supplementary Figures 4-7).

We also compare our Ising solver with the recently introduced free energy machine (FEM)~\cite{fem2025}, which is based on the principle of free-energy minimization. FEM shares some similarities with our approach, notably in its use of continuous relaxation and the incorporation of an external entropy term in the energy function. However, FEM inherits limitations common to SA methods, including sensitivity to hyperparameters, reliance on a carefully designed cooling schedule, and challenges in parameter tuning. Moreover, FEM employs optimizers like RMSprop and ADAM~\cite{2015-kingma} which do not offer convergence guarantees. Additionally, gradient computation becomes a significant bottleneck in FEM, complicating its implementation for large-scale problems ($n > 10^4$). As shown in (Supplementary Figures 1 and 2), our DOCH solvers consistently outperform FEM within $0.1$ and $1$ s for small ($10^3$ spin) and medium ($10^4$ spin) SK models respectively. It is also important to acknowledge that FEM is a more generalized formalism applicable to both the MAX-CUT and multi-spin Ising problems.

\subsubsection*{Ultra-large models} 
We benchmark the Ising solvers on ultra-large-scale Ising model problems involving $10^7-10^8$ spins, covering both sparse and dense connectivity regimes. Figure~\ref{fig:g10_M4_M7_M8}c shows the performance on a fully connected Ising model with $10^7$ spins with nearly $50$ trillion couplings. The coupling coefficients are generated using the pseudo-random function: $J_{ij}=\sin(i j + \mathrm{seed})$ with $\mathrm{seed} = 100$. The experiments are performed using four NVIDIA H100 GPUs. The runtime versus Ising energy plot in Figure~\ref{fig:g10_M4_M7_M8}c clearly demonstrates the superior performance of our methods. In particular, the ADOCH achieves the lowest energy level, followed closely by DOCH. These results underscore the scalability and efficiency of our solvers in tackling both sparse and fully connected Ising models at unprecedented scales. However the execution of Ising solvers on fully connected graphs demands much more computational resources than sparsely connected graphs.

{With a fixed runtime budget, state-of-the-art solvers must complete a full run for every hyperparameter choice to achieve the best Ising energy, while our solvers only need to inspect the first few iterations, allowing much faster parameter tuning. For instance, in the $10^7$-spin ultra-large Ising model from Figure~\ref{fig:g10_M4_M7_M8}c, tuning 10 hyperparameter values requires about $100$ hours for other solvers but only $20$ hours for ours. Subsequently, running each solver for $11$ hours showed that DOCH surpassed all others within $8$ hours, while ADOCH achieved its best results within $11$ hours.}

We further demonstrate the performance of the Ising solvers on a $10^8$-spin and $0.00001\%$ connected Ising model, having approximately 0.5 billion nonzero coupling coefficients. The nonzero $J_{ij}$ are sampled uniformly from the set of 9-bit signed integers. The matrix computations are performed using two H100 GPUs with parallelized matrix-vector multiplication (Supplementary Note 10). Each algorithm is run for $10^3$ s ($\approx 17$ minutes), during which the Ising energy reached a saturation point. The resulting runtime versus Ising energy plots are shown in Figure~\ref{fig:g10_M4_M7_M8}d. Both DOCH and ADOCH demonstrate superior performance, achieving the lowest energy values among all solvers. Additional results on similarly ultra-large-scale Ising models are provided in (Supplementary Figures 10 and 11). To our knowledge, no prior work has reported Ising solvers at this scale.

\section{Discussion}

We have introduced our Ising solver by relaxing the binary spins into continuous variables and formulating the resulting optimization as a difference-of-convex program. This approach is motivated by the success of first-order optimization methods for large-scale continuous optimization, particularly in modern deep learning~\cite{2015-kingma}. Central to our method is a tunable attractor function, which is coupled with the Ising energy to guide solutions toward binary spin assignments. By exploiting the structure of this attractor, we design two simple yet effective iterative algorithms (DOCH and ADOCH), which require only a single matrix-vector multiplication per iteration. We further establish that DOCH is guaranteed to converge to a critical point of the Hamiltonian. In contrast to conventional Ising solvers, our approach avoids reliance on cooling schedules or extensive hyperparameter tuning, making it especially well-suited for ultra-large-scale Ising problems. {Our methods require fewer iterations to reach high-quality Ising energies, which leads to a substantially lower overall runtime, even though the per-iteration cost is comparable to competing solvers (see \textbf{Remarks S1, S2} in the Supplementary).}


The key observations from our benchmarking are discussed below:

\begin{itemize}
    \item We demonstrate scalability by solving a fully connected Ising model with $10^7$ spins and $\sim 50$ trillion couplings (Figure~\ref{fig:g10_M4_M7_M8}c), far exceeding previously reported models capped at $10^6$ spins and $\sim 5$ billion couplings~\cite{goto2021}.

    \item On smaller dense graphs such as $K_{2000}$, DOCH performs best over short time horizons, while ADOCH converges fastest for large and ultra-large models.

    \item For benchmarking on small-scale Ising models, we used $100$ random initializations; the histograms in Figures~\ref{fig:1k_sk}b, \ref{fig:1k_sk}c, \ref{fig:2k_pm1}b, and \ref{fig:2k_pm1}c demonstrate strong robustness to initialization.

    \item Our solvers rapidly achieve GW-SDP-level energies and consistently perform better on the $G_{10}$ graph and Biq Mac instances.

    \item Across all scales, including the ultra-large models in Figures~\ref{fig:g10_M4_M7_M8}c, \ref{fig:g10_M4_M7_M8}d and Supplementary Figures~1 to 11, our methods exhibit faster convergence to lower-energy configurations than competing solvers.
\end{itemize}

A potential consideration arises when comparing our approach to traditional cooling-based Ising solvers, particularly for small graph instances. Annealing-type algorithms, which rely on carefully designed cooling schedules, are known to be asymptotically optimal under ideal conditions. For small-scale problems with modest computational requirements, these solvers can be executed for many iterations within a short runtime, potentially achieving more accurate ground state approximations than our solvers. In contrast, our solvers are inherently independent of cooling schedules and are designed for rapid convergence, often reaching high-quality approximate solutions in just a few iterations. This feature makes our approach especially well-suited for large and ultra-large-scale Ising models. For such large instances, the computational cost of achieving asymptotic optimality with traditional annealing methods becomes prohibitive. In this regime, our approach delivers efficient, high-quality solutions without extensive hyperparameter tuning or prolonged annealing runs. The distinction is clear: while annealing offers incremental benefits for small problems where long runtimes are acceptable, our solvers present a scalable and robust alternative that significantly broadens the range of tractable Ising models. 

In summary, our solvers offer a compelling alternative to traditional simulation-based, cooling-dependent approaches. Their simple update rules enable efficient implementation on GPU clusters, allowing for scalability to ultra-large problem instances. Furthermore, the inherently parallel structure of the algorithms makes them well-suited for deployment on low-level hardware such as FPGAs, offering significant gains in both computational speed and energy efficiency.

\section{Methods}

First-order methods have become the cornerstone of large-scale optimization, particularly in deep learning and related fields~\cite{2015-kingma,lecun2015deep,wir_beck_first-order_2017}.  Their success is driven by a combination of low computational cost per iteration and the ability to scale efficiently to problems with extremely large parameter spaces. By relying solely on gradient information, these methods bypass the need to compute or invert Hessians, making them highly effective for nonconvex problems involving millions or even billions of variables. Additionally, their algorithmic simplicity allows for straightforward parallelization and efficient deployment on hardware accelerators, which has been critical for scaling optimization in modern, data-intensive applications.

Our Ising solver builds on the observation that the Hamiltonian in~\eqref{eq:Hamilton} can be expressed as a difference of convex functions. A particularly effective technique for optimizing such functions is the Difference-of-Convex Algorithm (DCA)~\cite{abbaszadehpeivasti2023}, which is rooted in the classical principle of bound optimization~\cite{dempster1977maximum}. In this approach, the original nonconvex objective is iteratively approximated by a sequence of convex surrogate problems that upper-bound the original function.  By designing the surrogate problems to be easier to solve (for example, using efficient first-order methods), we can scale DCA to large-scale problems. 

As explained next, applying DCA to our problem~\eqref{eq:DCP} results in a simple iterative scheme. Starting from an initial point \(\x^{(k)}\), we linearize $g$ using its first-order approximation around \(\x^{(k)}\). This yields a convex surrogate, which is minimized to obtain the next iterate \(\x^{(k+1)}\). In particular, as $g$ is convex, its linear approximation at \(\x^{(k)}\) gives us a global lower bound~\cite{wir_beck_first-order_2017}. Specifically, for all $\x \in \Re^n$, we have
\begin{equation*}
g(\x) \geqslant g(\x^{(k)}) +  \nabla\! g(\x^{(k)})^\top\! (\x-\x^{(k)}).
\end{equation*}
This results in the following convex upper bound on the Hamiltonian:
\begin{equation*}
 \cH(\x)=  f(\x) - g(\x) \leqslant  f(\x) - g(\x^{(k)}) -  \nabla\! g(\x^{(k)})^\top\! (\x-\x^{(k)}). 
\end{equation*} 
We refine our current estimate $\x^{(k)}$ by minimizing this upper bound. Specifically, we set the update \(\x^{(k+1)}\) as the minimizer of the 
surrogate function:
\begin{equation*}
F(\x) = f(\x) - g(\x^{(k)})- \nabla\! g(\x^{(k)})^\top\! (\x-\x^{(k)}).
\end{equation*}
As $F$ is convex and differentiable, by first-order optimality, 
\begin{equation}
\label{eq:optcon}
\0=\nabla F(\x^{(k+1)}) =\nabla\! f(\x^{(k+1)}) - \nabla\! g(\x^{(k)}).
\end{equation}
Substituting $\nabla\! g(\x) = (\J + \alpha \I)\x$ in \eqref{eq:optcon} and after some calculation, we get
\begin{equation}
\label{eq:FPE}
    \x^{(k+1)}  = \cT (\x^{(k)}), 
\end{equation}
where 
\begin{equation*}
\cT(\x) = \varphi \big( \beta^{-1}(\J + \alpha \I) \x \big) \qquad \mbox{and} \qquad \varphi(x_1,\ldots,x_n) = (\sqrt[3] x_1,\ldots,\sqrt[3] x_n).
\end{equation*}
Thus, $\cT$ is a linear transform followed by a componentwise cube root operation. The resulting algorithm, called Difference-of-Convex Hamiltonian (DOCH), is summarized in \textbf{Algorithm~1}. 

A distinctive property of DOCH is that the updates are monotone:
\begin{equation}
\label{eq:descent}
\cH(\x^{(0)}) \geqslant \cH(\x^{(1)}) \geqslant \cH(\x^{(2)}) \geqslant  \cdots.
\end{equation}
That is, $\cH$ is guaranteed to decrease or stay constant after each iteration. Since $\cH$ is continuous and coercive, the descent property~\eqref{eq:descent} ensures that the sequence $\{\cH(\x^{(k)}\}$ converges to a limiting value. Moreover, because $\cH$ is a polynomial, we can establish that the sequence of iterates $\{\x^{(k)}\}$ converges to a limit point $\x^* \in \Re^n$. We (unconditionally) prove that $\x^*$ is  a critical point of $\cH$, and under mild additional assumptions, that it is a strict local minimizer of $\cH$ (Supplementary Note 3).

\begin{center}
    \begin{tabular}{ll}
        \hline
        \multicolumn{2}{l}{\textbf{Algorithm 1:} DOCH} \\
        \hline
        1: & \textbf{initialization}: $\x^{(0)}$, $N \geqslant 1$ \\
        2: & \textbf{for} $k = 0$ to $N-1$ \textbf{do} \\
        3: & \hspace{1em} compute $\x^{(k+1)} =\cT(\x^{(k)})$. \\
        4: & \textbf{end for} \\
        5: & \textbf{return:} $\s = \mathrm{sign}(\x^{(N)})$. \\
        \hline
    \end{tabular}
    \label{algo:DOCH}
\end{center}

A natural extension of our algorithm is to incorporate acceleration into DCA~\cite{phan2018}.  Accelerated first-order methods improve the convergence speed of traditional first-order algorithms. A prominent example is Nesterov’s acceleration~\cite{Nesterov2018,2015-kingma}, which achieves the optimal convergence rate for convex problems, outperforming standard gradient descent. Accelerated methods preserve the simplicity and low per-iteration cost of standard first-order methods, while achieving faster convergence. This makes them especially effective for large-scale optimization problems. In our case, we apply Nesterov-style acceleration to develop the Accelerated DOCH (ADOCH) algorithm, with the main steps outlined in \textbf{Algorithm~2}. 

\begin{center}
    \begin{tabular}{ll}
        \hline
        \multicolumn{2}{l}{\textbf{Algorithm 2:} Accelerated DOCH (ADOCH)} \\
        \hline
        1: & \textbf{initialization}: $\x^{(0)}$, $\y^{(0)} = \x^{(0)}$, $t_0=1$, $q \geqslant 1$, $N \geqslant 2$. \\
        2:  & \textbf{for} $k = 0$ to $N-1$ \textbf{do} \\
        3: & \hspace{1em} compute $t_{k+1} = \big(1+\sqrt{1+4t_k^2} \big)/2$. \\
        4: & \hspace{1em} \textbf{if} $k\geqslant 1$ \\
        5: & \hspace{2em} compute $\y^{(k)} = \x^{(k)} + ((t_k-1)/t_{k+1}) \, ( \x^{(k)} - \x^{(k-1)})$. \\
        6: & \hspace{1em} \textbf{if} $\cH(\y^{(k)}) \leqslant \max \big\{\cH(\x^{(\max(0, k-q))}),\dots, \cH(\x^{(k)}) \big\}$ \\
        7: & \hspace{2em} set $\bv^{(k)} = \y^{(k)}$. \\
        8: & \hspace{1em} \textbf{else} \\
        9: & \hspace{2em} set $\bv^{(k)} = \x^{(k)}$. \\
        10: & \hspace{1em} \textbf{end if} \\
        11: & \hspace{1em} compute $\x^{(k+1)} = \cT(\bv^{(k)})$. \\
        12: & \textbf{end for} \\
        13: & \textbf{return:} $\s = \mathrm{sign}(\x^{(N)})$. \\
        \hline
    \end{tabular}
     \label{algo:ADOCH}
\end{center}

The main difference with DOCH is the use of momentum. We extrapolate $\x^{(k)}$ using the previous update $\x^{(k-1)}$:
\begin{equation*}
\y^{(k)} = \x^{(k)} + \frac{t_k-1}{t_{k+1}} \, ( \x^{(k)} - \x^{(k-1)}),
\end{equation*}
where the sequence $\{t_k\}$ is set using Nesterov's optimal scheme~\cite{Nesterov2018}. To decide whether to accept the extrapolated point \(\y^{(k)}\), we check the condition
\begin{equation}
\label{eq:BBscheme}
\cH(\y^{(k)}) \leqslant \max  \big\{\cH(\x^{(k-q)}),\dots, \cH(\x^{(k)}) \big\},
\end{equation}
where $q$ controls the look-back window~\cite{phan2018}. This Barzilai–Borwein type criterion helps the method escape shallow local minima of \(\cH\) \cite{grippo2002}. Theoretically, a large value of $q$ ensures better acceleration and quicker convergence~\cite{wright2009}. If \eqref{eq:BBscheme} holds, we update \(\x^{(k+1)} = \cT(\y^{(k)})\); otherwise, we fall back to the non-extrapolated update. Taking the update at \(\y^{(k)}\) rather than \(\x^{(k)}\) is the key distinction from DOCH.

{\paragraph*{Data availability:} The authors declare that the main data supporting the findings of this study are provided within the paper and its supplementary files. The public benchmark datasets used in our experiments, including the G-set graphs and the Biq Mac Library, can be accessed at
\url{https://web.stanford.edu/~yyye/yyye/Gset/} and \url{https://biqmac.aau.at/biqmaclib.html}.}






\clearpage
\begin{sidewaystable}[t!]
\centering
    \begin{tabular}{|c|c|c|c|c|c|c|c|c|}
    \hline
    \multicolumn{3}{|c|}{\begin{cellvarwidth}[t]\centering Ising Solvers \end{cellvarwidth}} &
    \begin{cellvarwidth}[t]\centering Computational\\ Complexity \end{cellvarwidth} &
    \begin{cellvarwidth}[t]\centering Scalability/\\Parallelization \end{cellvarwidth} &
    \begin{cellvarwidth}[t]\centering Approximation\\Quality \end{cellvarwidth} &
    \begin{cellvarwidth}[t]\centering Convergence\\Guarantee \end{cellvarwidth} &
    \begin{cellvarwidth}[t]\centering Use of \\Cooling \end{cellvarwidth} &
    \begin{cellvarwidth}[t]\centering Parameter\\Tuning \end{cellvarwidth} \tabularnewline
    \hline\hline
    \multirow{4}{*}{\begin{cellvarwidth}[t]\centering $ $\\$ $\\Annealing\\$ $ \end{cellvarwidth}} &
    \begin{cellvarwidth}[t]\centering Markov Chain \\ Monte Carlo \\ (MCMC) \end{cellvarwidth} &
    \begin{cellvarwidth}[t]\centering SA~\cite{kirkpatrick1983,isakov2015},\\ MARS~\cite{Yavorsky2019},\\ NMFA~\cite{King2018} \end{cellvarwidth} &
    \begin{cellvarwidth}[t]\centering $ $\\high\\$ $ \end{cellvarwidth} &
    \begin{cellvarwidth}[t]\centering $ $\\yes\\$ $ \end{cellvarwidth} &
    \begin{cellvarwidth}[t]\centering $ $\\poor\\$ $ \end{cellvarwidth} &
    \begin{cellvarwidth}[t]\centering $ $\\no\\$ $ \end{cellvarwidth} &
    \begin{cellvarwidth}[t]\centering $ $\\no\\$ $ \end{cellvarwidth} &
    \begin{cellvarwidth}[t]\centering $ $\\hard\\$ $ \end{cellvarwidth} \tabularnewline
    \cline{2-9}
    & \multirow{3}{*}{\begin{cellvarwidth}[t]\centering Simulated\\Hamiltonian\\Dynamics \end{cellvarwidth}} &
    bSB~\cite{goto2019,goto2021} & high & yes & good & no & no & hard \tabularnewline
    \cline{3-9}
    & & CIM~\cite{tiunov2019,honjo2021} & high & yes & good & no & no & hard \tabularnewline
    \cline{3-9}
    & & SIA~\cite{jiang2024} & high & yes & good & no & no & hard \tabularnewline
    \hline
    \multirow{3}{*}{\begin{cellvarwidth}[t]\centering Continuous\\Optimization\\$ $ \end{cellvarwidth}} &
    \multicolumn{2}{c|}{CP~\cite{beck2002}} & low & \begin{cellvarwidth}[t]\centering no \end{cellvarwidth} & poor & yes & yes & --- \tabularnewline
    \cline{2-9}
    & \multicolumn{2}{c|}{SDP~\cite{Luo2010}} & low & \begin{cellvarwidth}[t]\centering no \end{cellvarwidth} & poor & yes & yes & --- \tabularnewline
    \cline{2-9}
    & \multicolumn{2}{c|}{\begin{cellvarwidth}[t]\centering DOCH (present work) \end{cellvarwidth}} & high & yes & good & yes & yes & easy \tabularnewline
    \hline
    \end{tabular}
    \caption{Comparison of key aspects of our Ising solver with existing solvers. \textbf{SA}: Simulated Annealing, MARS: Mean field Annealing from a Random State, \textbf{NMFA}: Noisy Mean Field Annealing, \textbf{bSB}: ballistic Bifurcation Machine, \textbf{CIM}: Coherent Ising Machine, \textbf{SIA}: Spring Ising Algorithm, \textbf{CP}: Convex Programming, \textbf{SDP}: Semidefinite Programming, \textbf{DOCH}: Difference of Convex Hamiltonian. (A convergence guarantee means that the iterative process will asymptotically reach a fixed point or a local minimum. An approximation guarantee, on the other hand, concerns the quality of the final solution, describing how close it is to the true ground state).}
    \label{table:ising_solvers_comp_table}
\end{sidewaystable}

\clearpage
\begin{figure}[t!]
	\centering
    \includegraphics[width=0.80\linewidth]{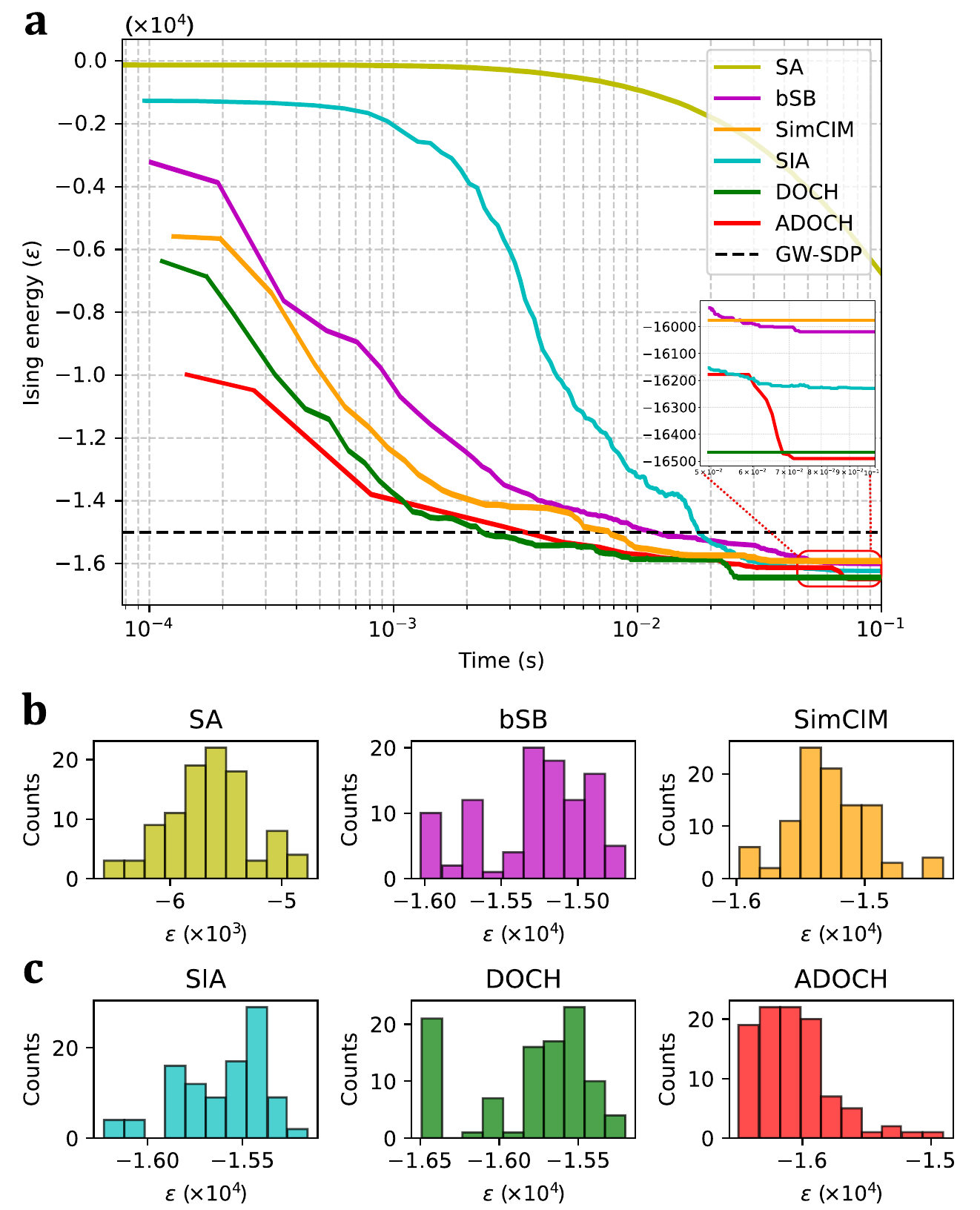}
	\caption{\textbf{Evolution of Ising energy with runtime for the $10^3$-spin SK model.} \textbf{(a)} Comparison of our solvers (DOCH and ADOCH) with state-of-the-art Ising solvers: Simulated Annealing (SA)~\cite{isakov2015, inagaki2016}, ballistic Bifurcation Machine (bSB)~\cite{goto2021}, Simulated Coherent Ising Machine (SimCIM)~\cite{goto2021}, and Spring Ising Algorithm (SIA)~\cite{jiang2024}. Solid curves show the best energy values achieved across $100$ independent runs (Supplementary Notes 6 and 7 for the parameter settings). The dashed black line indicates the lowest energy obtained using the GW-SDP~\cite{goto2019, inagaki2016}. \textbf{(b), (c)} Histograms of $100$ Ising energies, each obtained after $1000$ iterations and with different initializations. Results obtained on a NVIDIA RTX 3050 GPU laptop.}
	\label{fig:1k_sk}
\end{figure}

\clearpage

\begin{figure}[t!]
	\centering
\includegraphics[width=0.80\linewidth]{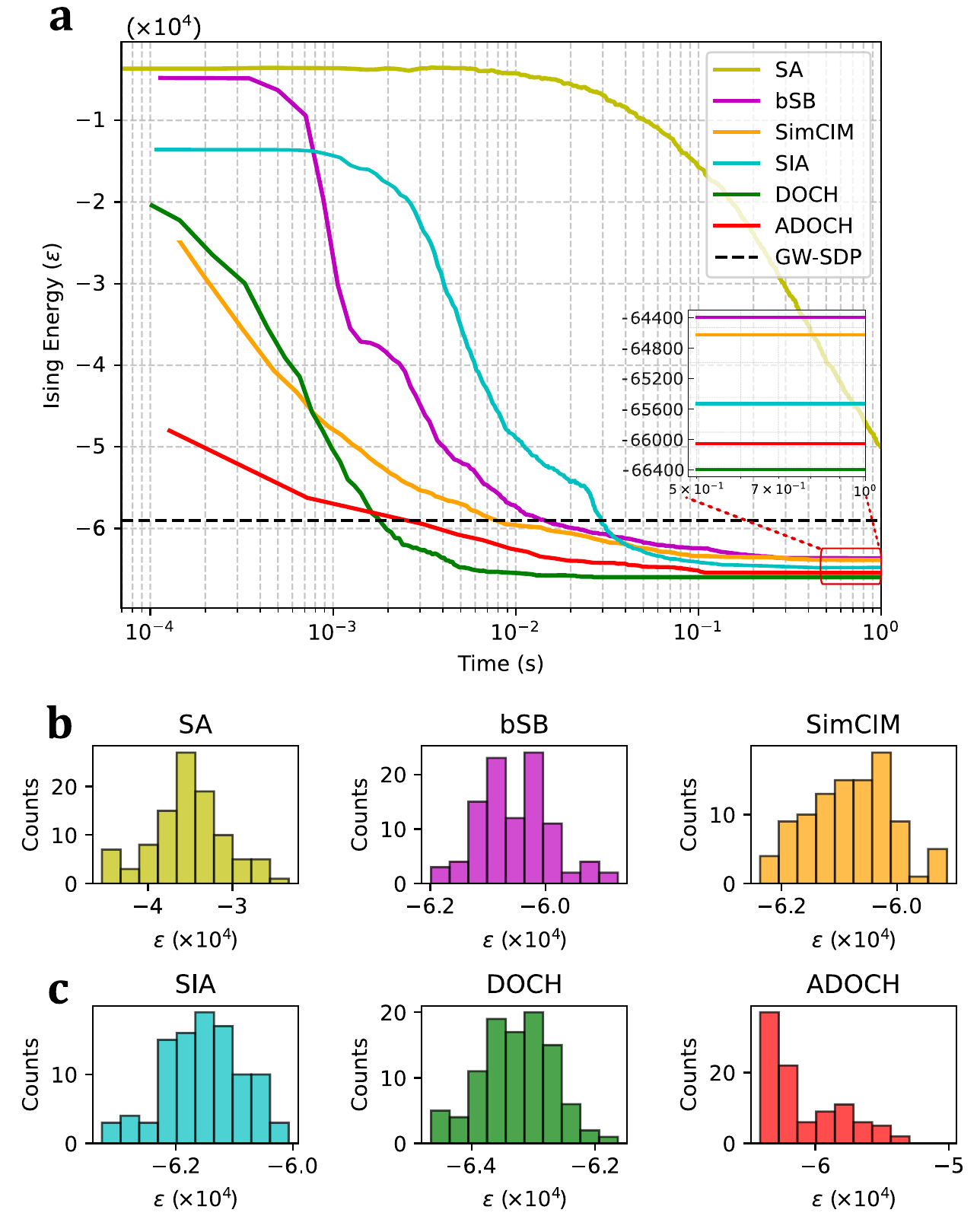}
	\caption{\textbf{Ising energy versus runtime on the \(K_{2000}\) benchmark.} \textbf{(a)} Comparison of our solvers (DOCH and ADOCH) with state-of-the-art Ising solvers. Solid curves show the mean Ising energy across $100$ independent runs (Supplementary Notes 6 and 7 for the parameter settings). The dashed black line indicates the lowest energy obtained using the GW-SDP~\cite{goto2019, inagaki2016}. \textbf{(b), (c)} Histograms of $100$ Ising energies, each obtained after $1000$ iterations and with different initializations. Results obtained on a NVIDIA RTX 3050 GPU laptop.}
	\label{fig:2k_pm1}
\end{figure}

\clearpage

\begin{figure}[t!]
	\centering
	\includegraphics[width=1.0\linewidth]{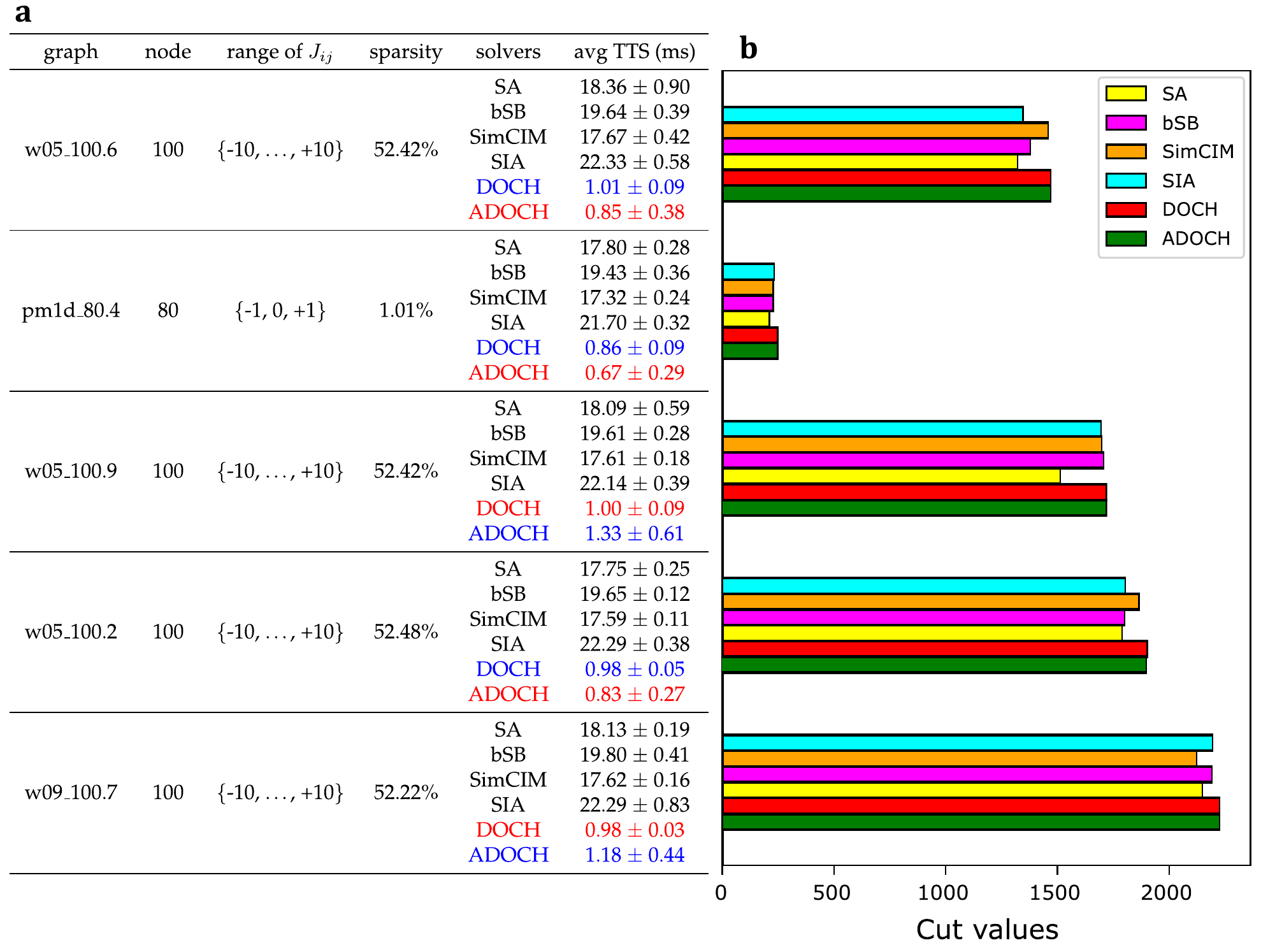}
	\caption{\textbf{MAX-CUT results on Biq Mac graphs.} \textbf{(a)} Table summarizing the graph characteristics and the average time-to-solution (Avg TTS) for each solver to reach $99$\% of the maximum cut value. (Best Avg TTS values are highlighted in red, and second-best in blue). \textbf{(b)} Bar chart comparing the best cut values achieved after $1000$ iterations by SA~\cite{isakov2015, inagaki2016}, bSB~\cite{goto2021}, SimCIM~\cite{goto2021}, SIA~\cite{jiang2024}, and our Ising solvers DOCH and ADOCH. Results obtained on NVIDIA Jetson Nano.}
    \label{fig:biq-mac}
\end{figure}

\clearpage

\begin{figure}[t!]
    \centering
    \includegraphics[width=1.0\linewidth]{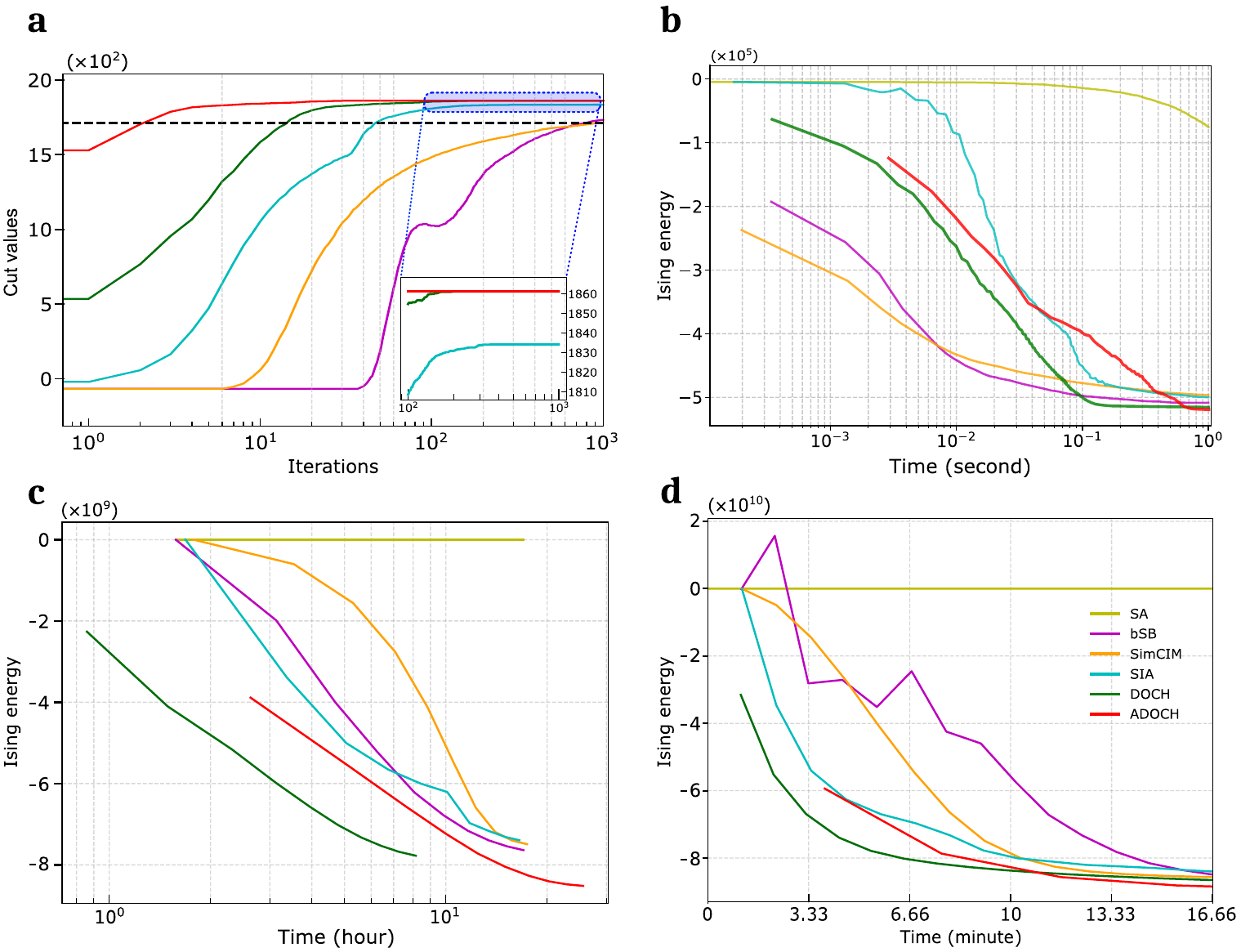}
    \caption{\textbf{Benchmarking solver performance at varying Ising model sizes.} 
    \textbf{(a)} MAX-CUT optimization on an $800$-node $G_{10}$ graph. The dashed black line marks the lowest energy found using the GW-SDP. {Each solver is run for $1000$ iterations, with $100$ different initializations. The mean cut values are then plotted against the number of iterations.} Results obtained on an NVIDIA Jetson Nano, with cut values plotted against iterations. \textbf{(b)} Benchmarking on a $10^4$-spin Sherrington–Kirkpatrick (SK) model, executed on NVIDIA Jetson Nano. Ising energy is shown as a function of time (log scale) in seconds.
    \textbf{(c)} Fully connected $10^7$-spin Ising model with couplings $J_{ij} = \sin(i \, j + \mathrm{seed})$. Experiments were conducted using four NVIDIA H100 GPUs. The energy trajectories over time are shown.
    \textbf{(d)} Sparse $10^8$-spin Ising model with $0.00001\%$ connectivity and $9$-bit signed integer couplings ($J_{ij} \in \{-2^9+1, \dots, 2^9-1\}$), executed on two NVIDIA H100 GPUs. Energy is plotted against computation time (in minutes).}
    \label{fig:g10_M4_M7_M8}
\end{figure}

\clearpage



\renewcommand{\thefigure}{S\arabic{figure}}
\renewcommand{\thetable}{S\arabic{table}}
\renewcommand{\theequation}{S\arabic{equation}}
\renewcommand{\thepage}{S\arabic{page}}
\setcounter{figure}{0}
\setcounter{table}{0}
\setcounter{equation}{0}
\setcounter{page}{1} 

\makeatletter
\renewcommand{\fnum@figure}{\textbf{Figure \thefigure}}
\renewcommand{\fnum@table}{\textbf{Table \thetable}}
\makeatother

\section{Supplementary Material}

\subsection*{Supplementary Figures}

\begin{figure}[htbp]
    \centering
    \includegraphics[width=0.8\textwidth]{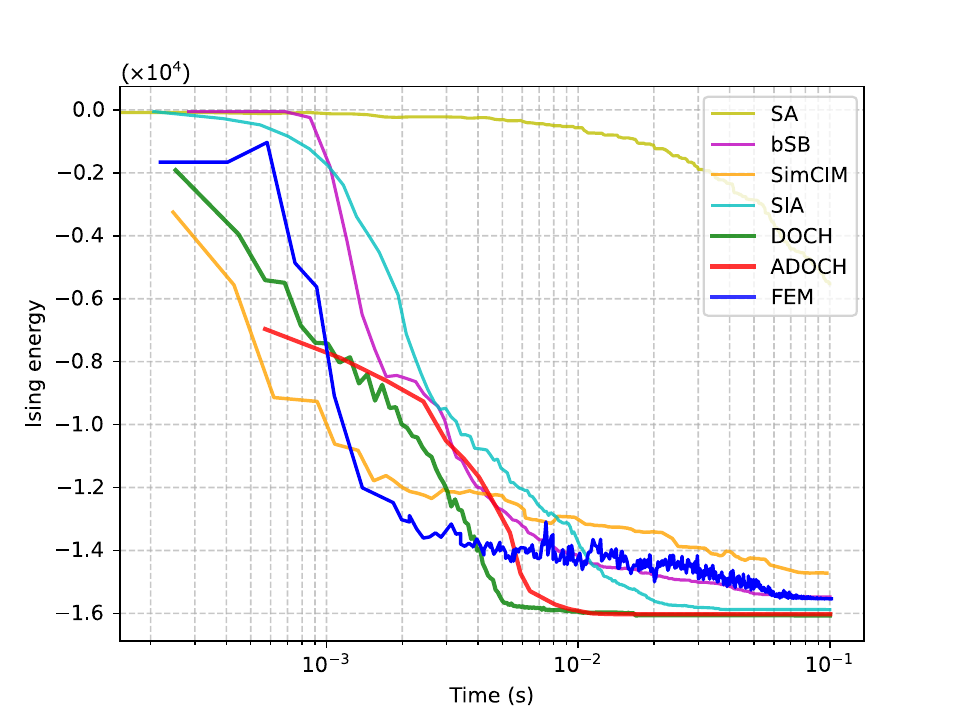}
    \caption{Benchmarking results for a $10^3$ spin Sherrington-Kirkpatrick (SK) model. The plot shows the evolution of Ising energy as a function of computation time (log scale) for several state-of-the-art solvers alongside our proposed solvers. 
    Our solvers exhibit significantly faster convergence and attain lower energy configurations compared to existing techniques. 
    The Free Energy Machine (FEM) was evaluated using its official implementation~\cite{fem2025} on an NVIDIA Jetson Nano.
    }
    \label{fig:supp_more_fem_1}
\end{figure}

\begin{figure}[htbp]
    \centering
    \includegraphics[width=0.8\textwidth]{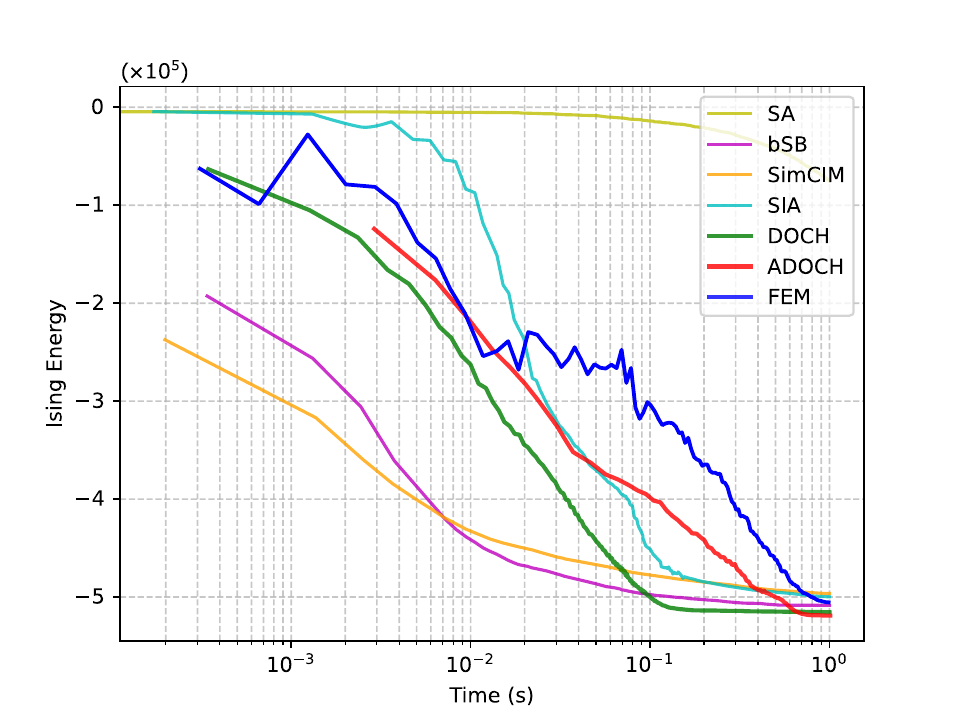}
    \caption{Benchmark results for a $10^4$ spin SK model. The plot shows Ising energy versus computation time on a logarithmic scale for a range of established and proposed solvers. 
    Our solvers consistently demonstrate advantages in both convergence speed and energy minimization as the problem size increases. 
    FEM results are obtained using the official implementation~\cite{fem2025}, executed on an NVIDIA Jetson Nano.
    }
    \label{fig:supp_more_fem_2}
\end{figure}

\begin{figure}[htbp]
    \centering
    \includegraphics[width=0.8\textwidth]{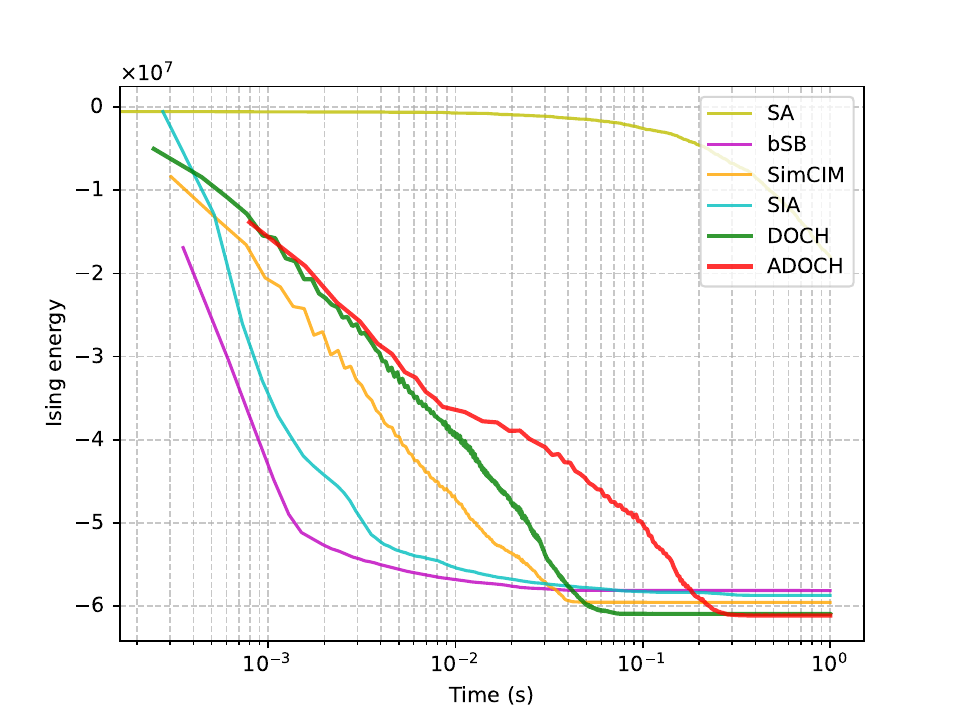}
    \caption{Benchmark results for a large-scale sparse Ising instance with $10^4$ spins and $1\%$ connectivity. 
    Non-zero couplings $J_{ij}$ are drawn uniformly from $9$-bit signed integers ($J_{ij} \in \{-2^9+1,\dots, 2^9-1\}$). 
    Ising energy is plotted against computation time on a logarithmic scale. 
    The results highlight the scalability and efficiency of our solvers under low-connectivity and discrete-weight regimes. 
    All experiments were executed on a NVIDIA Jetson Nano.
    }
    \label{fig:supp_more_3}
\end{figure}

\begin{figure}[htbp]
    \centering
    \includegraphics[width=0.8\textwidth]{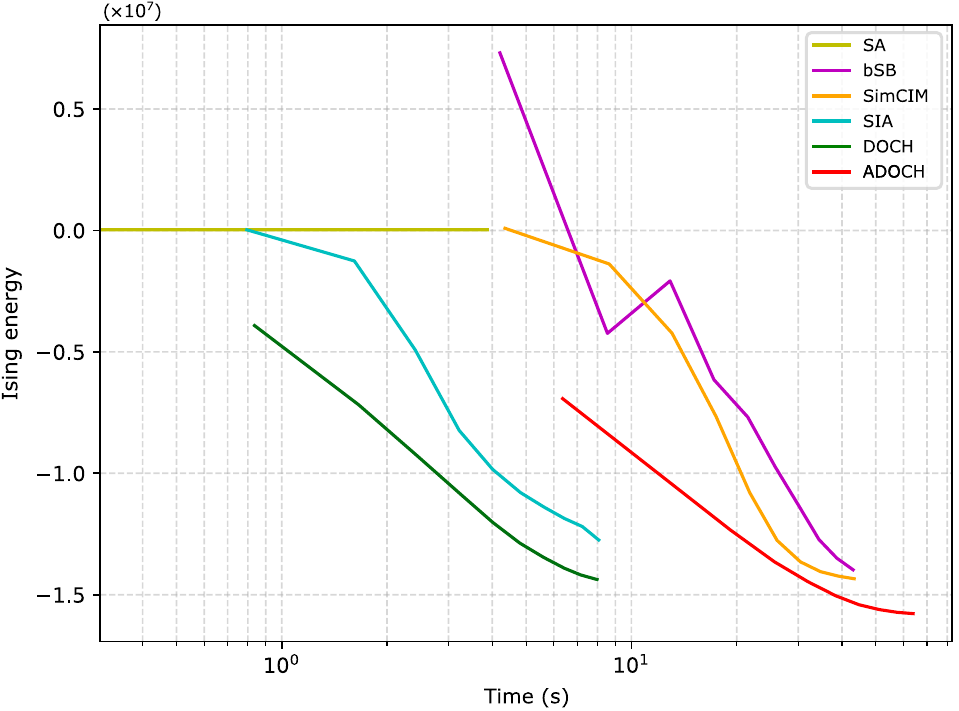}
    \includegraphics[width=0.8\textwidth]{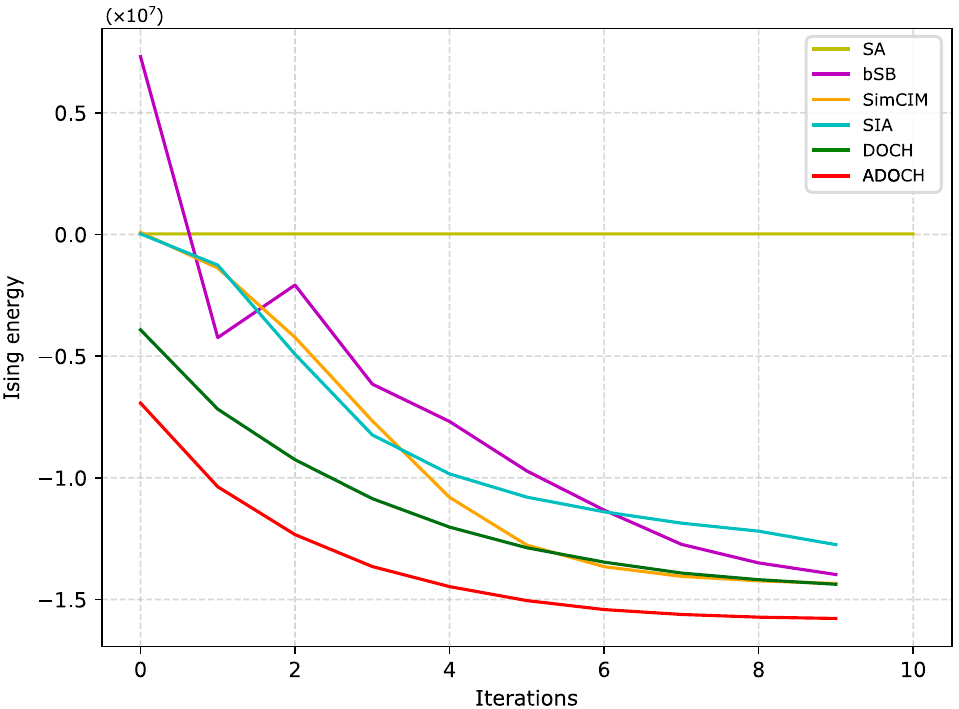}
    \caption{Performance comparison on K100000 (a fully connected Ising model with $10^5$ spins and binary couplings $J_{ij} \in \{-1, +1\}$ drawn with equal probability). 
    \textbf{Top:} Ising energy evolution over time (log scale). 
    \textbf{Bottom:} Ising energy as a function of iteration count. 
    These results demonstrate the rapid energy descent of our solvers on large and dense instances. All experiments were performed on a single NVIDIA RTX 3090 GPU.
    }
    \label{fig:supp_more_4}
\end{figure}

\begin{figure}[htbp]
    \centering
    \includegraphics[width=0.8\textwidth]{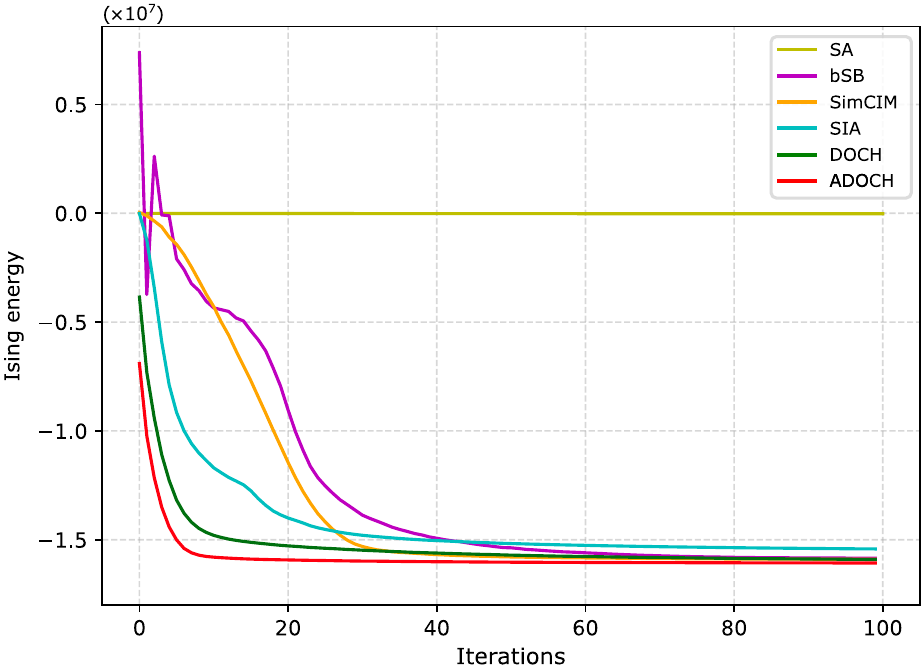}
    \caption{Performance comparison on a fully connected $10^5$ spin SK. The Ising energy is plotted against the iteration count, up to $100$ iterations.
    We used a single NVIDIA RTX 3090 GPU for the computations.
    }
    \label{fig:supp_more_5}
\end{figure}

\begin{figure}[htbp]
    \centering
    \includegraphics[width=0.8\textwidth]{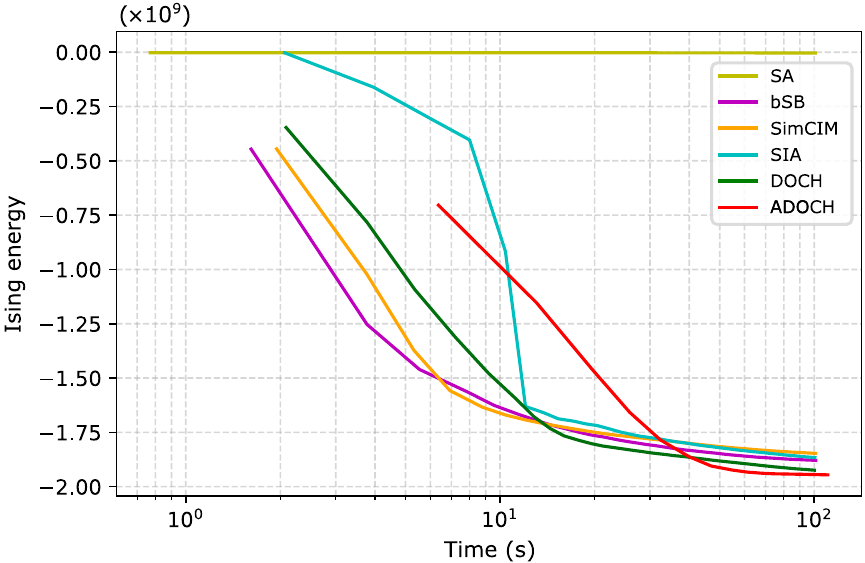}
    \caption{Performance comparison of Ising energy for a $10^5$ spin Ising model with $1\%$ connectivity. Couplings $J_{ij}$ are drawn uniformly from signed $9$-bit integers ($J_{ij} \in \{-2^9+1, \dots, 2^9-1\}$). The results are for a single NVIDIA RTX 3090 GPU.
    }
    \label{fig:supp_more_6}
\end{figure}

\begin{figure}[htbp]
    \centering
    \includegraphics[width=\textwidth]{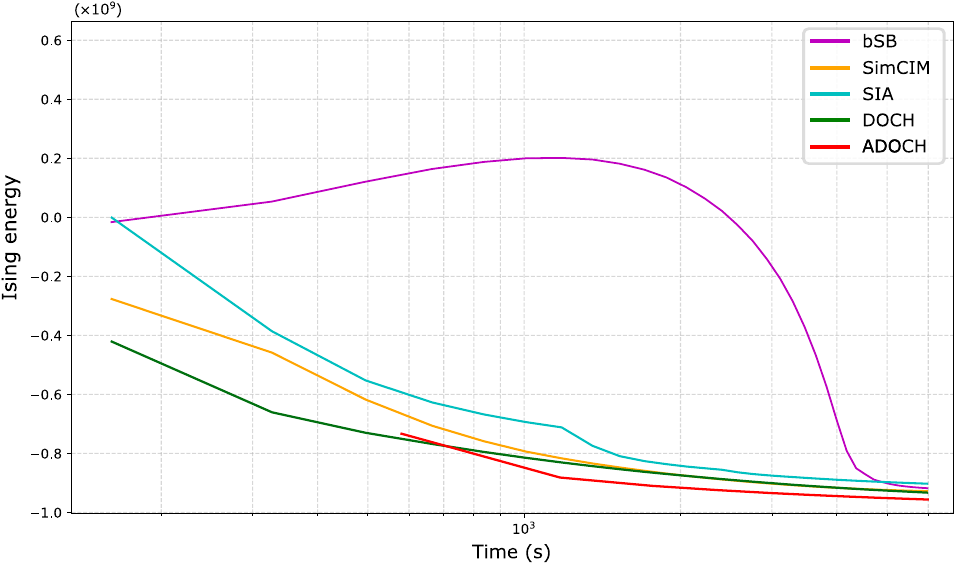}
    \caption{Time evolution of Ising energy for a $10^6$ spin fully connected Ising model with structured couplings $J_{ij} = \sin(i \, j + \mathrm{seed})$. 
    The results are for a single NVIDIA V100 GPU.
    }
    \label{fig:supp_more_7}
\end{figure}

\begin{figure}[htbp]
    \centering
    \includegraphics[width=0.8\textwidth]{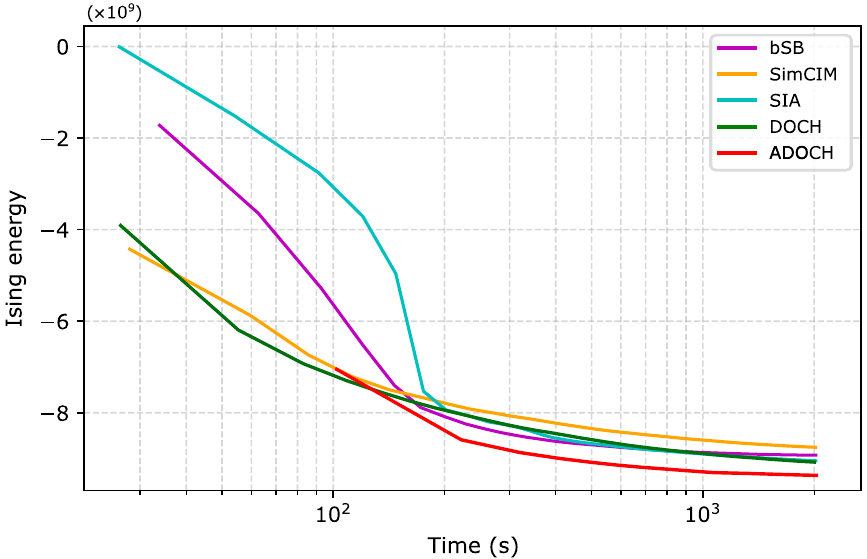}
    \caption{Time evolution of Ising energy for a $10^6$ spin sparsely connected Ising model, where $0.1\%$ of the couplings are nonzero. 
    Each active coupling $J_{ij}$ is sampled uniformly from the set of $9$-bit signed integers ($J_{ij} \in \{-2^9+1, \dots, 2^9-1\}$). 
    All solvers were run on a single NVIDIA V100 GPU.
    }
    \label{fig:supp_more_8}
\end{figure}

\begin{figure}[htbp]
    \centering
    \includegraphics[width=0.8\textwidth]{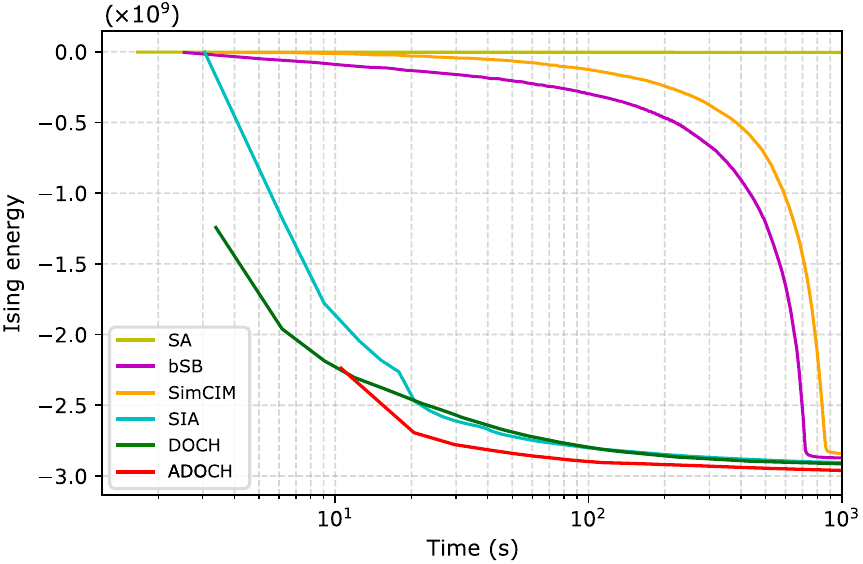}
    \caption{Time evolution of Ising energy for a $10^6$ spin extremely sparse Ising model with $0.01\%$ connectivity.
    Each nonzero coupling $J_{ij}$ is drawn uniformly from the set of $9$-bit signed integers ($J_{ij} \in \{-2^9+1, \dots, 2^9-1\}$).
    All solver experiments were executed on a single NVIDIA V100 GPU.
    }
    \label{fig:supp_more_9}
\end{figure}

\begin{figure}[htbp]
    \centering
    \includegraphics[width=0.8\textwidth]{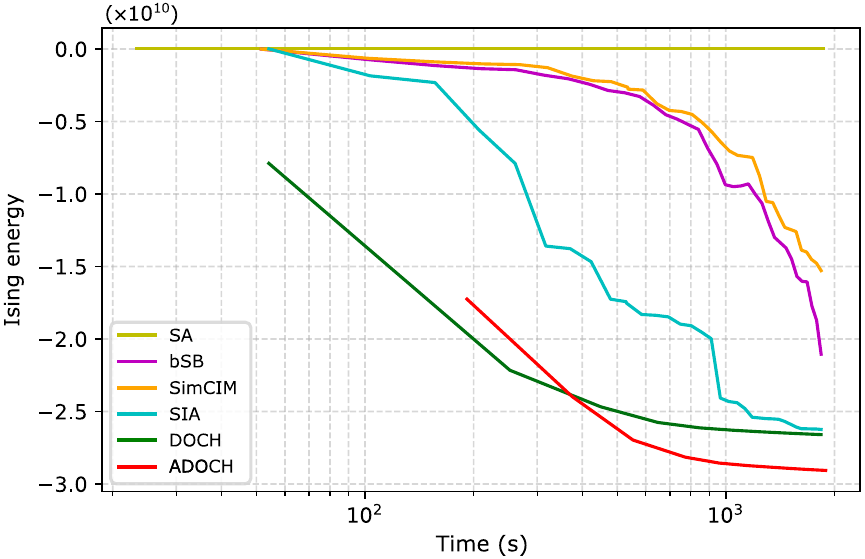}
    \caption{Benchmarking performance on a $10^7$ spin sparse Ising model with $0.001\%$ connectivity. Coupling coefficients $J_{ij}$ are drawn uniformly at random from the set of signed 9-bit integers ($J_{ij} \in \{-2^9+1, \dots, 2^9-1\}$). The vertical axis represents the Ising energy, while the horizontal axis (on a log scale) denotes computation time in seconds. Each solver was executed using four NVIDIA V100 GPUs.}
    \label{fig:supp_more_10}
\end{figure}

\begin{figure}[htbp]
    \centering
    \includegraphics[width=0.8\textwidth]{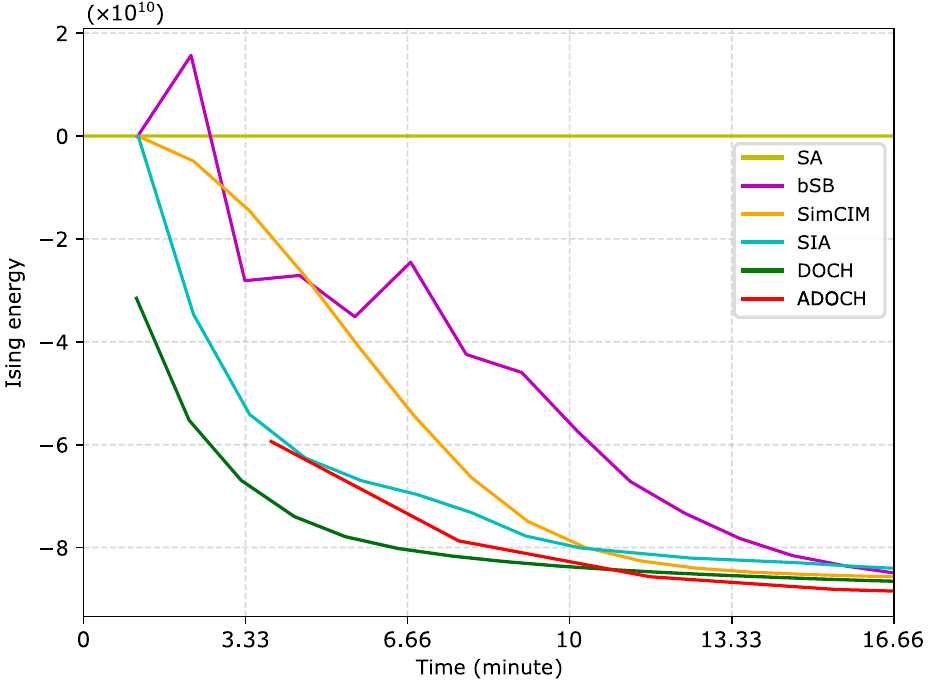}
    \caption{Benchmarking performance on a $10^8$-spin extremely sparse Ising model with connectivity of only $0.00001\%$. Each nonzero coupling $J_{ij}$ is drawn uniformly at random from the set of signed 9-bit integers, $\{-2^9 + 1, \dots, 2^9 - 1\}$. The algorithms were run on two NVIDIA H100 GPUs. The Ising energy is plotted as a function of computation time (in minutes).}
    \label{fig:supp_more_11}
\end{figure}

\begin{figure}[htbp]
   \centering
   \includegraphics[width=0.9\textwidth]{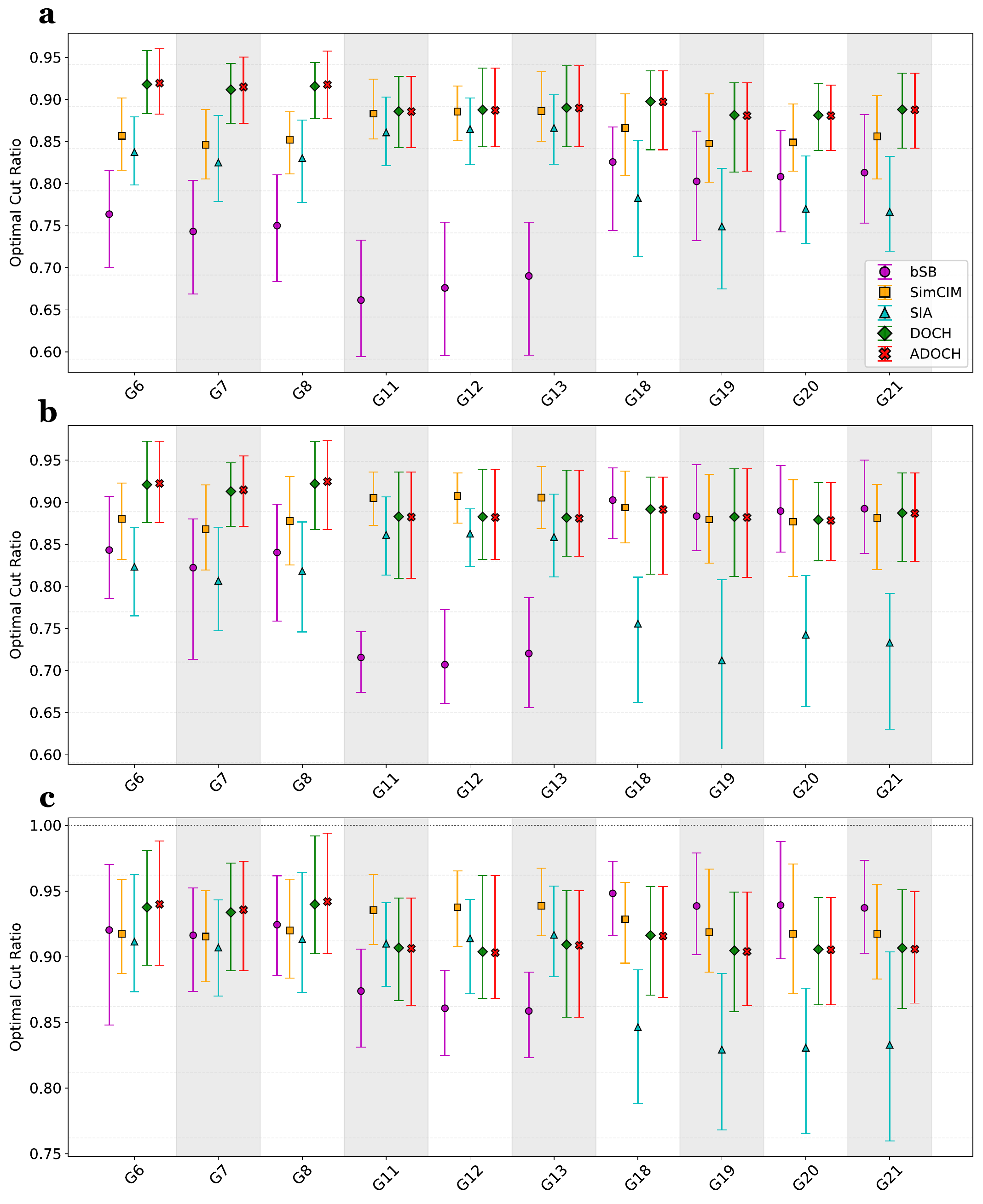}
   {\caption{\textbf{MAX-CUT on G-set graphs.} The mean cut ratio (mean cut normalized by the cut value reported in~\cite{maxcut2017multiple}) is shown for \href{https://web.stanford.edu/~yyye/yyye/Gset/}{G-set graphs} (random: $G_6, G_7, G_8$; toroidal: $G_{11}, G_{12}, G_{13}$; planar: $G_{18}, G_{19}, G_{20}, G_{21}$). Panels \textbf{(a)}, \textbf{(b)}, and \textbf{(c)} show the results after $5$, $10$, and $100$ iterations. The error bars indicate the maximum and minimum values across $100$ random initializations, respectively. The computations were performed on a NVIDIA RTX 3050 GPU laptop.}
   \label{fig:error_bar_plot_gset}}
\end{figure}

\begin{figure}[htbp]
   \centering
   \includegraphics[width=1.0\textwidth]{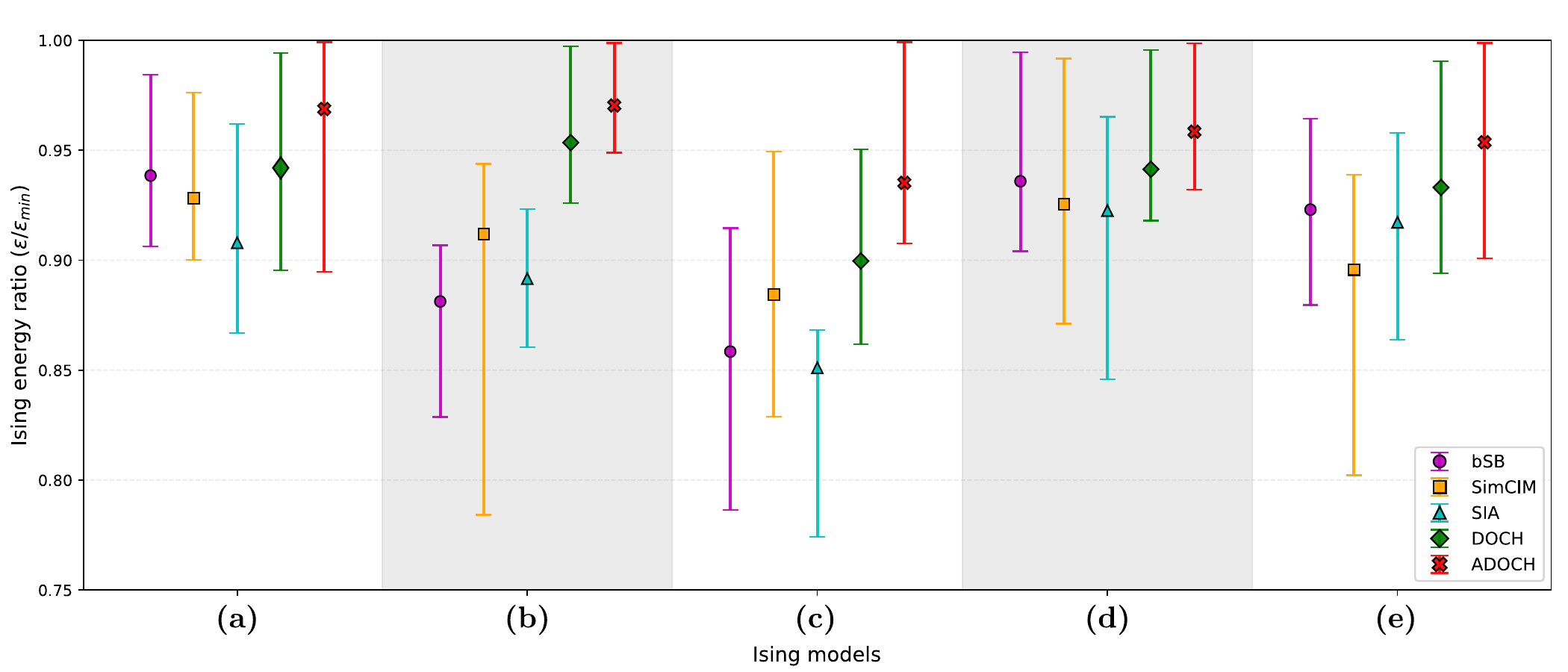}
   {\caption{\textbf{Error-bar analysis for medium and large-scale Ising models.} The plots show the Ising energy ratio (energy obtained by a solver normalized by the best energy across solvers) after a fixed runtime $T$. For each solver, the mean energy ratio is shown, with error bars indicating the maximum and minimum energy values over $100$ random initializations. (a) $10^4$-spin SK model, runtime $T = 1$ s; (b) $10^4$-spin Ising model with $1\%$ connectivity, runtime $T = 1$ s. Couplings $J_{ij}$ are drawn uniformly from signed 9-bit integers ($J_{ij} \in {-2^9+1,\dots,2^9-1}$); (c) $K100000$ graph, $T = 10$ s; (d) $10^5$-spin SK model, $T = 10$ s; (e) $10^5$-spin Ising model with $1\%$ connectivity, runtime $T = 10$ s. Couplings $J_{ij}$ are drawn uniformly from signed 9-bit integers ($J_{ij} \in {-2^9+1,\dots,2^9-1}$). All the computations were performed on a single NVIDIA RTX 3090 GPU.}
   \label{fig:error_bar_plots}}
\end{figure}

\begin{figure}[htbp]
    \centering
    \includegraphics[width=0.8\textwidth]{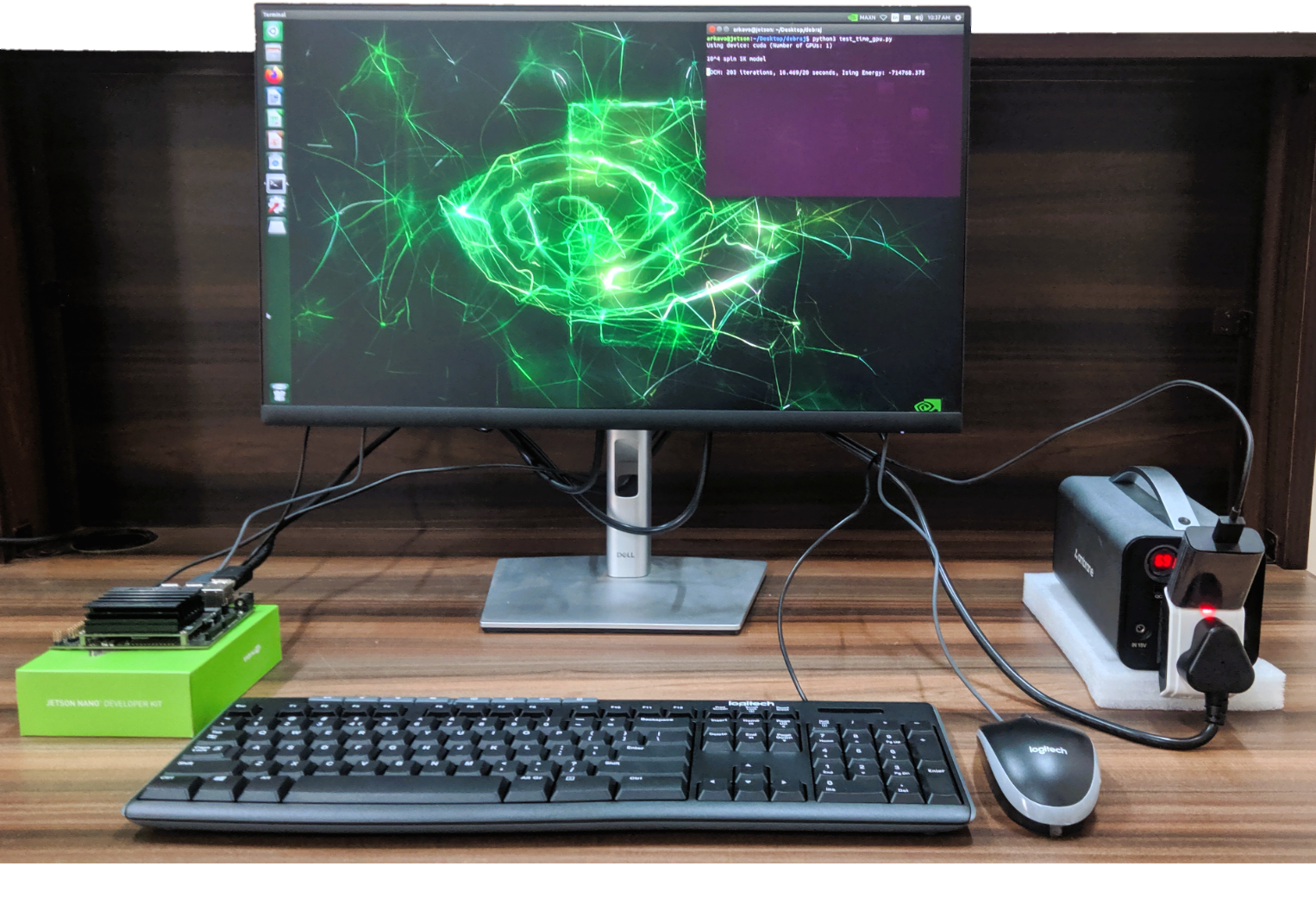}
    \caption{\textbf{Battery-powered edge computing setup based on the NVIDIA Jetson Nano module.} The system includes a display, keyboard, and mouse interfaced with the Jetson board, powered through a $60,000$ mAh portable power system.}
    \label{fig:battery_jetson_nano_setup}
\end{figure}

\clearpage

\subsection*{Supplementary Note 1: Reduction to the Homogeneous Model}

We can represent any $n$-spin Ising model with an external field as an $(n+1)$-spin Ising model without an external field. Consider the Ising model with an external field:
\begin{equation}\label{eq:h_field_ising}
 \cE(\s)= -\frac{1}{2}\s^\top \J \s- \h^\top \! \s \qquad \big(\s \in \{-1,1\}^n \big),
\end{equation} 
where \(\J\) is the coupling matrix and \(\h\) is the external field vector. Now, define the coupling matrix and energy function
\begin{equation}
\label{eq:zero_field_ising}
\hat{\J} =
 \begin{pmatrix}
        \J & \h\\
        \h^\top  & 0
    \end{pmatrix}, \qquad 
    \hat{\cE}(\boldsymbol{\sigma}) 
 = -\frac{1}{2} \boldsymbol{\sigma}^\top \hat{\J} \boldsymbol{\sigma} \qquad \big(\boldsymbol{\sigma} \in \{-1,1\}^{n+1}\big).
\end{equation}

\begin{proposition}
\label{prop:h_field_zero_field}
Suppose $\boldsymbol{\sigma}^* = (\s_0, t_0)$ is the ground state of \eqref{eq:zero_field_ising}, where $\s_0 \in \{-1, +1\}^n$ and $t_0 \in \{-1,1\}$. Then $\s^* = t_0  \s_0 \in \{-1, +1\}^n$ is the ground state of \eqref{eq:h_field_ising}.
\end{proposition}

\begin{proof}
Let $\s \in \{-1,1\}^n$ and $t \in \{-1,1\}$. From \eqref{eq:h_field_ising} and \eqref{eq:zero_field_ising}, 
\begin{equation}
\label{eq:equiv}
 \hat{\cE}\left((\s,t) \right) = \cE(t \s).
\end{equation}
On the other hand, by the definition of $\boldsymbol{\sigma}^*$, 
\begin{equation}
\label{eq:comp}
\hat{\cE}(\boldsymbol{\sigma}^*) = \hat{\cE}\left((\s_0, t_0)\right) \leqslant   \hat{\cE}\left((\s, t)\right).
\end{equation}
Comparing \eqref{eq:equiv} and \eqref{eq:comp}, we obtain 
\begin{equation*}
\cE(\s^*)=  \cE(t_0 \s_0) = \hat{\cE}(\left(\s_0, t_0\right) )\leqslant   \hat{\cE}(\left(\s, 1)\right) =  \cE(\s)
\end{equation*}
Since this holds for any $\s \in \{-1,1\}^n$, it follows that $\s^*$ is a ground state of \eqref{eq:h_field_ising}.
\end{proof}

\subsection*{Supplementary Note 2: Mathematical Results}

We establish several useful mathematical properties of the Hamiltonian underlying our optimization model. Some of these properties will be useful for the convergence analysis of our iterative solver. Recall that we work with the homogeneous Ising problem:
\begin{equation}
\label{eq:IsingE}
    \min_{\s \in\{-1,+1\}^n} \  \cE(\s) = -\frac{1}{2}\s^\top \J \s.
\end{equation}
We replace the binary vector \(\s\) with \(\x \in \Re^n\) and consider the corresponding relaxed energy
\begin{equation}
\label{eq:E}
 \cE(\x) = -\frac{1}{2}\x^\top \J \x.
\end{equation}
To bias the components of \(\x\) toward binary values, we introduce the attractor 
\begin{equation}
\label{eq:A}
        \cA(\x) = \frac{\beta}{4}(x_1^4 + \dots + x_n^4) - \frac{\alpha}{2}(x_1^2 + \dots + x_n^2).
\end{equation}
\begin{proposition}
\label{prop:global_minima_A}
The global minimizers of \(\cA\) are exactly the points $\{-\lambda,+\lambda\}^n,\, \lambda = \sqrt{\alpha/\beta}$.
\end{proposition}


\begin{proof}
The gradient of the attractor $\nabla \! \cA (\x) \in \Re^n$ has components
\begin{equation}
\label{eq:gradH}
\nabla \! \cA(\x)_i =  \beta x_i^3 - \alpha x_i \qquad (i=1,\ldots,n),
\end{equation}
and its Hessian $\nabla^2 \! \cA(\x) \in \Re^{n \times n}$ is a diagonal matrix with 
\begin{equation}
\label{eq:hessianH}
\nabla^2\! \cA(\x)_{ii} = 3\beta x_i^2 - \alpha  \qquad (i=1,\ldots,n).
\end{equation}
We can conclude from~\eqref{eq:gradH} that the critical points of $\cA$ are exactly the points $\{-\lambda, 0, \lambda\}^n$, where  $\lambda = \sqrt{\alpha/\beta}$. This is because $\nabla \! \cA(\x)= \0$ if and only if each $x_i \in \{-\lambda,0, \lambda\}$. Now, a sufficient condition for a critical point $\x$ to be a local minimizer is that $\nabla^2\! \cA(\x)_{ii} > 0$ for all $i=1,\ldots,n$. It follows from~\eqref{eq:hessianH} that this condition holds if each $x_i \in \{-\lambda, \lambda\}$. On the other hand, if $x_i =0$ for some $i=1,\ldots,n$, then $\nabla^2\! \cA(\x)_{ii} < 0$, so that $\x$ cannot be a local minimizer. Thus, the local minimizer of $\cA$ are exactly the points $\{-\lambda,+\lambda\}^n$. Finally, observe that \(\cA\) attains the same value at each of these points. Therefore, these points are exactly the global minimizers of \(\cA\).
\end{proof}

In other words, the attractor penalizes deviations from $\pm \sqrt{\alpha/\beta}$. This property is useful because it forces the relaxed continuous variables toward values that closely mimic the original binary spins, thereby guiding the optimization process toward candidate ground states of the Ising model.

\vspace{1em}

We form the Hamiltonian by combining the Ising energy \(\cE\) with the attractor \(\cA\):
\begin{equation}
\label{eq:Hamiltonian}
\cH = \cA + \cE.
\end{equation}
which forms the basis for our proposed Ising solver. The key observation behind our algorithm is that we can write~\eqref{eq:Hamiltonian} as a difference of convex functions:
\begin{equation*}
\cH = f-g,
\end{equation*}
where 
   \begin{equation}
   \label{eq:fandg}
        f(\x) = \frac{\beta}{4}\left(x_1^4 + \dots + x_n^4 \right) 
        \quad \mbox{and} \quad g(\x) = \frac{1}{2}\x^\top \! (\J + \alpha \I) \x.
    \end{equation}

\begin{proposition}
\label{prop:convexity_fg}
The function $f$ is convex for all \(\beta > 0\), and there exists $\alpha_0$ such that $g$ is strongly convex  for all \(\alpha > \alpha_0\).
\end{proposition}

\begin{proof}
Both $f$ and $g$ are twice differentiable, and their Hessians are
\begin{equation*}
\nabla^2\! f(\x) =
 3\beta \begin{pmatrix}
 x_1^2 & & \\
& \ddots & \\
& &  x_n^2
\end{pmatrix}  
\quad \text{and} \quad 
\nabla^2\! g(\x) = \J + \alpha \I.
\end{equation*}
Clearly, \(\nabla^2 \! f(\x)\) is positive semidefinite for all $\beta > 0$, which establishes the convexity of $f$. 

On the other hand, \(\nabla^2 \!g(\x)\) is positive definite if and only if $\lambda_{\min}(\J + \alpha\I) > 0$. As $\J$ is symmetric, it has real eigenvalues. Moreover, since $\J_{ii}=0$ for all $i$, one of its eigenvalues must be $<0$. In particular, if we define $\alpha_0= \lambda_{\max}(-\J)=-\lambda_{\min}(\J)$, then $\alpha_0 > 0$, and $\lambda_{\min}(\J + \alpha\I) > 0$ for all $\alpha > \alpha_0$. This establishes the strong convexity of $g$.
\end{proof}

\begin{proposition}\label{prop:coercive}
The Hamiltonian \(\cH\) is coercive for all \(\beta > 0\), that is, \(\cH(\x)\to+\infty\) as \(\|\x\|_2\to \infty\). 
\end{proposition}

\begin{proof}
We have
\begin{equation*}
 \sum_{i=1}^{n}x_{i}^{4} \geq \frac{1}{n^{2}}\Bigl(\sum_{i=1}^{n}x_{i}^{2}\Bigr)^{2} = \frac{1}{n^{2}}\|\x\|_2^{4},
\end{equation*}
where $\| \x\|_2$ is the standard Euclidean norm of $\x$. This gives us the lower bound
\begin{equation*}
\label{eq:f_lower_bound}
f(\x)
 = \frac{\beta}{4}  \sum_{i=1}^{n}x_{i}^{4} 
 \geq \frac{\beta}{4n^{2}}\|\x\|_2^{4}.
\end{equation*}
On the other hand, we have the upper bound
\begin{equation*}
\label{eq:g_upper_bound}
g(\x) = \frac{1}{2}\x^{\top}\!(\J + \alpha \I)\x  \leq \frac{1}{2}  c \|\x\|_2^{2}, 
\end{equation*}
where $c=\lambda_{\max}(\J+\alpha \I)>0$.
Combining the above bounds, we have
\begin{equation*}
\cH(\x) = f(\x)-g(\x) \geq \frac{\beta}{4n^{2}}\|\x\|_2^{4}
 - \frac{1}{2} c\|\x\|_2^{2} = \frac{1}{4 n^2} \|\x\|_2^{2} \, \Bigl(\beta \|\x\|_2^{2}-2n^2 c\Bigr).
\end{equation*}
As $\beta > 0$, the right side goes to $+\infty$ as $ \|\x\|_2 \to \infty$. Hence, \(\cH\) is coercive.
\end{proof}

\begin{theorem}\label{thm:existence_minimizer}
The Hamiltonian $\cH$ is bounded below on \(\Re^{n}\) and has a minimizer.
\end{theorem}

\begin{proof}
This follows from the Weierstrass extreme value theorem~\cite{rudin1987real}, since $\cH$ is continuous and coercive.
\end{proof}

\begin{proposition}\label{thm:scaled_ising}
Suppose $\x^*$ is a minimizer of $\cH$ and it lies on a vertex of the scaled hypercube $\{-\lambda, \lambda\}^n$, where $\lambda > 0$. Then $\mathrm{sign}(\x^*)$ is a minimizer of~\eqref{eq:IsingE}.
\end{proposition}

\begin{proof}
Let $\s^* \in \{-1,1\}^n$ be a minimizer of of \eqref{eq:IsingE}. By optimality, $\cE(\s^*) \leqslant \cE(\mathrm{sign}(\x^*))$. Hence, we just have to show that 
\begin{equation}
\label{eq:claim}
\cE(\mathrm{sign}(\x^*)) \leqslant \cE(\s^*).
\end{equation}
By assumption, for all $\x \in \Re^n$,
\begin{equation}
\label{eq:optH}
\cA(\x^*) + \cE(\x^*)=\cH(\x^*) \leqslant \cH(\x)=  \cA(\x) + \cE(\x).
\end{equation}
{Since $\x^* \in \{-\lambda, +\lambda\}^n$, we have $\x^* = \lambda \, \mathrm{sign}(\x^*)$. Setting $\x^* = \lambda \, \mathrm{sign}(\x^*)$ and $\x = \lambda \s^*$ in~\eqref{eq:optH} gives}
\begin{equation}
\label{eq:temp}
\cA(\lambda \, \mathrm{sign}(\x^*)) + \cE(\lambda \, \mathrm{sign}(\x^*))\leqslant  \cA(\lambda \s^*) + \cE(\lambda \s^*).
\end{equation}
{From~\eqref{eq:E} and ~\eqref{eq:A}, it follows that $\cE(\lambda \x) = \lambda^2 \cE(\x)$ and $\cA(\lambda\, \mathrm{sign}(\x^*)) = \cA(\lambda \s^*)$. Substituting these into~\eqref{eq:temp}, we get~\eqref{eq:claim}.}
\end{proof}

\subsection*{Supplementary Note 3: Convergence Analysis of DOCH}

We present a self-contained convergence analysis of our Ising solver DOCH. Recall that DOCH iteratively minimizes the Hamiltonian $\cH=f-g$ using the DCA algorithm. Following Proposition~\ref{prop:convexity_fg}, we assume that $g$ is strongly convex, and 
\begin{equation}
\label{eq:mu}
\mu=\lambda_{\min}(\J + \alpha\I) > 0.
\end{equation}

\begin{proposition}
\label{prop:monotone_descent_H} 
Let \(\{\x^{(k)}\}\) be the iterates generated by DOCH. Then
\begin{equation}
\label{eq:descent_ineq}
\cH\bigl(\x^{(k)}\bigr) - \cH\big(\x^{(k+1)}\big) \geq \frac{\mu}{2}\, \|\x^{(k+1)}-\x^{(k)}\|^{2}.
\end{equation}
\end{proposition}

\begin{proof}
By construction, \(\x^{(k+1)}\) is the minimizer of the function 
\begin{equation}
\label{eq:defxk}
F(\x) = f(\x) - \big(g(\x^{k}) + \nabla \! g(\x^{(k)})^\top \!(\x-\x^{(k)}) \big).
\end{equation}
Therefore,
\begin{equation*}
f(\x^{(k+1)}) - g(\x^{k}) - \nabla \! g(\x^{(k)})^\top \!(\x^{(k+1)}-\x^{(k)})  = F(\x^{(k+1)}) \leqslant F(\x^{(k)}) = f(\x^{(k)}) -  g(\x^{k}),
\end{equation*}
giving us
\begin{equation}
\label{eq:delf}
f(\x^{(k)}) - f(\x^{(k+1)})  \geqslant  \nabla \! g(\x^{(k)})^\top \!(\x^{(k)}-\x^{(k+1)}).
\end{equation}
As $\cH= f- g$, we have from~\eqref{eq:delf} that
\begin{align}
\label{eq:right}
\cH\bigl(\x^{(k)}\bigr) - \cH\big(\x^{(k+1)}\big) &= f(\x^{(k)}) - g(\x^{(k)}) - \big(f(\x^{(k+1)}) - g(\x^{(k+1)}) \big) \nonumber \\
& \geqslant  g(\x^{(k+1)}) -  g(\x^{(k)})  - \nabla \! g(\x^{(k)})^\top \!(\x^{(k+1)}-\x^{(k)}).
\end{align}
Substituting the formula of $g$ and $\nabla \! g$  (see~\eqref{eq:fandg}) in~\eqref{eq:right}, we get
\begin{align*}
\cH\bigl(\x^{(k)}\bigr) - \cH\big(\x^{(k+1)}\big)  \geqslant  \frac{1}{2} (\x^{(k+1)}-\x^{(k)} )^\top \! (\J+\alpha \I)  (\x^{(k+1)}-\x^{(k)} ).
\end{align*}
Moreover, from~\eqref{eq:mu},
\begin{align*}
(\x^{(k+1)}-\x^{(k)} )^\top \! (\J+\alpha \I)  (\x^{(k+1)}-\x^{(k)} )  \geqslant \mu \|\x^{(k)}-\x^{(k+1)}\|^{2}.
\end{align*}
This establishes the desired result~\eqref{eq:descent_ineq}.
\end{proof}

In general, the challenging part of convergence analysis is to show that the iterates remain bounded. This is often assumed in the analysis~\cite{phan2018}. However, for DOCH, boundedness of the iterates can be established unconditionally.

\begin{proposition}\label{prop:bounded_iterates}
The DOCH iterates are bounded, and
\begin{equation}
\label{eq:step_vanish}
\lim_{k\to\infty} \ \bigl\|\x^{(k+1)}-\x^{(k)}\bigr\|_2  = 0.
\end{equation}
\end{proposition}

\begin{proof}
By Proposition \ref{prop:monotone_descent_H}, the sequence \(\{\cH(\x^{(k)})\}\) is monotone (nonincreasing). Hence, for all \(k \geqslant 1\),
\begin{equation}
\label{eq:mono}
\cH(\x^{(k)}) \leq \cH(\x^{(0)}).
\end{equation}
As $\cH$ is coercive, the sublevel set $\{\x : \cH(\x) \leq \cH(\x^{(0)})\}$ is bounded. However, by~\eqref{eq:mono}, the sequence $\{\x^{(k)}\}$ is contained in this sublevel set and is hence bounded.

On the other hand, summing the descent inequality~\eqref{eq:descent_ineq} from \(k=0\) to \(N\) gives
\begin{equation*}
\cH(\x^{(0)}) - \cH(\x^{(N+1)})= \sum_{k=0}^{N}\Bigl(\cH\bigl(\x^{(k)}\bigr)-\cH\bigl(\x^{(k+1)}\bigr)\Bigr)
 \geq \frac{\mu}{2}\sum_{k=0}^{N}\bigl\|\x^{(k)}-\x^{(k+1)}\bigr\|_2^{2}.
\end{equation*}
Since $\cH(\x^{(k)})$ is bounded, there exists $C>0$ such that $\cH(\x^{(k)}) \geqslant C$ for all $k \geqslant 1$. Hence, 
\begin{equation*}
\frac{\mu}{2}\sum_{k=0}^{N} \|\x^{(k)}-\x^{(k+1)} \|_2^{2}  \leqslant  \cH(\x^{(0)}) - \cH(\x^{(N+1)}) \leqslant \cH(\x^{(0)}) - C < +\infty .
\end{equation*}
Letting $N \to \infty$ gives
\begin{equation}
\label{eq:finite_square_sum}
\frac{\mu}{2}\sum_{k=0}^{\infty} \|\x^{(k)}-\x^{(k+1)} \|_2^{2}  <  +\infty.
\end{equation}
Consequently, we obtain~\eqref{eq:step_vanish}.
\end{proof}

Note that~\eqref{eq:finite_square_sum} by itself does not guarantee that the sequence \(\{\x^{(k)}\}\) converges. To prove convergence, we must show that \(\{\x^{(k)}\}\) is a Cauchy sequence~\cite{rudin1987real}. This will be done in Theorem~\ref{thm:DOCH_convergence}, using the fact that the Hamiltonian $\cH$ is real-analytic (being a polynomial) and therefore satisfies the Łojasiewicz gradient inequality~\cite{KL1963}. Based on the analysis so far, we can already state the following result.

\begin{proposition}\label{prop:limit_critical}
All limit points of the DOCH iterates are critical points of \(\cH\).
\end{proposition}

\begin{proof}
The iterates \(\{\x^{(k)}\}\) are bounded by Proposition~\ref{prop:bounded_iterates}, and hence, by the Bolzano-Weierstrass theorem~\cite{rudin1987real}, have at least one limit point, say, \(\x^* \in \Re^n\). In particular, we can extract a subsequence \(\{\x^{(k_i)}\}\) such that
\begin{equation}
\label{eq:lim1}
\lim_{i \to \infty} \ \x^{(k_i)} = \x^*.
\end{equation}
We claim that $\x^*$ is a critical point of $\cH$, i.e., $\nabla\cH(\x^{*})=\0$. Indeed, since $\x^{(k+1)}$ is a minimizer of \eqref{eq:defxk},  we have by first-order optimality that
\begin{equation*}
\0 = F(\x^{(k+1)})=\nabla \! f(\x^{(k+1)}) - \nabla \! g(\x^{(k)}).
\end{equation*}
Thus, for all $k \geqslant 1$, 
\begin{equation}
\label{eq:successive}
\nabla \! f(\x^{(k+1)}) = \nabla \! g(\x^{(k)})
\end{equation}
In particular, 
\begin{equation}
\label{eq:diffH}
\nabla  \cH (\x^{(k_i)}) = \nabla \! f(\x^{(k_i)}) - \nabla \! g(\x^{(k_i)}) = \nabla \! f(\x^{(k_i)}) - \nabla \! f(\x^{(k_i+1)}).
\end{equation}
Now, by the triangle inequality, 
\begin{equation*}
\|\x^{(k_{i+1})} -\x^* \|_2 \leqslant \|\x^{(k_i+1)} -\x^{(k_i)} \|_2+ \|\x^{(k_i)} -\x^* \|_2.
\end{equation*}
Letting $i \to \infty$ on both sides using \eqref{eq:finite_square_sum} gives
\begin{equation}
\label{eq:lim2}
\lim_{i \to \infty} \ \x^{(k_i+1)} = \x^*.
\end{equation}
Finally, letting $i \to \infty$ in~\eqref{eq:diffH}, we obtain $\nabla \cH(\x^*)=\0$ from~\eqref{eq:lim1} and~\eqref{eq:lim2}.
\end{proof}

We now use the classical Łojasiewicz gradient inequality \cite{KL1963} to show that the DOCH iterates actually converge, that is, they admit a unique limit point. For this, we need the following result.

\begin{proposition}\label{prop:H_grad_bound}
There exists $\sigma > 0$ such that  for all $k \geqslant 1$,
    \begin{equation}
    \label{eq:H_grad_bound}
        \|\nabla\cH(\x^{(k)})\|_2 \leq \sigma \|\x^{(k)} - \x^{(k-1)}\|_2.
    \end{equation}
\end{proposition}

\begin{proof}
From~\eqref{eq:successive}, we have
\begin{equation*}
\nabla\cH(\x^{(k)}) = \nabla \! f(\x^{(k)}) -  \nabla \! g(\x^{(k)})= \nabla \! g(\x^{(k-1)}) -  \nabla \! g(\x^{(k)}).
\end{equation*}
In particular, since $ \nabla \! g(\x) = (\J+\alpha \I) \x$, 
\begin{equation*}
\| \nabla\cH(\x^{(k)}) \|_2 \leqslant \| (\J+\alpha \I)(\x^{(k-1)}-\x^{(k)}) \|_2.
\end{equation*}
Letting $\sigma = \lambda_{\max} (\J+\alpha \I) >0$, we get~\eqref{eq:H_grad_bound}.
\end{proof}

The following is the main convergence result. The proof follows the analysis in~\cite{absil2005,bolte_kl_2013}.

\begin{theorem}\label{thm:DOCH_convergence}
The DOCH iterates \(\{\x^{(k)}\}\) converge to a critical point of $\cH$.
\end{theorem}

\begin{proof}
We know that \(\{\x^{(k)}\}\)  is bounded and hence has at least one limit point \(\x^* \in \Re^n\). Moreover, we know from Proposition~\ref{prop:limit_critical} that $\x^*$ is a critical point of \(\cH\). To complete the proof, we just have to show that the entire sequence \(\{\x^{(k)}\}\) converges to \(\x^*\). 

From Proposition \ref{prop:monotone_descent_H}, we know that \(\{\cH(\x^{(k)})\}\) is non-increasing. Moreover, since $\x^*$ is a limit point of $\x^{(k)}$ and $\cH$ is continuous,  
\begin{equation}
\label{eq:limH}
\lim_{k \to \infty} \, \cH(\x^{(k)}) =  \cH(\x^{*}).
\end{equation}
Let 
\begin{equation*}
\Delta_{k}=\cH(\x^{(k)})-\cH(\x^{*}) \qquad \mbox{and} \qquad \delta_k = \|\x^{(k)}-\x^{(k-1)}\|.
\end{equation*}
We can assume that $\Delta_k > 0$ and $\delta_k >0$ for all $k \geqslant 1$. Indeed, if $\Delta_K=0$ for some $K \geqslant 1$, then $\cH(\x^{(k)})=\cH(\x^{(k+1)})$ for all $k \geqslant  K$. Consequently, we have from~\eqref{eq:descent_ineq} that $\delta_k =0$ for all $k > K$. In other words, if $\Delta_K$ or $\delta_k$ vanishes for some $K$, then the sequence \(\{\x^{(k)}\}\) becomes eventually constant, and therefore must converge to the limit point \(\x^*\).

Since \(\cH\) is a polynomial, it satisfies the Łojasiewicz gradient inequality~\cite{KL1963,absil2005} in some neighbourhood of the critical point \(\x^*\). More specifically, there exist \(\theta \in [0, 1), c >0\) and \(r>0\), such that for all $ \|\x-\x^{*}\|<r$,
\begin{equation}
\label{eq:L_inequality}
    |\cH(\x) - \cH(\x^*)|^\theta \leq c \, \|\nabla\cH(\x)\|_2,
\end{equation}
We can conveniently reformulate this in terms of the function
\begin{equation}
\label{eq:varphi}
\varphi(t) = \frac{c}{1-\theta} \, t^{1-\theta}.
\end{equation}
We can verify that \eqref{eq:varphi} is concave and non-decreasing on $[0, \infty)$, and $\varphi'(t) = c t^{-\theta}$ for any $t > 0$. Therefore, we can write~\eqref{eq:L_inequality} as
\begin{equation}
\label{eq:newform}
\varphi' \big(|\cH(\x) - \cH(\x^*)| \big)\,  \bigl\|\nabla\cH(\x)\bigr\|_2 \geq 1.
\end{equation}

We know from~\eqref{eq:step_vanish} and~\eqref{eq:limH} that $\delta_k \to 0$ and $\Delta_k \to 0$. Since $\varphi$ is non-decreasing and satisfies $\varphi(t)\to 0$ as $t\to 0$, it follows that $\varphi(\Delta_k)\to 0$. Moreover, since $\x^*$ is a limit point of $\{\x^{(k)}\}$, there exists $k_0 \geqslant 1$ such that 
\begin{equation}
\label{eq:prox}
\| \x^{(k_0)} -\x^*\| < \frac{r}{2} \qquad \mbox{and} \qquad \frac{2\sigma}{\mu} \varphi(\Delta_{k_0}) + \delta_{k_0} < \frac{r}{2}.
\end{equation}
In particular, we have from~\eqref{eq:newform} that $\varphi' (\Delta_{k_0}  )  \|\nabla\cH(\x^{(k_0)})\|_2 \geq 1$. Combining this with Proposition~\ref{prop:H_grad_bound}, we get 
\begin{equation}\label{eq:kll}
\varphi'\!\bigl(\Delta_{k_0} \bigr) \, \geq \frac{1}{\sigma\delta_{k_0}}.
\end{equation}
On the other hand, by~\eqref{eq:descent_ineq}, we have
\begin{equation}
\label{eq:diffDelta}
\Delta_{k_0}-\Delta_{k_0+1} = \cH\bigl(\x^{(k_0)}\bigr)-\cH\bigl(\x^{(k_0+1)}\bigr)
 \geq \frac{\mu}{2}\,\delta_{k_0+1}^{2}.
\end{equation}
Since \(\varphi\) is concave and differentiable  on $(0, \infty)$, 
\begin{equation*}
\varphi(\Delta_{k_0})-\varphi(\Delta_{k_0+1}) \geq \varphi'(\Delta_{k_0}) (\Delta_{k_0}-\Delta_{k_0+1}) .
\end{equation*}
From~\eqref{eq:kll} and~\eqref{eq:diffDelta}, we obtain
\begin{align*}
\varphi(\Delta_{k_0})-\varphi(\Delta_{k_0+1}) \geq \frac{\mu}{2}\,\varphi'(\Delta_{k_0})\,\delta_{k_0+1}^{2} \geqslant  \frac{\mu}{2\sigma}\,
\left(\frac{\delta_{k_0+1}^{2}}{\delta_{k_0}} \right).
\end{align*}
In other words,
\begin{equation*}
\delta_{k_0+1}^2 \leqslant \frac{2\sigma}{\mu} \bigl(\varphi(\Delta_{k_0})-\varphi(\Delta_{k_0+1})\bigr)\,\delta_{k_0}.
\end{equation*}
Applying the AM-GM inequality, we get
\begin{equation}
\label{eq:delbnd}
 2\delta_{k_0+1} \leqslant \frac{2\sigma}{\mu} \bigl(\varphi(\Delta_{k_0})-\varphi(\Delta_{k_0+1})\bigr) + \delta_{k_0}.
\end{equation}
Combining this with~\eqref{eq:prox} gives
\begin{equation}
\label{eq:bnd}
\delta_{k_0+1}  \leqslant \frac{2\sigma}{\mu} \bigl(\varphi(\Delta_{k_0})-\varphi(\Delta_{k_0+1})\bigr) + \delta_{k_0} \leqslant \frac{2\sigma}{\mu} \varphi(\Delta_{k_0}) + \delta_{k_0} \leqslant \frac{r}{2}.
\end{equation}
From~\eqref{eq:prox} and~\eqref{eq:bnd}, we have
\begin{equation*}
\|\x^{k_0+1} - \x^*\|_2 \leqslant \|\x^{k_0+1} - \x^{k_0}\|_2+\|\x^{k_0} - \x^*\|_2 = \delta_{k_0+1} +\|\x^{k_0} - \x^*\|_2< r.
\end{equation*}

Thus $\x^{k_0+1}$ is in the neighbourhood of $\x^*$ stipulated in~\eqref{eq:L_inequality}. In fact, we claim that the entire tail $\{\x^{k}\}, k > k_0$, belongs to this neighbourhood. This can be shown using contradiction. Indded, suppose \(k_1\) is the smallest \(k > k_0\) such that $\|\x^{(k_1)}-\x^{*}\|_2 \geqslant r$. This means that the points $\x_{k_0},\ldots,\x_{k_1-1}$ are with the neighbourhood stipulated in~\eqref{eq:L_inequality}. Applying~\eqref{eq:delbnd} repeatedly, we have 
\begin{equation*}
\label{eq:recurrent_inequality}
\delta_{k+1} \leqslant \frac{2\sigma}{\mu} 
\bigl(\varphi(\Delta_{k})-\varphi(\Delta_{k+1})\bigr) + (\delta_k - \delta_{k+1}), \qquad (k_0 \leq k < k_1).
\end{equation*}
Summing both sides, we obtain
\begin{equation}
\label{eq:summed_recurrent_inequality_0}
 \sum_{k=k_0}^{k_1-1} \delta_{k+1} \leqslant \frac{2\sigma}{\mu}  \bigl(\varphi(\Delta_{k_0})-\varphi(\Delta_{k_1})\bigr) + (\delta_{k_0} - \delta_{k_1}).
\end{equation}
In particular, 
\begin{equation*}
\sum_{k=k_0}^{k_1-1} \delta_{k+1} \leqslant  \frac{2\sigma}{\mu}  \varphi(\Delta_{k_0}) + \delta_{k_0} < \frac{r}{2},
\end{equation*}
and 
\begin{equation*}
\|\x^{(k_1)}-\x^{*}\|_2 \leqslant \|\x^{(k_0)}-\x^{*}\|_2 +  \sum_{k=k_0}^{k_1-1} \delta_{k+1} < r,
\end{equation*}
which contradicts our assumption about $\x_{k_1}$. Thus, must we have $\|\x^{(k)}-\x^{*}\|_2 < r$ for all $k > k_0$. 

Similar to~\eqref{eq:summed_recurrent_inequality_0}, for all $K > k_0$, we have
\begin{equation*}
 \sum_{k=k_0}^{K-1} \|\x^{(k+1)}-\x^{(k)}\|_2 \leqslant \frac{2\sigma}{\mu}  \bigl(\varphi(\Delta_{k_0})-\varphi(\Delta_{K})\bigr) + (\delta_{k_0} - \delta_{K}) \leqslant \frac{2\sigma}{\mu}  \varphi(\Delta_{k_0})+\delta_{k_0} <r.
\end{equation*}
Thus the sequence \(\{\x^{(k)}\}\) is Cauchy and therefore convergent. Since \(\x^*\) is a limit point of \(\{\x^{(k)}\}\), it follows that the entire sequence converges to \(\x^*\).
\end{proof}

Unlike DOCH, providing a convergence guarantee for ADOCH is more challenging. This difficulty arises primarily from the use of the Barzilai-Borwein-type update~\cite{grippo2002} in ADOCH, which can lead to nonmonotonic behaviour in the objective sequence \(\{\mathcal{H}(\x^{(k)})\}\). On the other hand, this scheme allows for more effective exploration of the energy landscape, often leading to better local minima of $\cH$. Empirically, we found that ADOCH gives higher-quality ground states than DOCH, particularly for large Ising models. We wish to clarify that even if we assume that the ADOCH iterates \(\{\x^{(k)}\}, \{\y^{(k)}\}\) are bounded, it is difficult to guarantee convergence of objective values \(\{\cH(\x^k)\}\) to a local minimum of \(\cH\)~\cite{phan2018}.

We now present a result that supports our choice of $\beta$ in the experiments. For sufficiently large $\alpha$, the function $g$ becomes convex (see Proposition~\ref{prop:convexity_fg}), and the boundedness of the iterates \(\{\x^{(k)}\}\) follows directly from the structure of DOCH and the properties of \(\cH\) (Proposition~\ref{prop:bounded_iterates}). However, if $\alpha$ is small, \(\cH\) is no longer a difference of convex functions. The next result tells us that we can still ensure boundedness of the iterates through a careful choice of $\beta$.

\begin{proposition}
\label{prop:bounded_iters_all_alpha}
The DOCH iterates \(\{\x^{(k)}\}\) are bounded if \(\alpha >0\) and
        \begin{equation}\label{eq:l1_norm}
        \beta \geqslant \|\J + \alpha \I \|_{\infty}=\max_ {1 \leqslant i \leqslant n} \Big( \alpha  + \sum_{j \neq i} |  J_{ij}| \Big) .
            \end{equation}
\end{proposition}

\begin{proof}
Starting with an initial point $\x^{(0)} \in \Re^n$, the DOCH iterates are generated using 
\begin{equation}
\label{eq:fpe_supp}
\x^{(k+1)}  = \cT (\x^{(k)}),
\end{equation}
where 
\begin{equation}
\label{eq:FPE_supp}
\cT(\x) = \varphi \big( \beta^{-1}(\J + \alpha \I) \x \big) \qquad \mbox{and} \qquad \varphi(x_1,\ldots,x_n) = (\sqrt[3] x_1,\ldots,\sqrt[3] x_n).
\end{equation}
The matrix norm $\|\cdot  \|_{\infty}$ is induced by the vector norm $\|\z\|_{\infty}= \max_{1 \leqslant i \leqslant n} \, |z_i|$. Hence,
\begin{equation*}
\|\beta^{-1}(\J + \alpha \I) \x \|_{\infty} \leqslant \beta^{-1} \|\J + \alpha \I\|_{\infty} \| \x \|_{\infty}.
\end{equation*}
On the other hand,
\begin{equation*}
\| \varphi(\boldsymbol z) \|_{\infty} = \|\boldsymbol z \|_{\infty}^{1/3}.
\end{equation*}
Combining these, we have from~\eqref{eq:FPE_supp} that
\begin{equation*}
\| \cT(\x) \|_{\infty} \leqslant \big(\beta^{-1} \|\J + \alpha \I\|_{\infty}\big)^{1/3} \|\x \|_{\infty}^{1/3}.
\end{equation*}
In particular, $\| \cT(\x) \|_{\infty} \leqslant    \|\x \|_{\infty}^{1/3}$ if $\beta \geqslant \|\J + \alpha \I \|_{\infty}$. Thus, letting $\x=\x^{(k)}$, we get
\begin{equation*}
\|\x^{(k+1)}\|_{\infty} = \|\cT(\x^{(k)})\|_{\infty} \leqslant  \|\x^{(k)} \|_{\infty}^{1/3}.
\end{equation*}
From this estimate, we can easily verify using induction that \(\|\x^{(k)}\|_\infty \leqslant \max (1, \|\x^{(0)}\|_\infty\)) for all $k \geqslant 1$. This shows that the iterates are bounded.
\end{proof}

We have established that the DOCH iterates converge to a critical point of \(\cH\). Under additional assumptions, this critical point is in fact guaranteed to be a local minimizer.

\begin{proposition}
\label{prop:nonzero_x0_fixed-point_critical _point}
Suppose the limit point of the iterates \(\x^*\) is such that   
\begin{equation}
\label{eq:stablefp}
\|\J_{\cT}(\x^*)\|_2 < 1,
\end{equation}
where $\J_{\cT}$ is the Jacobian (derivative) of $\cT$ and $\|\cdot \|_2$ is the spectral norm (largest singular value). Moreover, suppose that the components of \(\x^*\) are nonzero. Then $\x^*$ is a strict minimizer of \(\cH\).
\end{proposition}

\begin{proof}
To prove that $\x^*$ is a strict minimizer of \(\cH\), it suffices to show that the Hessian $\nabla^2\cH(\x^*) $ is positive definite. It follows from definition~\eqref{eq:fandg} that
\begin{align}
\label{eq:HessH}
\nabla^2\cH(\x^*) = \nabla^2 \! f(\x^*)  - \nabla^2 \! g(\x^*)  = 3\beta\, \mathrm{diag}({\x^*}^2) - (\J + \alpha \I),
\end{align}
where ${\x^*}^2$ denotes componentwise squaring and $\mathrm{diag}(\z)$ denotes a diagonal matrix such that $\mathrm{diag}(\z)_{ii}=z_i$. On the other hand, applying the chain rule of differentiation to~\eqref{eq:FPE_supp}, we have
\begin{equation}
 \label{eq:jacobian}
        \J_{\cT}(\x^*) = \frac{1}{3\beta^{1/3}} \mathrm{diag}\big[(\A\x^*)^{-2/3}\big]\A, \qquad (\A= \J + \alpha\I).
\end{equation}
This is where we need the assumption that the components of \(\x^*\) are nonzero, namely, to ensure that the Jacobian is well-defined at \(\x^*\).
We can simplify~\eqref{eq:jacobian} using the fact that $\x^*$ is a fixed point of $\cT$. Indeed, letting $k \to \infty$ in~\eqref{eq:fpe_supp}, 
we have
\begin{equation*}
 \x^* = \cT(\x^*)= \varphi(\beta^{-1}  \A \x^*),
\end{equation*}
that is, $\A\x^*= \beta \, (\x^*)^3$. Plugging this in~\eqref{eq:jacobian}, we obtain
 \begin{equation}
 \label{eq:jacob}
 \J_{\cT}(\x^*)=\frac{1}{3\beta} \, \mathrm{diag}({\x^*})^{-2} (\J + \alpha \I).
\end{equation}
Comparing~\eqref{eq:HessH} and \eqref{eq:jacob}, we see that the condition $\|\J_{\cT}(\x^*)\|_2 < 1$ implies that $\nabla^2\cH(\x^*)$ is positive definite. This establishes our claim.
\end{proof}

\begin{remarks}
\label{remarks:s1} 
In practice, the fixed points of \(\cT\) have nonzero components, which are rounded to get the spin vector. If these fixed points are stable in the sense of~\eqref{eq:stablefp}, then Proposition~\ref{prop:nonzero_x0_fixed-point_critical _point} ensures that the iterates \(\{\x^{(k)}\}\) converge to a local minimum of the Hamiltonian~(see Supplementary Figure~\ref{fig:DOCH_iterates_convergence}). Empirically, we found that setting $\beta$ sufficiently large, on the order of \(\mathcal{O}\left(n\sqrt{n}\|\J + \alpha \I\|_\infty\right)\), as stipulated by Proposition~\ref{prop:bounded_iters_all_alpha}, keeps the Jacobian norm at the fixed point below one~(see Supplementary Figure~\ref{fig:jacob_norm}). We found that this also yields better quality solutions~(see Supplementary Figure~\ref{fig:opt_beta_value}). In summary, this means that we can set $\beta$ to a high value and tune $\alpha$ freely. 
\end{remarks}

\begin{figure}[htbp]
    \centering
    \includegraphics[width=0.7\textwidth]{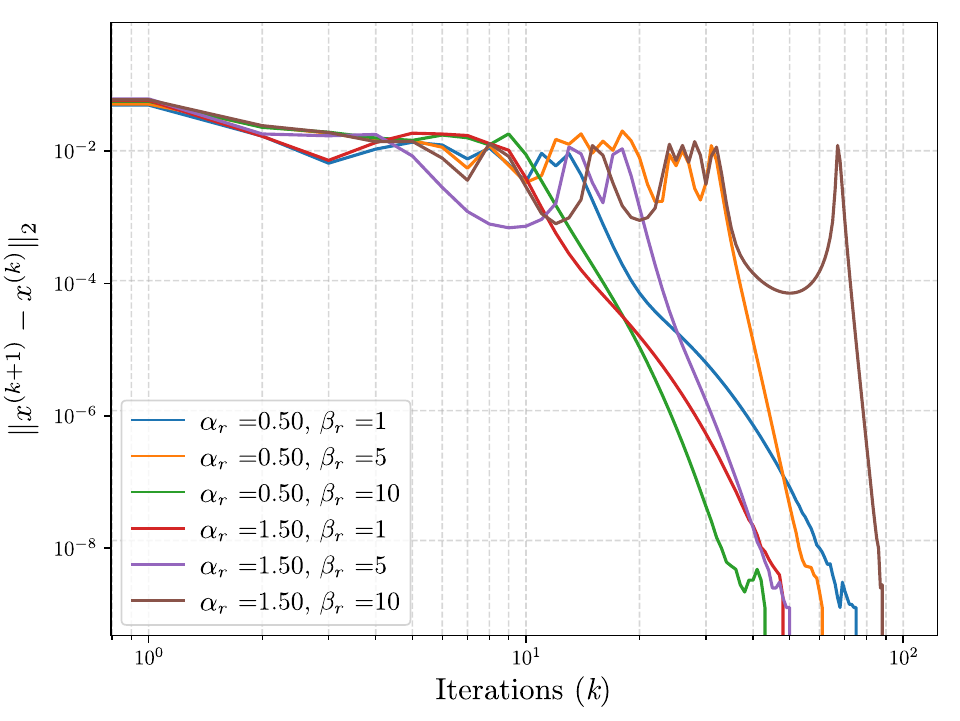}
    \caption{\textbf{Log-log plot of convergence behaviour for the DOCH algorithm.} The plot shows the norm of the difference of successive iterates with iteration count. A $100$-spin SK model is considered. We observe rapid convergence of the iterates \(\{\x^{(k)}\}\) toward a fixed point \(\x^*\) within $100$ iterations, as computed via the fixed-point update rule~\eqref{eq:fpe_supp}. Results are shown for various combinations of hyperparameters: $\alpha_r \in \{0.50, 1.50\}$ and $\beta_r \in \{1, 5, 10\}$.}
    \label{fig:DOCH_iterates_convergence}
\end{figure}

\begin{figure}[htbp]
    \centering
    \includegraphics[width=0.6\textwidth]{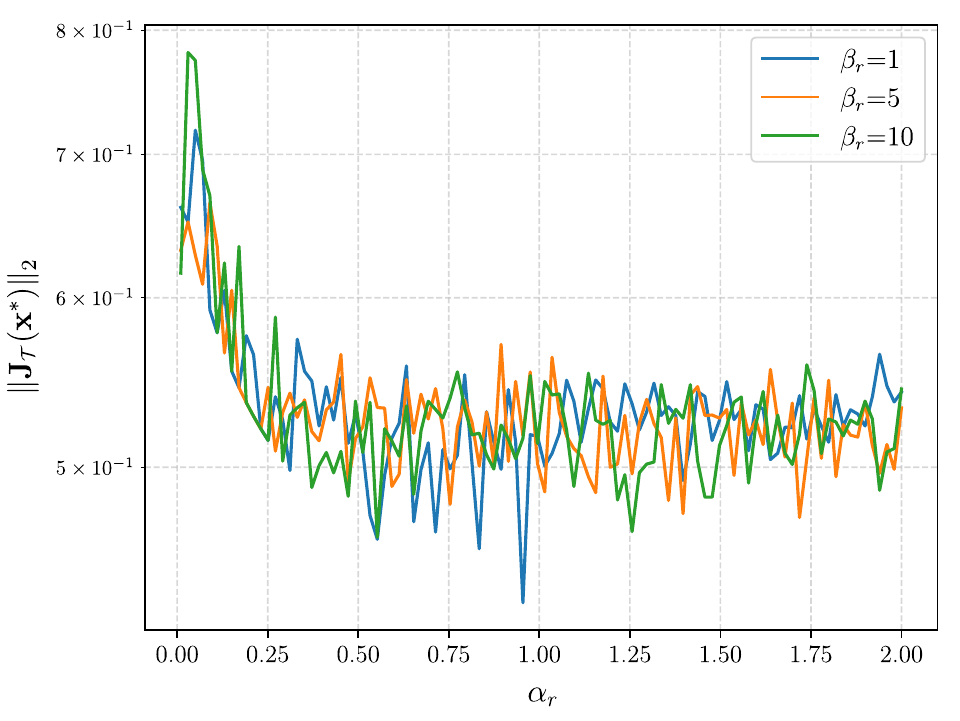}
    \caption{\textbf{Dependence of the (average) norm of the Jacobian \(\|\J_\cT(\x^*)\|_2\) on the ratio $\alpha_r$}. We consider a $100$-spin SK model. The ratios are defined as \(\alpha_r = \alpha / \lambda_{\max}(-\J)\) and \(\beta_r = \beta / (n \sqrt{n} \| \J + \alpha \I \|\infty)\). For each $\alpha_r \in (0, 2)$ and fixed $\beta_r \in \{1, 5, 10\}$, we compute the fixed point \(\x^*\) by applying the update rule in~\eqref{eq:fpe_supp} for $100$ iterations. The spectral norm of the Jacobian \(\|\J_{\cT}(\x^*)\|_2\) is then evaluated at the resulting point. This process is repeated over multiple random initializations \(\x^{(0)}\) and the average Jacobian norms are plotted against $\alpha_r$.}
    \label{fig:jacob_norm}
\end{figure}

\begin{figure}[htbp]
    \centering
    \includegraphics[width=0.6\textwidth]{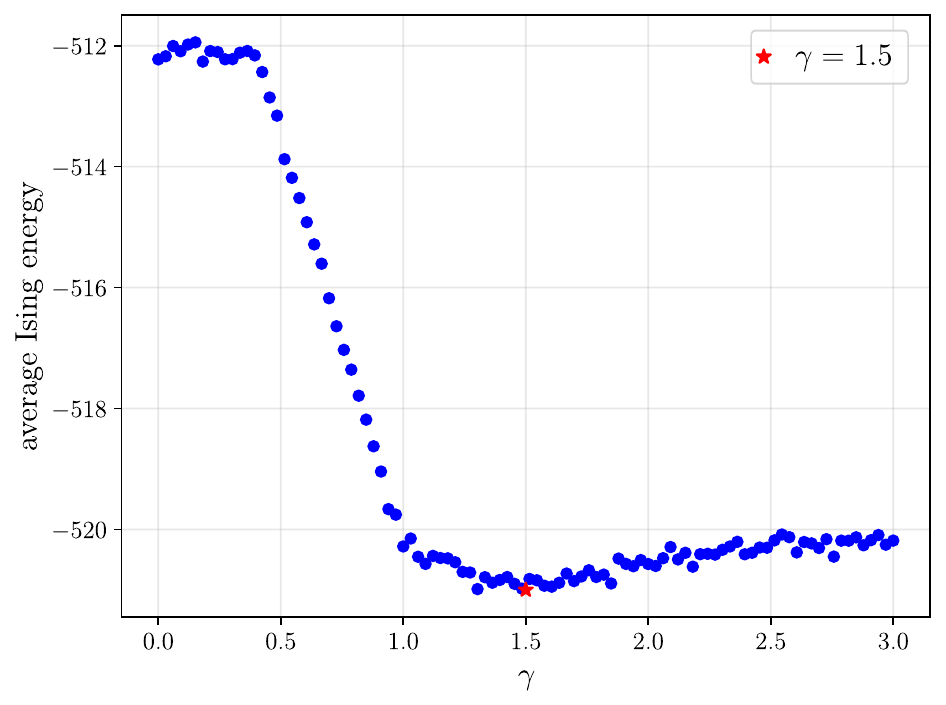}
   \caption{\textbf{Average Ising energy as a function of the exponent $\gamma$}. We evaluate the impact of $\gamma$ in the formula: $\beta = n^{\gamma} \| \J + \alpha \I \|\infty$. This is done for a $100$-spin SK model. For each $\gamma \in (0, 3]$, we set $\beta$ accordingly and run the DOCH solver for $100$ iterations, using the optimal value of $\alpha$. This procedure is repeated over multiple random initializations \(\x^{(0)}\), and the resulting Ising energies are averaged. The average energy is plotted along the y-axis. The results indicate that $\gamma = 1.5$ produces the lowest average Ising energy.}
    \label{fig:opt_beta_value}
\end{figure}

{
\begin{remarks}\label{remarks:s2}
    We can leverage this convergence guarantee to introduce a principled stopping rule. One option is to monitor the relative error
    \begin{equation*}\label{eq:rel_error_iterates}
        r^{(k)} = \frac{\|\x^{(k+1)} - \x^{(k)}\|_2}{\|\x^{(k)}\|_2},
    \end{equation*}
    and stop once $r^{(k)} < \epsilon$, for some tolerance $\epsilon > 0$. We tested this criterion on the K2000 graph (see Supplementary Figure~\ref{fig:termination_criteria}), and found that it accurately detects the saturation point without any loss in performance. 
    The other solvers lack this flexibility, as they lack comparable convergence guarantees and depend critically on cooling schedules.
\end{remarks}}

\begin{figure}[htbp]
        \centering
        \includegraphics[width=0.7\linewidth]{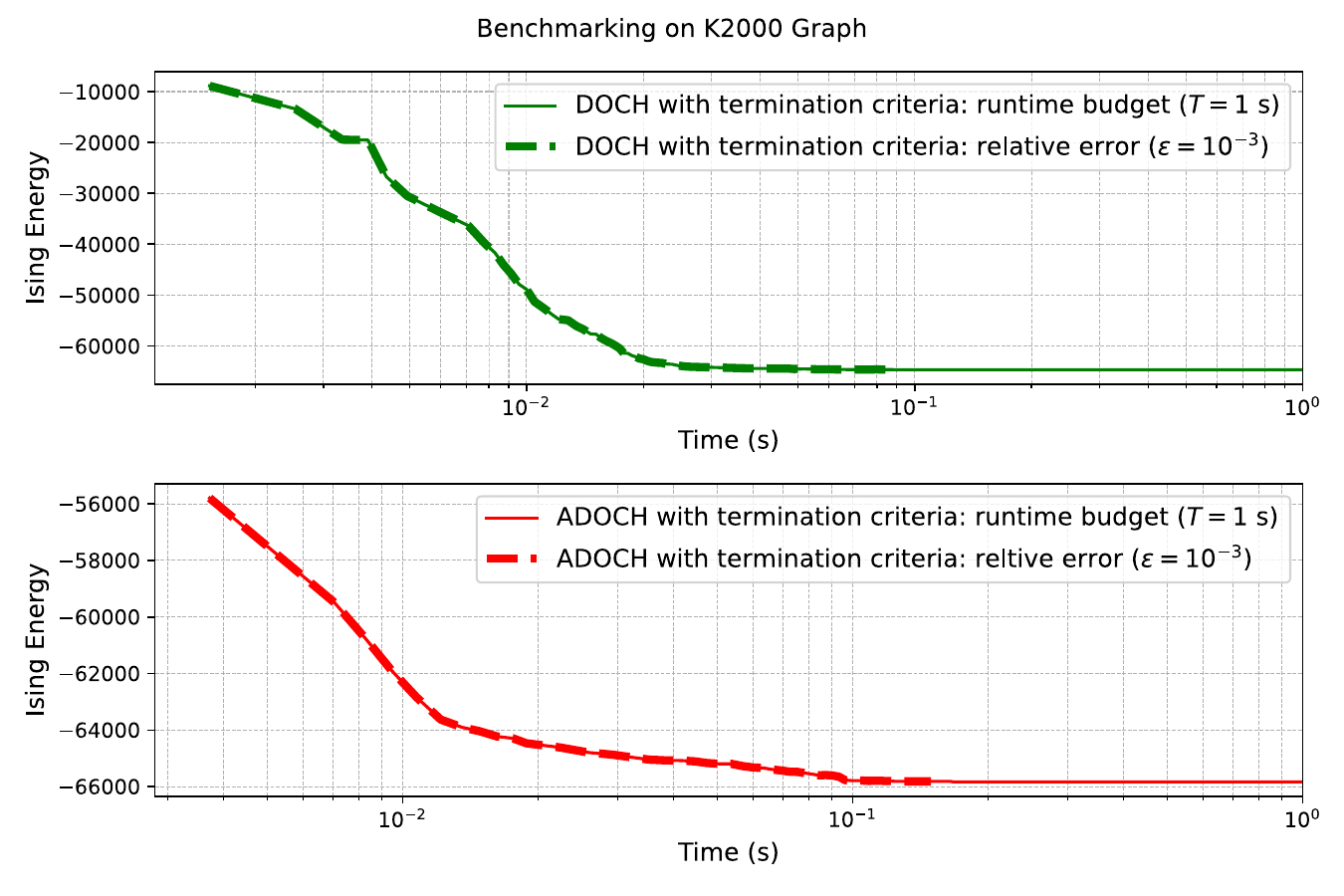} 
        {\caption{\textbf{Comparison of two different termination criteria:} total runtime budget ($T = 1$ s) vs. relative error tolerance ($\epsilon = 10^{-3}$), for DOCH and ADOCH solvers on the K2000 graph. The results were obtained on a laptop equipped with a 4GB NVIDIA RTX 3050 GPU.}
        \label{fig:termination_criteria}}
\end{figure}

\subsection*{Supplementary Note 4: Relation with KPO and OPO}

We show that for certain parameter settings and under some assumptions, the Hamiltonians in~\cite{goto2018, goto2019} resemble our Hamiltonian. The Hamiltonian for the Kerr-nonlinear parametric oscillator (KPO) is of the form
\begin{equation}
\label{eq:s4}
    \cH_c(\x,\y,t) = \sum_{i=1}^{N} \left( \frac{K}{4} (x_i^2 + y_i^2)^2 - \frac{p(t)}{2} (x_i^2 - y_i^2) + \frac{\Delta_i}{2} (x_i^2 + y_i^2) \right) - \frac{\xi_0}{2} \sum_{i=1}^{N} \sum_{j=1}^{N} J_{ij} (x_i x_j + y_i y_j)
\end{equation}
where \(\xi_0\) is a positive constant, \(K\) is the Kerr coefficient, \(p(t)\) is the time-dependent pumping amplitude, \(\Delta_i\) is the detuning frequency of the \(i\)-th oscillator, \(x_i\) is the state of the $i$-th spin, and \(y_i\) is its momentum. On setting \(y_i=0, \Delta_i=\Delta\) in~\eqref{eq:s4}, we get
\begin{align} \label{eq:s6}
    \cH_{KPO}(\x,\y,t) &= \sum_{i=1}^{N} \left( \frac{K}{4} x_i^4 - \frac{p(t)-\Delta_i}{2} x_i^2\right) - \frac{\xi_0}{2} \sum_{i=1}^{N} \sum_{j=1}^{N} J_{ij} x_i x_j  \nonumber \\
  & = \frac{K}{4}(x_1^4+ \dots + x_n^4) - \frac{1}{2}\big(p(t)-\Delta\big)(x_1^2+ \dots + x_n^2) - \frac{\xi_0}{2}\x^\top \J \x.
\end{align}
Thsu, we can identify $\xi_0^{-1}\cH_{KPO}$ with our Hamiltonian $\cH = f - g$, where \(\beta = K\xi_0^{-1}\) and \(\alpha = \xi_0^{-1}(p(t)-\Delta)\) in~\eqref{eq:fandg}. Similarly, the Hamiltonian for the optical parametric oscillators (OPO) can be written as
\begin{equation}\label{eq:s7}
    \cH_{OPO} = \sum_{i=1}^{N} \left( \frac{\kappa_2}{4} (x_i^2 + y_i^2)^2 - \frac{p}{2} (x_i^2 - y_i^2) + \frac{\kappa_1}{2} (x_i^2 + y_i^2) \right) - \frac{\xi_0}{2} \sum_{i=1}^{N} \sum_{j=1}^{N} J_{i,j} (x_i x_j + y_i y_j)
\end{equation}
with $\kappa_1$ and $\kappa_2$ being the one-photon and two-photon loss rates. As in the reduction for the KPO, a scaled version of \(\cH_{\mathrm{OPO}}\) can be identified with our Hamiltonian by setting $\alpha = \xi_0^{-1}(p - \kappa_1)$ and $\beta = \kappa_2 \xi_0^{-1}$ in~\eqref{eq:fandg}. Although the structure of our Hamiltonian resembles \eqref{eq:s6} and \eqref{eq:s7}, it is important to note that the latter cannot generally be expressed as a difference of convex functions.

\subsection*{Supplementary Note 5: MAX-CUT and GW-SDP}

We explain how the MAX-CUT problem on a graph can be reformulated as an Ising model and how the Goemans-Williamson Semidefinite Program (GW-SDP)~\cite{goemans1995} can be used to approximate its ground state.

Given a simple undirected graph $G = (V, E)$, the MAX-CUT problem seeks a partition of the vertex set $V$ into two disjoint subsets $S$ and $T=V\setminus S$ such that the total weight of edges (called the cut value) across $S$ and $T$ is maximized. Specifically, if $\omega_{ij}$ denotes the weight of an edge  $(i,j) \in E,$ then the cut value is 
\begin{equation}
\label{eq:cut}
 \sum_{\substack{(i,j) \in E \\ i \in S,\,  j \in T}}   \omega_{ij}.
\end{equation}
The MAX-CUT problem is to maximize \eqref{eq:cut} with respect to the choice of subsets $S$ and $T$. 

To reformulate this as an Ising problem, we define the weighted adjacency matrix \( \W = \{W_{ij}\} \) given by
\begin{equation*}
W_{ij} =
\begin{cases}
\omega_{ij}, & \text{if } (i,j) \in E, \\
0, & \text{otherwise}.
\end{cases}
\end{equation*}
Note that $W_{ii} = 0$ (no self-loops), and $W_{ij} = W_{ji}$ since the graph is undirected. For any given $S$ and $T$, we can uniquely associate a spin vector $\s \in \{-1,1\}^n$ such that $s_i=1$ if $i \in S$, and $s_i=-1$ if $i \in T$. With this encoding, we can write cut value~\eqref{eq:cut} as
\begin{equation*}
 \frac{1}{4} \sum_{i,j=1}^n \, (1-s_i s_j) \, W_{ij} = \mathrm{constant} - \frac{1}{4} \sum_{i,j=1}^n \, s_i s_j W_{ij}.
\end{equation*}
Consequently, the MAX-CUT problem becomes: 
\begin{equation*}
   \max_{\s \in \{-1,1\}^n} \ \mathrm{constant} - \frac{1}{4} \sum_{i,j=1}^n \, s_i s_j W_{ij} \equiv  \min_{\s \in \{-1,1\}^n} - \frac{1}{2}\, \s^\top \! \J \s.
   \end{equation*}
where \( \J = -(1/2) \W \). This is precisely the Ising problem defined in~\eqref{eq:IsingE}.
 
The GW-SDP is based on the observation that by introducing the matrix variable  \(\X = \s\s^\top, \) we can write the Ising energy as
\begin{equation*}\label{eq:s17}
    \cE(\s) = -\frac{1}{2}\s^\top \J \s
    = -\frac{1}{2} \mathrm{Trace}(\J \s\s^\top)
    = -\frac{1}{2} \mathrm{Trace}(\J \X).
\end{equation*}
The variable \(\X\) is a rank-one, positive semidefinite matrix (\(\X \succeq \0\)) with $X_{ii} = 1$ for $i = 1, \ldots, n$. There exists a one-to-one correspondence between spin assignments in $\{-1, 1\}^n$ and matrices satisfying these properties. Consequently, the Ising problem can be reformulated as:
\begin{equation*}
    \min_{\substack{\substack{\mathrm{rank}(\X) = 1 \\ \X\succeq \0}\\ \X_{ii}=1}}\ \  - \frac{1}{2} \mathrm{Trace}(\J \X)
\end{equation*}
By dropping the nonconvex rank constraint, we obtain a convex optimization problem known as a semidefinite program:
\begin{equation}
\label{eq:sdp}
    \underset{\substack{\X\succeq \0 \\ \X_{ii}=1}}  {\max} \quad \mathrm{Trace}(\J \, \X)
\end{equation}
We can solve~\eqref{eq:sdp} in polynomial time using standard interior-point methods. However, the computational complexity for achieving a given accuracy is relatively high, namely $\mathcal{O}(n^{4.5})$ for an n-spin Ising model~\cite{Luo2010}. In practice, this limits scalability to about $10^3$ spins. It was famously shown in~\cite{goemans1995} that by using a randomized rounding scheme to extract a binary spin  \( \s \in \{-1,1\}^n \) from the optimal solution of~\eqref{eq:sdp}, the expected cut value is at least $87.8\%$ of the optimal value of~\eqref{eq:cut}.

\subsection*{Supplementary Note 6: Parameter Setting}

We describe how we set the parameters for our Ising solvers DOCH and ADOCH. We know from Proposition~\ref{prop:convexity_fg} that the function $g$ in~\eqref{eq:fandg} is convex whenever \(\alpha \geq \lambda_{\max}(-\J)\). Following this, we set \(\alpha = \eta \lambda_{\max}(-\J)\), and the parameter $\eta$ is optimized within the interval $(0, 2]$  (as per Remark~\ref{remarks:s1}) depending on the structure of the coupling matrix \(\J\). Since our solver converges quickly, $\eta$ can be tuned by examining the Ising energy over the first few iterations. In practice, we run $5-10$ iterations for large graphs and $1-3$ iterations for ultra-large-scale graphs to determine a suitable $\eta$. 

For small matrices, the largest eigenvalue \(\lambda_{\max}(\J)\) is computed using the power method~\cite{Horn_Johnson_1985}. For \(n \geq 10^4\), we approximate \(\lambda_{\max}(\J)\) using Wigner’s semicircle law~\cite{mehta2004rmt}: \(\lambda_{\max}(\J) \approx 2\langle \J \rangle \sqrt{n}\), where \(\langle \J \rangle\) is the empirical variance of the entries of \(\J\) (see main text).

Following Proposition~\ref{prop:bounded_iters_all_alpha}, the parameter $\beta$ is set as \(\beta= n\sqrt{n}\| \J+\alpha \I\|_{\infty}\), which ensures that the DOCH iterates remain bounded. This choice also promotes convergence to a strict minimum of the Hamiltonian~(see Remark~\ref{remarks:s1}). 

The look-back parameter $q$ in ADOCH is assigned a value between 5 and 10, depending on the structure of the Ising model. For larger Ising models \(n \geqslant 10^4\), we take $q \leq 5$.

\subsection*{Supplementary Note 7: Existing Ising Solvers}

\subsubsection*{Simulated Annealing (SA)}

We implemented the Simulated Annealing (SA) algorithm following the approach in~\cite{inagaki2016}. The pseudocode is provided in \textbf{Algorithm S1}. We used a logarithmic cooling schedule $\beta(t) = \beta_0 \log\, (1 +t/T)$, where $\beta_0$ is the initial inverse temperature parameter and $T$ is the total number of iterations (or runtime). SA has two tunable parameters, $\beta_0$ and $T$. For the $K_{2000}$ Ising model, the ratio $T/\beta_0 \approx 1000$ is recommended in~\cite{inagaki2016}; accordingly, we use $\beta_0 = 1$ and $T = 1000$. For other cases, $\beta_0$ is selected from the interval $[1, 2]$ to obtain the best performance.

\begin{center}\label{algo:sa}
    \begin{tabular}{lll}
        \hline
        \multicolumn{2}{l}{\textbf{Algorithm S1:} Simulated Annealing} & \\
        \hline
        1: & \textbf{initialize:} spin state $\s$, maximum iterations $T$ & \\
        2: & $n \leftarrow \text{dim}(\J)$ & \\
        3: & $E \leftarrow \cE(\s)$ & \\
        4: & \textbf{for} $t = 0$ to $T$ \textbf{do} & \\
        5: & \hspace{1em} $\beta \leftarrow \beta_0 \log(1 + t/T)$ & $\triangleright$ cooling rate \\
        6: & \hspace{1em} randomly pick spin $i \in \{1,2,\dots n\}$ & \\
        7: & \hspace{1em} $\s' \leftarrow \text{flip}(\s, i)$ & $\triangleright$ flip \(i\)-th spin $(s_i\to -s_i)$\\
        8: & \hspace{1em} $\Delta E \leftarrow \cE(\s') - E$ & \\
        9: & \hspace{1em} uniformly sample $z$ from (0, 1) \\
        9: & \hspace{1em} \textbf{if} $\Delta E < 0$ \textbf{or} $\exp(-\beta\Delta E) \geqslant z$ & \\
        10: & \hspace{2em} $\s \leftarrow \s'$ & \\
        11: & \hspace{2em} $E \leftarrow E + \Delta E$ & \\
        12: & \hspace{1em} \textbf{end if} & \\
        13: & \textbf{end for} & \\
        14: & \textbf{return:} $\s$ & \\
        \hline
    \end{tabular}
\end{center}

\subsubsection*{ballistic Bifurcation Machine (bSB)}

We implemented the ballistic Bifurcation Machine following the approach described in~\cite{goto2021}. The pseudocode is provided in \textbf{Algorithm S2}. Following the setup in~\cite{goto2021}, the pump rate is set to $a_0 = 1$, so that $a_t$ increases linearly from 0 to 1 according to the prescribed pump rate schedule. The time step $\Delta_t$ is selected from the set \{0.25, 0.5, 0.75, 1, 1.25\}, based on empirical performance. The coupling strength $c_0$ is set using the formula:
\begin{equation}\label{eq:s29}
c_0 = \frac{1}{2\langle \J \rangle\sqrt{n}},
\end{equation}
where \(\langle \J \rangle\) denotes the sample variance of the entries of the matrix \(\J\). For instance, for the $K_{2000}$ model with $J_{ij} = \pm 1$ and $n=2000$, this yields $c_0 \approx 0.0112$. For other problem instances (e.g., SK, G-Set, Biq Mac), we compute $c_0$ explicitly based on the corresponding \(\J\) matrix for each graph.

\begin{center}\label{algo:bsb}
    \begin{tabular}{lll}
        \hline
        \multicolumn{3}{l}{\textbf{Algorithm S2:} ballistic Simulated Bifurcation Machine} \\
        \hline
        1: & \multicolumn{2}{l}{\textbf{input:} pump rate $a_0$, coupling strength $c_0$, time step $\Delta_t$, total steps $T$} \\
        2: & $n \leftarrow \text{dim}(\J)$ & \\
        3: & \textbf{initialize:} $\x \leftarrow 2 \cdot \text{Bernoulli}(0.5, n) - 1$ & $\triangleright$ random $\pm1$ spin vector \\
        4: & \textbf{initialize:} $\y \leftarrow \0_n$ & $\triangleright$ zero vector of size $n$ \\
        5: & \textbf{for} $t = 0$ to $T$ \textbf{do} & \\
        6: & \hspace{1em} $a_t\leftarrow (a_0t/T)$ & $\triangleright$ pump rate schedule \\
        7: & \hspace{1em} $\y \leftarrow \y + [-(a_0 - a_t)\x + c_0\J\x]\Delta_t$ & $\triangleright$ update momentum \\
        8: & \hspace{1em} $\x \leftarrow \x + a_0\y\Delta_t$ & $\triangleright$ update position \\
        9: & \hspace{1em} $\x \leftarrow \text{clip}(\x, -1, 1)$ & $\triangleright$ boundary conditions \\
        10: & \hspace{1em} $y_i \leftarrow 0$ if $x_i = \pm1$ & $\triangleright$ momentum at the boundaries \\
        12: & \textbf{end for} & \\
        13: & \textbf{return:} $\s = \mathrm{sign}(\x)$ & \\
        \hline
    \end{tabular}
\end{center}

\subsubsection*{Simulated Coherent Ising Machine (SimCIM)}

{We implemented the Simulated Coherent Ising Machine following~\cite{goto2021}. The pseudocode is provided below, where $\mathcal{N}[\mathbf{0}, \mathbf{I}]$ denotes the standard $n$-dimensional Gaussian distribution. We adopt the same parameter settings from~\cite{goto2021} 
for the pump rate $a_0$, the coupling strength $c_0$, and the time step $\Delta_t$}. The noise amplitude $A$ is chosen from the interval $[0.1, 1]$, with values in this range empirically yielding the best performance.

\begin{center}\label{algo:simcim}
    \begin{tabular}{lll}
        \hline
        \multicolumn{3}{l}{\textbf{Algorithm S3:} Simulated Coherent Ising Machine} \\
        \hline
        1: & \multicolumn{2}{l}{\textbf{input:} noise amplitude $A$, initial rate $a_0$, coupling $c_0$, time step $\Delta_t$, steps $T$} \\
        2: & $n \leftarrow \text{dim}(\J)$ & \\
        3: & \textbf{initialize:} $\x \leftarrow 2 \cdot \text{Bernoulli}(0.5, n) - 1$ & $\triangleright$ random $\pm1$ vector \\
        4: & \textbf{for} $t = 0$ to $T$ \textbf{do} & \\
        5: & \hspace{1em} ${a}_t \leftarrow (a_0t/T)$ & $\triangleright$ pump rate schedule \\
        6: & \hspace{1em} $\w \sim \cN[0,\I]$ & $\triangleright$ noise vector \\
        7: & \hspace{1em} $\x \leftarrow \x + \left(-(a_0 - a_t)\x + c_0\J \, \mathrm{sign}(\x) \right)\Delta_t + A\w\sqrt{\Delta_t}$ & \\
        8: & \hspace{1em} $\x \leftarrow \text{clip}(\x, -1, 1)$ & $\triangleright$ boundary conditions \\
        10: & \textbf{end for} & \\
        11: & \textbf{return:} $\s = \mathrm{sign}(\x)$ & \\
        \hline
    \end{tabular}
\end{center}

{
The original Coherent Ising Machine (CIM)~\cite{inagaki2016} is based on a degenerate optical parametric oscillator (DOPO) and is modeled by the following system of stochastic differential equations~\cite{inagaki2016,wang2013coherent}:
\begin{equation}\label{eq:CIM_s1}
dc_i = 
\bigg[ (-1 + p - c_i^2 - s_i^2)c_i + r\sum_j J_{ij} \tilde{c}_j \bigg] dt
+ \frac{1}{A_s} \sqrt{c_i^2 + s_i^2 + 0.5}\, dW_1 , 
\end{equation}

\begin{equation}\label{eq:CIM_s2}
ds_i = 
\left[ (-1 - p - c_i^2 - s_i^2)s_i \right] dt
+ \frac{1}{A_s} \sqrt{c_i^2 + s_i^2 + 0.5}\, dW_2 ,
\end{equation}

\noindent
where 
\begin{itemize}
\item $c_i$ and $s_i$ are the in-phase and quadrature amplitudes of the $i$-th DOPO,
\item $\tilde{c}_i$ is the measured in-phase amplitude (perturbed by measurement noise),
\item $\{J_{ij}\}$ are the the coupling coefficients with coupling strength $r$,
\item $p$ is the normalized pump rate,
\item $A_s$ is the saturation amplitude,
\item $dW_1$ and $dW_2$ are independent Gaussian Wiener processes representing vacuum and pump fluctuations for in-phase and quadrature components, respectively.
\end{itemize}

In Supplementary Figure~\ref{fig:CIM_s1_s2_SimCIM}, we simulate the CIM dynamics for the MAX-CUT problem on the K2000 graph. To numerically solve \eqref{eq:CIM_s1} and \eqref{eq:CIM_s2}, we apply the Euler-Maruyama integrator~\cite{maruyama1955euler}, yielding the updates

\begin{align}\label{eq:simulate_CIM_s1}
c_i(t+1) = c_i(t) + \bigg[(-1 + p(t) - c_i(t)^2 - s_i(t)^2)c_i(t) \,+\, &r \sum_j J_{ij}\widetilde{c}_j(t)\bigg]\Delta_t \nonumber \\
& + \frac{w_i \sqrt{\Delta_t}}{A_s}\sqrt{c_i(t)^2 + s_i(t)^2 + 0.5},
\end{align}
and
\begin{align}\label{eq:simulate_CIM_s2}
s_i(t+1) &= s_i(t) + \Big[\left(-1 - p(t) - c_i(t)^2 - s_i(t)^2\right)s_i(t)\Big]\Delta_t
+ \frac{w_i' \sqrt{\Delta_t}}{A_s}\sqrt{c_i(t)^2 + s_i(t)^2 + 0.5}.
\end{align}

We have used $\widetilde{c}_i = c_i$ and  $\widetilde{c}_i = \mathrm{sign}(c_i)$ (clamped) in the experiments. The pump rate is defined as $p(t) = t/T$, and $w_i, w_i'$ are sampled independently from the standard normal distributions. After runtime $T$, the $i$-th spin state is obtained by taking the $\mathrm{sign}$ of $c_i(T)$. The hyperparameters $r, A_s, \Delta_t$ are tuned for the best results. In Supplementary Figure~\ref{fig:CIM_s1_s2_SimCIM}, we observe that \textbf{Algorithm S3} for SimCIM (adapted from \cite{goto2021}) outperforms the simulated version of CIM in~\cite{inagaki2016} based on \eqref{eq:CIM_s1} and \eqref{eq:CIM_s2}.}

\begin{figure}[htbp]
\centering
\includegraphics[width=1.0\linewidth]{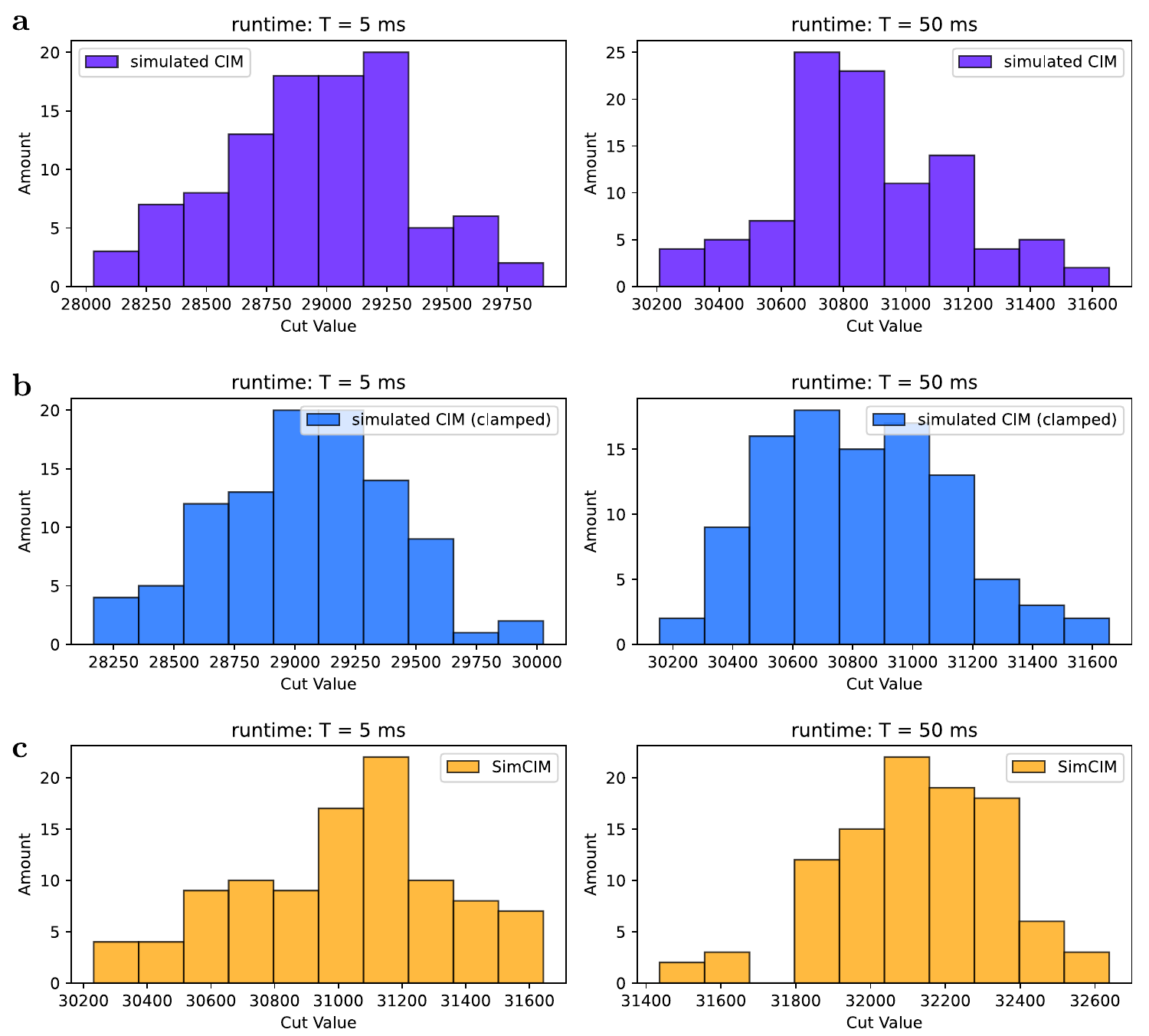}
{
\caption{\textbf{Comparison of simulated CIM~\cite{inagaki2016} with SimCIM.} Histogram plots of cut values obtained for K2000 graph, with $100$ different initializations and two runtimes ($T=5$ ms, $50$ ms). \textbf{(a), (b)} Cut values obtained by two variants of simulated CIM (from equations \eqref{eq:simulate_CIM_s1}, \eqref{eq:simulate_CIM_s2}). \textbf{(c)} Cut values obtained by SimCIM (from Algorithm S2). The results were obtained on a laptop equipped with a 4GB NVIDIA RTX 3050 GPU.}
\label{fig:CIM_s1_s2_SimCIM}}
\end{figure}

\subsubsection*{Spring Ising Algorithm (SIA)}

We implemented the Spring Ising Algorithm (SIA) as described in~\cite{jiang2024}. The corresponding pseudocode is provided in \textbf{Algorithm S4}. Following~\cite{jiang2024}, we set the mass coefficient $m = 1$, the elastic coefficient $k = 0.5$, and the scaling coefficient $\zeta(t)$ to vary linearly from $0.8\zeta_0$ to $10\zeta_0$, where $\zeta_0 = 0.05$ is the base value. The time step $\Delta$ is selected within the interval $(0, 1]$ for the best result. The momentum vector \(\p\) is initialized such that each component $p_i$ is sampled uniformly from the range $(-0.0005, 0.0005)$.

\begin{center}\label{algo:sia}
    \begin{tabular}{lll}
        \hline
        \multicolumn{3}{l}{\textbf{Algorithm S4:} Spring Ising Algorithm} \\
        \hline
        1: &\multicolumn{2}{l}{\textbf{input:} mass coefficient $m$, elastic coefficient $k$, scaling coefficients $\zeta(t)$, $\Delta$, total steps $N$}\\
        2: & $n \leftarrow \text{dim}(\J)$ & \\
        3: & $\q \leftarrow \0$ & $\triangleright$ zero vector of size $n$\\
        4: & $\p \leftarrow \text{Random}(n)$ & $\triangleright$ small random perturbation in $(-0.0005,0.0005)$ \\
        5: & \textbf{for} $t = 1$ to $N$ \textbf{do} & \\
        6: & \hspace{1em} $(\q, \p) \leftarrow \text{Boundary}(\q,\p)$ & $\triangleright$ boundary conditions (as per \eqref{eq:SIA_pq_update}) \\
        7: & \hspace{1em} $\q \leftarrow \q + (\Delta/m)\p$ & \\
        8: & \hspace{1em} $\p \leftarrow \p - \Delta k\q + \zeta(t)\Delta\J \q$ & \\
        9: & \hspace{1em} \textbf{end if} & \\
        10: & \textbf{end for} & \\
        11: & \textbf{return:} $\s = \mathrm{sign}(\x)$ & \\
        \hline
    \end{tabular}
\end{center}

The function \(\mathrm{Boundary}(\q, \p)\) enforces the following componentwise constraints on the vectors \(\q\) and \(\p\):

\begin{equation}\label{eq:SIA_pq_update}
    q_i \leftarrow \begin{cases}
        \sqrt{2}, & \quad q_i > \sqrt{2},\\ 
        q_i, &  \quad -\sqrt{2} \leq q_i \leq \sqrt{2}, \\
         -\sqrt{2}, &  \quad q_i < -\sqrt{2},
    \end{cases}
    \hspace{4em}
    p_i \leftarrow \begin{cases}
        2, &   \quad p_i > 2,\\ 
        p_i, &  \quad  -2 \leq p_i \leq 2, \\ 
        -2, &   \quad p_i < -2.
    \end{cases}
\end{equation}

\subsection*{Supplementary Note 8: Benchmarking Information}

The table below lists the Ising models used for benchmarking. For each model, we specify the number of spins $n$, the connectivity level $(1-p)\%$ where $p$ represents the sparsity, the distribution of the coupling coefficients $J_{ij}$, and the GPU hardware used for computation, including their memory specifications.


\begin{table}[htbp]
    \[
    \begin{array}{|c|c|c|c|c|}
        \hline
        \text{Fig.}& n & \text{connectivity} & J_{ij} & \text{GPU} \\
        \hline \hline
        \ref{fig:supp_more_fem_1} & 10^3 & \text{fully connected} & \text{SK model} & \text{NVIDIA jetson nano}\\
        \ref{fig:supp_more_fem_2} & 10^4 & \text{fully connected} & \text{SK model} & \text{NVIDIA jetson nano}\\
        \ref{fig:supp_more_3} & 10^4 & 1\% & \{-2^9+1, \dots, 2^9-1\} & \text{NVIDIA jetson nano}\\
        \hline
        \ref{fig:supp_more_4} & 10^5 & \text{fully connected} & \{-1, 1\} & \text{$1 \times$ RTX  3090, 24 GB}\\
        \ref{fig:supp_more_5} & 10^5 & \text{fully connected} & \text{SK model} & \text{$1 \times$ RTX 3090, 24 GB}\\
        \ref{fig:supp_more_6} & 10^5 & 1\% & \{-2^9+1, \dots, 2^9-1\} & \text{$1 \times$ RTX 3090, 24 GB}\\
        \hline
        \ref{fig:supp_more_7} & 10^6 & \text{fully connected} & \sin(i\, j + \mathrm{seed}) & \text{$1\times$ V100, 32 GB}\\
        \ref{fig:supp_more_8} & 10^6 & 0.1\% & \{-2^9+1, \dots, 2^9-1\}  & \text{$1 \times$ V100, 32 GB}\\
        \ref{fig:supp_more_9} & 10^6 & 0.01\% & \{-2^9+1, \dots, 2^9-1\}  & \text{$1 \times$ V100, 32 GB}\\
        \hline
        \ref{fig:supp_more_10} & 10^7 & 0.001\% & \{-2^9+1, \dots, 2^9-1\}  & \text{$4 \times$ V100, 30 GB}\\
        \hline
        \ref{fig:supp_more_11} & 10^8 & 0.00001\% & \{-2^9+1, \dots, 2^9-1\}  & \text{$2 \times$ H100, 80 GB}\\
        \hline

        {\ref{fig:error_bar_plot_gset}} & {800} & {\text{variable}} & {\text{G-set graphs}}  & {\text{$1 \times$ RTX 3050, 4 GB}}\\
        \hline
        {\ref{fig:error_bar_plots}} & {10^4-10^5} & {\text{variable}} & {\text{variable}}  & {\text{$1 \times$ RTX 3090, 24 GB}}\\
        \hline
        
    \end{array}
    \]
    \caption{Ising model parameters used for benchmarking.}
    \label{table:benchmarking_graphs}
\end{table}

\subsection*{Supplementary Note 9: Data Generation}

To benchmark our algorithm, we construct large to ultra-large coupling matrices \(\J\), in both dense and sparse forms. A sparse \(\J\) has approximately $(1-p)\%$ nonzero entries. We generate only the lower triangular part of the matrix \(\J\), excluding the diagonal (which has zeros), and then symmetrize it by mirroring the values to the upper triangle. To sample $p$-sparse matrix $\J$ with $J_{ij}$ from $9$-bit signed integers, we sample $z \in \{1, \dots, N_p\}$, with $N_p = \lfloor 102300/p \rfloor$; if $z < 1023$, we set $J_{ij} = J_{ji} = z - 511$ (described in detail in \textbf{Algorithm S5}). To manage large matrices ($n \geq 10^5$), we use the Compressed Sparse Row (CSR) format~\cite{CSR2003book}, which stores only nonzero values, column indices, and row offsets. This enables efficient storage and access. Using CSR, we generate sparse matrices with $n = 10^5$ to $10^8$, and connectivity levels from $10^{-6}$ to $10^{-1}$, yielding billions of nonzeros with a fraction of the memory required by dense formats.

For ultra-large dense coupling matrices, explicit storage is prohibitively expensive due to the quadratic memory requirement. To address this, we implement an efficient, memory-free approach based on procedural generation. Instead of considering the full $n\times n$ matrix, we use a deterministic pseudo-random function to generate matrix elements on demand. Specifically, each coupling coefficient is generated using the rule:
$J_{ij} = \sin(i \, j + \mathrm{seed})$ with $\mathrm{seed} = 100$. This function-based representation ensures that the matrix remains symmetric and reproducible while eliminating the need to store it in memory. As a result, we can work with dense matrices of size exceeding $10^5 \times 10^5$, which would otherwise require terabytes of memory if stored explicitly.

$$
\begin{array}{l}
    \hline
    \text{\textbf{Algorithm S5:} Generation of ultra-large coupling matrix in CSR format} \\
    \hline
    1:\quad\text{\textbf{input:} } n  \geqslant 1, p \in (0, 100] \\
    2:\quad N_p \leftarrow \lfloor 102300/p \rfloor \\
    3:\quad \text{\textbf{initialize arrays:} data} \leftarrow [\;],\; \text{column\_indices} \leftarrow [\;],\; \text{row\_offset} \text{ of length } (n+1) \\
    4:\quad \text{row\_offset}[1] \leftarrow 0 \\
    5:\quad\text{\textbf{for} $i = 1$ \textbf{to} $n$ \textbf{do}} \\
    6:\quad\quad \text{non\_zero\_in\_row} \leftarrow 0 \\
    7:\quad\quad \text{\textbf{for} $j = 1$ \textbf{to} $i-1$ \textbf{do}} \\
    8:\quad\quad\quad z \leftarrow \text{RandomInt}(1, \dots, N_p) \\
    9:\quad\quad\quad \text{\textbf{if} $z < 1023$ \textbf{then}} \\
    10:\quad\quad\quad\quad \text{data.store}(z - 511) \\
    11:\quad\quad\quad\quad \text{column\_indices.store}(j) \\
    12:\quad\quad\quad\quad \text{non\_zero\_in\_row} \leftarrow \text{non\_zero\_in\_row} + 1 \\
    13:\quad\quad\quad \text{\textbf{end if}} \\
    14:\quad\quad \text{\textbf{end for}} \\
    15:\quad\quad \text{row\_offset}[i+1] \leftarrow \text{row\_offset}[i] + \text{non\_zero\_in\_row} \\
    16:\quad\text{\textbf{end for}} \\
    17:\quad \J_{\mathrm{upper}} \leftarrow \text{CSR}(\text{data}, \text{column\_indices}, \text{row\_offset}) \\
    18:\quad \J \leftarrow \J_{\mathrm{upper}} + {\J_{\mathrm{upper}}}^\top \\
    19:\quad \text{\textbf{return }} \J \\
    \hline
\end{array}
$$

\subsection*{Supplementary Note 10: Large Matrix-Vector Multiplication}

For ultra-large matrices, matrix-vector multiplications are expensive and memory-intensive. To address this, we adopt a block-based computation strategy that improves memory efficiency and parallel performance. The matrix \(\J\) is partitioned into smaller $b \times b$ blocks. Each block is generated on demand, used immediately for partial matrix-vector product computation, and then discarded to conserve memory. The computation is distributed across multiple GPUs (typically 2 to 4), with workload allocated to minimize idle time and ensure balanced utilization. Each GPU processes its assigned matrix blocks independently and in parallel, significantly accelerating the overall computation while maintaining scalability for extremely large problems.

$$
\begin{array}{l}
    \hline
    \text{\textbf{Algorithm S6:} Large Matrix-Vector Multiplication} \\
    \hline
    1:\quad\text{\textbf{input:} $n \times n$ matrix $\J$, $n\times 1$ vector $\vv$, block size $b \times b$, number of GPUs $G$} \\
    2:\quad\text{\textbf{output:} $\y = \J v$} \\
    3:\quad\text{\textbf{initialize:} $\y \leftarrow (0)^n$} \\
    4:\quad\text{// distribute the blocks across $G$ GPUs to balance workload} \\
    5:\quad\text{\textbf{parallel for} $g = 1$ \textbf{to} $G$ \textbf{do}} \\
    6:\quad\quad \text{\textbf{for} $(i,j)$-th assigned block on GPU $g$} \\
    7:\quad\quad \quad \text{// dynamically generate block $J_{i:i+b, j:j+b}$} \\
    8:\quad\quad \quad \J_{\text{block}} = \J[i:i+b, j:j+b] \\
    9:\quad\quad \quad \vv_{\text{sub}} \leftarrow \vv[j:j+b] \\
    10:\quad\quad \quad \y[i:i+b] \leftarrow \J_{\text{block}} \cdot v_{\text{sub}} \\
    11:\quad\quad \quad \text{// discard $J_{\text{block}}$ to free memory} \\
    12:\quad\quad \text{\textbf{end}} \\
    13:\quad\text{\textbf{end parallel}} \\
    14:\quad \text{\textbf{return }} \y \\
    \hline
    \label{algo:block_mat_mul}
\end{array}
$$

\bibliography{references}
\bibliographystyle{bib_style_science}

\end{document}